\documentclass[11pt]{article}
\usepackage{geometry}
\usepackage{amsmath,amsfonts,amssymb,amsthm}
\usepackage{graphicx}
\usepackage{enumerate}
\usepackage{bbm}
\usepackage{tikz}
\usetikzlibrary{math}
\usepackage{verbatim}
\usepackage{hyperref,color}
\hypersetup{colorlinks=true,citecolor=blue, linkcolor=blue, urlcolor=blue}
\usepackage[linesnumbered,ruled,vlined]{algorithm2e}
\usepackage{textcomp}

\newtheorem{thm}{Theorem}
\numberwithin{thm}{section}

\newtheorem{defn}[thm]{Definition}
\newtheorem{lemma}[thm]{Lemma}

\newtheorem{rmk}[thm]{Remark}

\newtheorem{fact}[thm]{Fact}
\newtheorem{claim}[thm]{Claim}

\newcommand{\poly}{\mathrm{poly}}
\newcommand{\e}{\mathrm{e}}

\newcommand{\E}{\mathrm{E}}
\newcommand{\Var}{\mathrm{Var}}

\newcommand{\Ind}{\mathbbm{1}}

\newcommand{\norm}[1]{\left\lVert #1 \right\rVert}

\newcommand{\TV}[2]{\left\| #1 - #2 \right\|_{\textsc{tv}}}
\newcommand{\ol}[1]{\hat{#1}}
\newcommand{\ceil}[1]{\left\lceil #1 \right\rceil}
\newcommand{\floor}[1]{\left\lfloor #1 \right\rfloor}

\newcommand{\BIS}{\textsc{\#BIS}}
\newcommand{\BIShard}{{\textsc{\#BIS-hard }}}
\newcommand{\BIPthreeCOLOR}{\textsc{\#BIP-3-COL}}

\def\eps{\varepsilon}
\def\epsilon{\varepsilon}

\newcommand{\N}{\mathbb{N}}
\newcommand{\R}{\mathbb{R}}
\newcommand{\Ham}{\mathcal{H}}
\newcommand{\Gr}{\mathcal{G}}
\newcommand{\muemp}{\mu_{\textsc{emp}}}
\newcommand{\muemph}{\mu^*_{\textsc{emp}}}

\def\MC{\textsc{MaxCut}}

\newcommand{\UMC}{\textsc{TwoLargeCuts}}
\newcommand{\din}{d_{\textsc{in}}}
\newcommand{\dout}{d_{\textsc{out}}}

\newcommand{\Gu}{G^*}
\newcommand{\bu}{\beta^*}
\newcommand{\Hu}{\Ham^*}

\newcommand{\Mhat}{M^*}

\title{Lower bounds for testing graphical models: colorings and antiferromagnetic Ising models\thanks{An extended abstract of this paper appeared in the Proceedings of COLT 2019~\cite{BBCSV}.}}

\author{
Ivona Bez\'akov\'a\thanks{Rochester Institute of Technology. Email: ib@cs.rit.edu. Research supported in part by NSF grant 1819546.}
\and
Antonio Blanca\thanks{Pennsylvania State University.  Email: ablanca@cse.psu.edu. Research supported in part by NSF grants CCF-1617306 and CCF-1563838.}
\and
Zongchen Chen\thanks{Georgia Institute of Technology.  Email: \{chenzongchen,vigoda\}@gatech.edu.
Research supported in part by NSF grants CCF-1617306 and CCF-1563838.}
\and
Daniel \v{S}tefankovi\v{c}\thanks{University of Rochester.
 Email: stefanko@cs.rochester.edu. Research
supported in part by NSF grant CCF-1563757.}
\and
Eric Vigoda$^\mathsection$
}

\begin{document}
	
	\maketitle

\begin{abstract}
	We study the identity testing problem in the context of spin systems
	or undirected graphical models,
	where it takes the following form:
	given the parameter specification of the model $M$
	and a sampling oracle for the  distribution $\mu_{\Mhat}$ of an unknown model $\Mhat$,
	can we efficiently determine if the two models $M$ and $\Mhat$ are the same?
	We consider identity testing for both soft-constraint and hard-constraint systems.
	In particular, we prove hardness results in two prototypical cases, the \emph{Ising model} and \emph{proper colorings},
	and explore whether identity testing is any easier than structure learning.

	For the ferromagnetic
	(attractive) Ising model,
	Daskalakis et al.\ (2018) presented a polynomial time algorithm
	for identity testing.
	We prove hardness results in the antiferromagnetic (repulsive) setting in the same regime of parameters where structure learning
	is known to
	require a super-polynomial number of samples.
	Specifically, for $n$-vertex graphs of maximum degree $d$, we
	prove that if
	$|\beta| d = \omega(\log{n})$
	(where $\beta$ is the inverse temperature parameter),
	then there is no polynomial
	running time identity testing
	algorithm unless $RP=NP$. We also establish computational lower bounds for a broader set of
	parameters under the (randomized) exponential
	time hypothesis. Our proofs utilize insights into the design
	of gadgets
	using random graphs in recent works concerning the hardness of
	approximate counting by Sly (2010).
	
	In the hard-constraint setting, we
	present hardness results 
	for identity testing for
	proper colorings.  
	Our results are based on the presumed
	hardness of \textsc{\#BIS}, the problem of
	(approximately) counting independent sets in bipartite graphs.
	In particular, we prove that identity testing is hard in the same range of parameters 
	where structure learning is known to be hard, which in turn matches the parameter regime for 
	NP-hardness of the corresponding decision problem.
	
\end{abstract}
	
	\thispagestyle{empty}
	
	\newpage
	
	\clearpage
	\setcounter{page}{1}

\section{Introduction}

	We study the \emph{identity testing} problem in the context of
	\emph{spin systems}.
	Spin systems,
	also known as Markov random fields or undirected graphical models,
	are a general framework in statistical physics, theoretical computer science and machine learning
	for modeling interacting systems of simple elements.
	In this type of model,
	the identity testing problem, sometimes also called \emph{goodness-of-fit testing},
	takes the following form:
	given the parameter specification of the model $M$
	and a sampling oracle for the distribution $\mu_{\Mhat}$ of an unknown model $\Mhat$,
	can we efficiently determine if the two models $M$ and $\Mhat$ are the same?
	
	 A spin system consists of a finite graph $G=(V,E)$
	and a set~$S$ of {\it spins}; a {\it configuration\/} $\sigma\in S^V$ assigns a spin value
	to each vertex $v\in V$.  
	The probability of finding the
	system in a given configuration~$\sigma$ is given by the {\it Gibbs\/} (or {\it Boltzmann\/})
	distribution
	\begin{equation*}\label{eqn:gibbs}
	\mu_{G,\Ham}(\sigma) = \frac{\e^{-\Ham(\sigma)}}{Z},
	\end{equation*}
	where $Z$
	 is the normalizing factor known as the partition function and the Hamiltonian~$\Ham$
	contains terms that depend on the spin values at each vertex (a ``vertex potential'') and
	at each pair of adjacent vertices (an ``edge potential''). 
	
	When  $\mu_{G,\Ham}(\sigma) > 0$ for every configuration $\sigma \in S^V$ (i.e., the Gibbs distribution has full support), the spin system is known as a \emph{soft-constraint model};
	otherwise, it is called a \emph{hard-constraint model}. 
	This is a fundamental distinction among spin systems, as it
	determines their application domains 
	and the computational complexity of several inherent problems.	
	We provide here hardness results for identity testing for both soft-constraint and hard-constraint models by considering two
	prototypical systems: the \emph{Ising model} and \emph{proper colorings}.
	
	A naive approach to the identity testing problem
	is to learn first the unknown model $(\Gu,\Hu)$ and then check whether $(G,\Ham)=(\Gu,\Hu)$.
	The problem of learning $\Gu$ from samples is known as \emph{structure learning} and has received tremendous attention; see, e.g.,~\cite{CL,Dasgupta,LGK,AHHK,RWL,BMS,BGSa,Bresler,VMLC,HKM,KM}.
	Once the graph $\Gu$ is known,
	it is often a simpler task to estimate $\Hu$ \cite{Bresler};  this is known as the \emph{parameter estimation} problem.
	Hence, 	one may be inclined to conjecture that identity testing is in fact easier than structure learning,
	and we investigate whether or not this is the case.
%
	The main takeaway from our results is evidence that identity testing is as hard as structure learning for antiferromagnetic (repulsive) systems,  as we 
	show that the settings where these two problems are hard in both the Ising model and proper colorings coincide. 
	

	\subsection{Lower bounds for the Ising model}
	
	The Ising model is the quintessential example of a soft-constraint system
	and is studied in a variety of fields, including
	phylogeny~\cite{Felsenstein,DMR}, computer vision~\cite{GG,Roth},
	statistical mechanics \cite{Georgii,FV}
	and
	deep learning, where it
	appears under the guise of Boltzmann machines~\cite{Hinton2,SLbm,Hinton1}.
	The Ising model on a graph $G=(V,E)$ is parameterized by the inverse temperature $\beta$ which controls
	the strength of the nearest-neighbor interactions.  Configurations of the model are
	the assignments of spins $S = \{+,-\}$ to the vertices of $G$.
	The probability of a configuration $\sigma \in S^V$ is given by the
	Gibbs distribution:
\begin{equation}
\label{eq:ising-def}
\mu_{G,\beta}(\sigma) = \frac{{\e}^{\beta \cdot A(\sigma)}}{Z_{G,\beta}},
\end{equation}
where $A(\sigma)$ is the number of edges of $G$ connecting vertices with the same spin and
$Z_{G,\beta} = \sum_{\sigma \in S^V} \exp(\beta \cdot A(\sigma))$ is the partition function; the associated Hamiltonian is $\Ham(\sigma) = - \beta \cdot A(\sigma)$.

In the \emph{ferromagnetic} case ($\beta>0$) neighboring vertices prefer to align to the same spin,
whereas the opposite happens in the \emph{antiferromagnetic} setting ($\beta<0$).
In more general variants of the model,
one can allow different inverse temperatures $\beta_e$ for each edge $e\in E$,
as well as a vertex potential or external magnetic field.
However, in this work,
our emphasis will be on lower bounds for the identity testing problem, and hence we focus on the
above mentioned simpler homogeneous
setting (all $\beta_e=\beta$) with no external field.

The identity testing problem in the context of the Ising model is the following:
given a graph $G=(V,E)$, a real number $\beta$  and oracle access to independent random samples
from an unknown Ising distribution $\mu_{\Gu,\bu}$,
can we determine if $(G,\beta)=(\Gu,\bu)$?
If the models are distinct but
their associated Gibbs distributions $\mu_{G,\beta}$ and $\mu_{\Gu,\bu}$
are statistically close,
an exponential (in $|V|$) number of samples may be required to determine that $(G,\beta)\neq(\Gu,\bu)$.
Hence,
following a large body of work on identity testing (see, e.g.,~\cite{BFFKRW,DKN,DK16,VV17,DGPP,DDK,CDGR}),
we study this problem in the property testing framework~\cite{RS96,GGR98}.
That is, we are
guaranteed that
either $(G,\beta)=(\Gu,\bu)$ or
$\|\mu_{G,\beta} - \mu_{\Gu,\bu}\| > \varepsilon$,
for some standard distance $\|\cdot\|$ between distributions and $\eps>0$ fixed. 	

The
most common distances for identity testing are
total variation distance and
Kullback-Leibler (KL) divergence, and it is known that a testing algorithm for the latter immediately provides one
for the former~\cite{DDK}. Therefore, since our focus is on lower bounds,
we work with total variation distance, which we denote by ${\|\cdot\|}_\textsc{tv}$.

Identity testing for the Ising model is then formally defined as follows.
For positive integers $n$ and $d$
let $\mathcal{M}(n,d)$ denote the family of all
$n$-vertex graphs of maximum
degree at most~$d$.
\begin{center}
	\fbox{
		\parbox{0.88\textwidth}{
			Given a graph $G \in \mathcal{M}(n,d)$,
			$\beta \in \R$ and
			sample access to a distribution $\mu_{\Gu,\bu}$ for an
			unknown Ising model $(\Gu,\bu)$, where $\Gu\in\mathcal{M}(n,d)$ and $\bu \in \R$,
			distinguish with probability at least $3/4$ between the cases:
			\begin{enumerate}
				\item $\mu_{G,\beta}=\mu_{\Gu,\bu}$;
				\item ${\|\mu_{G,\beta}-\mu_{\Gu,\bu}\|}_\textsc{tv} > \frac 13$.
	\end{enumerate}}}
\end{center}
\noindent
As usual in the property testing setting, the choice of $3/4$ for
the probability of success is arbitrary, and it can be replaced by any constant in the interval $(\frac 12,1)$ at the expense of a constant factor in the running time of the algorithm.
The  choice of $1/3$ for the accuracy parameter is also arbitrary:
we shall see in our proofs that our lower bounds hold for any constant $\varepsilon \in (0,1)$, provided $n$ is sufficiently large; see also Remark~\ref{rmk:error}.

%

Identity testing
for the Ising model
 was studied first by Daskalakis, Dikkala and Kamath~\cite{DDK} who provided
a polynomial time algorithm for the \emph{ferromagnetic} Ising model (the $\beta > 0$ case).
(We will discuss their results in more detail after further discussion.)
In contrast, we present lower bounds for the \emph{antiferromagnetic} Ising model ($\beta < 0$).
Our lower bounds will be for the case when $\bu = \beta$, which means
that they hold even under the additional promise that
the hidden parameter $\bu$ is equal to~$\beta$.
For a discussion of the case $\bu \neq \beta$, as well as for some related open problems, see Section~\ref{sec:disc}.

The structure learning and parameter estimation problems, which, as discussed earlier, can be used to solve the identity testing problem,
have been particularly well-studied in the context of the Ising model \cite{Bresler,VMLC,HKM,KM}.
Recently, Klivans and Meka~\cite{KM}
solved both of these problems for the Ising model
with a nearly optimal algorithm. Their algorithm
learns $\Gu \in \mathcal M(n,d)$ and the parameter $\bu$ in
running time ${\e}^{O(|\bu| d)} \times O(n^2\log{n})$
and sample complexity ${\e}^{O(|\bu| d)} \times O(\log{n})$.
Consequently, when $|\bu| d = O(\log n)$, or when $\beta = \bu$ and $|\beta| d = O(\log n)$,
this method
provides an identity testing algorithm with polynomial (in $n$) running time and sample complexity.
In contrast, when $|\bu| d = \omega(\log n)$ (i.e., $|\bu| d / \log n \rightarrow \infty$),
it is known that the structure learning problem cannot be solved in polynomial time~\cite{SW},
and this approach to identity testing fails.

Our first result is that the identity testing problem for the {antiferromagnetic} Ising model is computationally hard in the same range of parameters.
Specifically, we show that
when
$|\beta| d = \omega (\log n)$---or equivalently when $\beta = \bu$ and $|\beta^*| d = O(\log n)$---there is no polynomial running time identity testing algorithm for $\mathcal{M}(n,d)$ unless $RP = NP$;
$RP$ is the class of problems that can be solved in polynomial time by a randomized algorithm.

\begin{thm}\label{thm:hardness-Ising}
	Suppose
	$n$, $d$ are positive integers  such that
	$3 \le d \le n^{\theta}$ for constant $\theta \in (0,1)$.
	If $RP \neq NP$, then for all real $\beta < 0$
	satisfying
	$|\beta| d = \omega(\log n)$ and all $n$ sufficiently large,	
	there is no polynomial running time algorithm to solve
	the identity testing problem for the \emph{antiferromagnetic} Ising model in $\mathcal{M}(n,d)$.
\end{thm}

In contrast to the above result, Daskalakis, Dikkala and Kamath~\cite{DDK} designed an identity testing algorithm for the Ising model 
with polynomial
running time and 
sample complexity that works for arbitrary values of $\beta$ (positive, negative or even non-homogeneous).
This 
appears to contradict our lower bound in Theorem~\ref{thm:hardness-Ising}.
However, the model in~\cite{DDK}
assumes not only sampling access to the unknown distribution~$\mu_{\Gu,\bu}$, 
but also that the covariances
between the spins at every pair of vertices in
the visible graph $G=(V,E)$ are given.
More precisely, they assume that for every $u,v \in V$ the quantity ${{\E}}_{\mu_{G,\beta}}[X_u X_v]$ is known, where $X_u,X_v \in \{+1,-1\}$ are the random variables corresponding to the spins at $u$ and $v$, respectively.

This
is a reasonable assumption when these quantities can be computed (or approximated up to an additive error) efficiently. However, an immediate consequence of our results is that
in the antiferromagnetic setting 
when $|\beta| d = \omega(\log n)$ there is no 
FPRAS\footnote{A fully polynomial-time randomized approximation scheme (FPRAS) for an optimization problem with optimal solution $Z$ produces an approximate solution $\hat{Z}$ such that, with probability at least $1-\delta$, $(1-\eps)\hat{Z} \leq Z \leq (1+\eps)\hat{Z}$ with running time polynomial in the instance size, $\eps^{-1}$ and $\log (\delta^{-1})$.}
for estimating
$\E_{\mu_{G,\beta}}[X_u X_v]$ unless $RP = NP$.
In a related result, 
Goldberg and Jerrum~\cite{GJcov} showed recently that
there is no FPRAS for (multiplicatively) approximating the pairwise covariances 
for the antiferromagnetic Ising model unless $RP = \#P$.
Further evidence for the hardness of this problem comes from the fact that
sampling is hard in the antiferromagnetic
setting~\cite{SlySun,GSV:Ising}
and in the ferromagnetic model in the presence of inconsistent magnetic fields~\cite{GJfield} 
(i.e., the vertex potential of distinct vertices may have different signs).
In summary, the algorithmic results of~\cite{DDK} are most interesting
for the ferromagnetic Ising model (with consistent fields), where there are known polynomial running time algorithms for estimating the pairwise covariances (see, e.g.,~\cite{JS:Ising,RW,GuoJ,CGHT}).

In Theorem~\ref{thm:hardness-Ising} we assume that $|\beta| d = \omega(\log n)$,
but
our main technical result (Theorem~\ref{thm:main}) is actually more general.
We show that when $|\beta| d \ge c \ln n$, where $c > 0$ is a sufficiently large constant,
if there is an identity testing algorithm with running time $T=T(n)$ and sample complexity $L=L(n)$
then there is also a randomized algorithm with running time $O(T + L n)$ for computing the maximum cut of any graph with $N = n^{\Theta(1)}$ vertices.
Theorem~\ref{thm:hardness-Ising} then follows immediately from the fact that 
either $T$ or $L$
ought to be super-polynomial in $n$, as otherwise we obtain a randomized algorithm
for the maximum cut problem
with
polynomial running time; this would imply that $RP = NP$.

Under a stronger (but also standard) computational theoretic assumption, namely that there is no randomized algorithm
with sub-exponential running time for the 3-SAT problem, i.e., the \emph{(randomized) exponential time hypothesis or $rETH$}~\cite{IP,CIKP}, our main theorem also implies a general lower bound for identity testing that holds for all $\beta$ and $d$ satisfying
$|\beta| d \ge c \ln n$.

\begin{thm}\label{thm:main:intro:eth}
	Suppose $n$, $d$ are positive integers  such that
	$3 \le d \le n^{\theta}$ for constant $\theta \in (0,1)$.
	Then, there exist constants $c=c(\theta) > 0$ and $\alpha=\alpha(\theta) \in (0,1)$
	such that when $|\beta| d \geq c \ln n$, $rETH$ implies that the running time $T(n)$ of any algorithm that solves
	the identity testing problem for the \emph{antiferromagnetic} Ising model in $\mathcal{M}(n,d)$
	satisfies  $T(n) \ge \min\left\{\exp(\Omega(n^{\alpha})),\frac{\exp(\Omega(|\beta| d))}{n}\right\}$.
\end{thm}

We remark that the bound in this theorem is comparable to the $\exp(\Omega(|\beta|d))$ 
lower bound for the sample complexity of structure learning~\cite{SW}, 
albeit requiring that $rETH$ is true.

 The very high level idea of the proof of our main theorem for the Ising model  (Theorem~\ref{thm:main}),
 from which Theorems~\ref{thm:hardness-Ising} and~\ref{thm:main:intro:eth} are derived as corollaries, is as follows: given a graph $H$ and an integer $k$,
 we construct an identity testing instance $\Lambda$ so that the output of the identity testing algorithm on $\Lambda$
 can be used to determine whether there is a cut in $H$ of size at least $k$. 
 A crucial component in our construction is a ``degree reducing'' gadget, which consists of a random bipartite graph and is inspired by similar random gadgets in seminal works on the hardness of
 approximate counting~\cite{Sly}.
 One of the main technical challenges in the paper is to establish precise bounds on the \emph{edge expansion} of these random gadgets. 
 A detailed overview of our proof is given in Section~\ref{sub:proof-overview}.

\subsection{Lower bounds for proper $q$-colorings}

The \emph{proper $q$-colorings} of a graph $G=(V,E)$
constitute
a canonical hard-constraint spin system, with multiple applications in statistical physics and theoretical computer science. 
In this model,
the
vertices of graph $G$ are assigned spins (or colors) from $\{1,\dots,q\}$, and the Gibbs distribution $\mu_G$ 
becomes the uniform distribution over the proper $q$-colorings of the graph $G$.
The identity testing problem for this model in $\mathcal{M}(n,d)$ is defined as follows: given $q$, a graph $G\in \mathcal{M}(n,d)$ and sample access to random $q$-colorings of an unknown graph $G^*\in \mathcal{M}(n,d)$, distinguish with probability at least $3/4$ whether 
$\mu_G=\mu_{G^*}$ or $\norm{\mu_G-\mu_{G^*}}>1/3$. 

We establish lower bounds for this problem, thus initiating the study of identity testing in the context of hard-constraint spin systems.
While identity testing does not seem to have been studied in this context before,
the related structure learning problem has received some attention~\cite{BGS,BCSVsl}.
For proper colorings,
it is known that when $q\geq d+1$ the hidden graph $G$ can be learned from $\poly(n,d,q)$ samples, whereas when $q\leq d$ then the problem is non-identifiable, i.e., there are distinct
graphs with the same collection of $q$-colorings~\cite{BCSVsl}.  
Moreover, for $d\geq d_c(q) = q+\sqrt{q}+\Theta(1)$, or equivalently $q \le d - \sqrt{d} + \Theta(1)$, it was also established in~\cite{BCSVsl} that the easier \emph{equivalent structure learning problem}
(learning any graph with the same collection of $q$-colorings as the unknown graph)
is computationally hard in the sense that the sample complexity is exponential in $n$.
The threshold  $d_c(q)$ coincides exactly with the one for polynomial time/NP-completeness for the problem
of determining if $G$ is $q$-colorable~\cite{EHK, MR}; see~\eqref{eq:dc} for the precise definition of $d_c(q)$.

We prove here that the identity testing problem is also hard when $d\geq d_c(q)$, 
thus establishing another connection between the hardness of identity testing and structure learning.
For this we utilize the complexity of {\BIS},
which is the problem of counting the independent sets in a bipartite graph.
{\BIS} is believed not to have an FPRAS, and
it
has achieved considerable interest
in approximate counting
as a tool for proving relative complexity hardness; see, e.g.,~\cite{DGGJ,GJ,DGJ,BDGJM,CDGJLMR,CGGGJSV,GGJ}.

\begin{thm}
	\label{thm:colorings-hardness}
		Suppose $n$, $d$ and $q$  are positive integers  such that
	$q\geq 3$ and $d\geq d_c(q)$. If {\BIS} does not admit an FPRAS, then there is no polynomial running time algorithm that solves the identity testing problem for proper $q$-colorings in $\mathcal M(n,d)$.
\end{thm}

In the proof of this theorem we reduce the {\BIS}-hard problem of counting 3-colorings in bipartite graphs to identity testing for $q$-colorings. 
The high level idea of our proof is as follows: given a bipartite graph $H$ 
and an approximation $\hat Z$ for the number of $3$-colorings $Z_3(H)$ of $H$,
we construct an identity testing instance that depends on both $H$ and the value of $\hat Z$.
We then show how to use an identity testing algorithm on this instance to check whether $\hat Z$ is an upper or lower bound for $Z_3(H)$. By adjusting $\hat Z$ and repeating this process we converge to a good approximation for $Z_3(H)$.
A crucial element in our construction is again the design of a degree reducing gadget; in this case, our gadget is inspired by similar constructions in \cite{EHK,MR,BCSVsl} for establishing the computational hardness of the decision and (equivalent) structure learning problems for $d\geq d_c(q)$.
Finally, we mention that for $3$-colorings, $d_c(3)=4$ and thus our hardness result holds for all graphs
with maximum degree at least $4$.

\subsection{An algorithm for the ferromagnetic Ising model}

We provide an improved algorithm for the \emph{ferromagnetic} Ising model.
As mentioned, by combining the algorithm in~\cite{DDK} with previous results for sampling (see~\cite{JS:Ising,RW,GuoJ,CGHT}),
one obtains
a polynomial running time algorithm for identity testing in the ferromagnetic setting.
The algorithm
in \cite{DDK} works for symmetric-KL divergence which is a stronger notion of distance.
We show that if one considers instead total variation distance,
then there is a polynomial running time algorithm that solves the identity testing problem
with sample complexity~$\tilde{O}(n^2d^2 \varepsilon^{-2})$.
This is an improvement over the~$\tilde{O}(n^2d^2\beta^2\varepsilon^{-2})$ bound in~\cite{DDK}, as there is no dependence on the inverse temperature $\beta$.
See
Theorem~\ref{thm:alg-main} in Section~\ref{sec:alg} for a precise statement of this result.

\bigskip
The rest of the paper is organized as follows. In Section~\ref{sec:main-result} we state our main technical theorem (Theorem~\ref{thm:main}), and we derive Theorems~\ref{thm:hardness-Ising} and~\ref{thm:main:intro:eth} as corollaries.
In Section~\ref{sub:proof-overview}, we sketch the key ideas in the proof our main result.
 The actual proof of Theorem~\ref{thm:main} is fleshed out in Section~\ref{section:main-proof}.
 Before that, we introduce a useful variant of the maximum cut problem and study its complexity in Section~\ref{sec:two-large-cuts}.
 In Section~\ref{subsection:key-gadget-fact} we provide bounds for the edge expansion of random bipartite graphs which are crucially used in our proofs and could be of independent interest.
Our algorithm for the ferromagnetic Ising model is analyzed in Section~\ref{sec:alg}, and our results for proper $q$-colorings (Theorem~\ref{thm:colorings-hardness}) are derived in Section~\ref{sec:BIS-hardness};
specifically, the reduction for the $q\geq 4$ case is presented in Section~\ref{subsec:proof:col:main:q>=4}, and the more elaborate construction for $q=3$ is given in Section~\ref{subsec:proof:col:main:q=3}.

\section{Lower bounds for the Ising model}
\label{sec:main-result}

To establish our lower bounds in Theorems~\ref{thm:hardness-Ising} and~\ref{thm:main:intro:eth}
we use the computational hardness of the maximum cut (\MC) problem.
Recall that in the search variant of this problem,
given a graph $H$ and an integer $k>0$,
the goal
is to determine
whether there is a cut of size at least $k$ in $H$.
Our main technical result, from which Theorems~\ref{thm:hardness-Ising} and~\ref{thm:main:intro:eth} are derived, is the following.

\begin{thm}\label{thm:main}
	Suppose
	$n$ and $d$ are positive integers  such
	$3 \le d \le n^{1- \rho}$ for some constant $\rho \in (0,1)$.
	Then, for all $n$ sufficiently large,
	there exist $c = c(\rho)>0$
	and an integer $N = \Theta (n^{\min\{\frac \rho 4,\frac{1}{14}\}})$
	such that when $|\beta| d \geq c \ln n$,
	any identity testing algorithm for
	$\mathcal{M}(n,d)$ for the \emph{antiferromagnetic} Ising model with running time $T(n)$ and sample complexity
	$
	L(n) \le \frac{\exp\left(|\beta| d/c\right)}{30n}
	$
	provides a randomized algorithm
	for $\MC$
	on any graph with $N$ vertices.
	This algorithm outputs the correct answer with probability at least $11/20$ and has running time $O(T(n) + n \cdot L(n))$.
\end{thm}

In words, this theorem says that under some mild assumptions, when $|\beta| d \ge c \ln n$,
any identity testing algorithm
with running time $T(n)$ and sample complexity $L(n)$
provides a randomized algorithm for $\MC$~on graphs of $\poly(n)$ size with running time $O(T(n)+n\cdot L(n))$.
The main ideas in the proof of this theorem are described next in Section~\ref{sub:proof-overview}; its actual proof is fleshed out in Section~\ref{subsection:key-gadget-fact}.
Several important consequences of this result, including Theorems~\ref{thm:hardness-Ising} and~\ref{thm:main:intro:eth} from the introduction,
are derived in Section~\ref{subsection:cor:proofs}.


\subsection{Main result for the Ising model: proof overview}
\label{sub:proof-overview}

To establish Theorem~\ref{thm:main} we construct a class $\mathcal N$ of $n$-vertex graphs of maximum degree at most $d$
and show how an algorithm that solves identity testing for $\mathcal N \subset \mathcal M(n,d)$
can be used to solve the $\MC$ problem on graphs with $N = \Theta(n^\alpha)$ vertices, where $\alpha \in (0,1)$ is a constant.
(The exact value for $\alpha$ depends on $d$: if $d=O(1)$, then we can take $\alpha = 1/14$; otherwise, we set $\alpha = \rho/4$.)

Suppose we want to solve the $\MC$ problem for a graph $H = (V,E)$ and $k\in \N$.
For this, we
add two vertices $s$ and $t$ to $H$ and connect both $s$ and $t$ to every vertex in $V$ with $N=|V|$ edges (adding a total of $2N^2$ edges); we also add $w$ edges between $s$ and $t$.
Let $\ol{H}_w$ be the resulting multigraph. (In our proofs we will convert $\ol{H}_w$ into a simple graph, but it is conceptually simpler
to consider the multigraph for now.)
The cut $(\{s,t\},V)$ in $\ol{H}_w$ is of size $2N^2$.

We consider a variant of the $\MC$ problem which we call the $\UMC$ problem.
In this problem, given the graph $H$ and $w\in\N$,
	the goal is to
	determine whether there are at least two cuts in $\ol{H}_w$ of size at least $2N^2$ (see Definition~\ref{def:tlc}).
$\MC$ can be reduced to $\UMC$ by
treating $w$, the number of edges between $s$ and $t$, as a parameter.
Specifically,
if $(S,V\setminus S)$ is a cut of size $k$ in the original graph $H$,  then $(S\cup\{s\},(V\setminus S)\cup\{t\})$
is a cut of size
$$
w+k + N|S| + N|V\setminus S| = w+k + N^2
$$
in $\ol{H}_w$. Hence, $(\{s,t\},V)$ is the unique large cut (i.e., cut of size $\geq 2N^2$) if and only if 
$$w+\MC(H)+ N^2 < 2N^2,$$ 
where $\MC(H)$ denotes the size of the maximum cut of $H$.
Therefore, to solve $\MC$ for $H$ and $k$,
it is sufficient to solve the $\UMC$ problem for $\ol{H}_w$ with $w = N^2-k$; see Section~\ref{sec:two-large-cuts} for the proof of this fact.
This yields that the $\UMC$ problem is NP-complete and the following useful lemma.

\begin{lemma}
	\label{lemma:unique-mc:alg}
	Let $H=(V,E)$ be an $N$-vertex  graph and let $\delta \in (0,1/2]$.
	Suppose there exists a randomized algorithm
	that solves the $\UMC$ problem on inputs $H$ and $w \le N^2$
	with probability at least $1/2 + \delta$ and running time $R$.
	Then, there exists a randomized algorithm to solve $\MC$ for $H$ and $k \in \N$
	with running time $R+O(N^2)$ 
	and success probability at least $1/2 + \delta$. 
\end{lemma}

To determine if $((s,t),V)$ is the maximum cut of $\ol{H}_w$ we can use the antiferromagnetic Ising model on $\ol{H}_w$ as follows.
Every Ising configuration of $\ol{H}_w$
determines a cut:
all the ``$+$'' vertices belong to one side of the cut and the ``$-$'' vertices to the other (or vice versa).
Observe that for every cut of $\ol{H}_w$ there are exactly two corresponding Ising configurations.
The intuition is that the maximum cut of $\ol{H}_w$ corresponds to the two configurations of maximum likelihood in the Gibbs distribution.
Indeed,
when $|\beta|$ is sufficiently large, the distribution
will be well-concentrated on the two configurations that correspond to the maximum cut.
Therefore, a sample from the Gibbs distribution would reveal the maximum cut of $\ol{H}_{w}$ with high probability.

To simulate large magnitudes of $\beta$,
we strengthen the interactions between neighboring vertices of $\ol{H}_w$ by replacing every edge by $2\ell$ edges.
However, sampling from the antiferromagnetic Ising distribution on the resulting multigraph $\ol{H}_{w,\ell}$ is a hard problem,
and we would need to provide a sampling procedure.
Instead, we use the identity testing algorithm as follows.
We construct a simpler Ising model $M^*$ with two key properties: (i) we can easily generate samples from $M ^*$ and (ii) $M^*$ is close in total variation distance to
the Ising model $M = (\ol{H}_{w,\ell},\beta)$ if and only if $(\{s,t\},V)$ is the unique large cut of $\ol{H}_w$. Then, we give $\ol{H}_{w,\ell}$, the parameter $\beta$ and samples from $M^*$ as input to the tester. If the tester outputs \textsc{Yes}, it means that it regarded the samples from $M^*$ as samples from $M$ and so $(\{s,t\},V)$ must be the unique large cut of~$\ol{H}_w$. Conversely, if the tester outputs \textsc{No}, then the total variation distance between $\mu_{M}$ and $\mu_{M^*}$ must be large, in which case $(\{s,t\},V)$ is not the unique large cut of $\ol{H}_w$.

In summary, this argument implies that an identity testing algorithm for $n$-vertex multigraphs gives a polynomial time randomized algorithm for $\MC$ on graphs with $n-2$ vertices.
However,
the maximum degree of $\ol{H}_{w,\ell}$ depends on $\ell$, $N$ and $w$ and could be much larger than $d$. 
Hence,
this argument does not apply for small values of $d$, even if we overlook the fact that we would be using identity testers for multigraphs instead of graphs.
To extend the argument to \emph{simple} graphs in $\mathcal M (n,d)$ for all $3 \le d \le n^{1-\rho}$,
we introduce a ``degree reducing'' gadget, which is reminiscent of gadgets used
in works concerning the hardness of approximate counting~\cite{Sly,SlySun}.

Every vertex of $\ol{H}_{w,\ell}$ is replaced by a random bipartite graph $G = (L \cup R,E_G)$; see Section~\ref{section:main-proof} for the precise random graph model. The graph $G$ has maximum degree at most $d$, and some of its vertices, which we call \emph{ports}, will have degree strictly less than $d$, so that they can be used for connecting the gadgets as indicated by the edges of $\ol{H}_{w,\ell}$.
The resulting {simple} graph, which we denote by $\ol{H}_{w}^\Gamma$, will have maximum degree $d$. ($\Gamma$ 
is the set of parameters of our random graph model;
see Section~\ref{section:main-proof} for the details.)
In similar manner as described above for $\ol{H}_{w,\ell}$,
an identity testing algorithm for
the antiferromagnetic Ising model on $\ol{H}_{w}^\Gamma$ 
can be used to determine whether $(\{s,t\},V)$ is the unique large cut of $\ol{H}_{w}$. 
Since $\ol{H}_{w}^\Gamma$ has maximum degree at most $d$, Theorem~\ref{thm:main} follows.

Finally, we mention that the main technical challenge in our approach is to establish that in every gadget, with high probability, 
either every vertex of $L$ is assigned ``$+$''  and  every vertex of $R$ is assigned ``$-$'' or vice versa; see Theorems~\ref{thm:gadget:fact-ld} and~\ref{thm:gadget:fact-hd}.
To show this, we require very precise bounds on the \emph{edge expansion} of the random bipartite graph $G$.
When $d \rightarrow \infty$, these bounds  can be derived in a fairly straightforward manner from the results in~\cite{BDH}.
However, the case of $d=O(1)$ is more difficult, and it requires for us to define the notion of edge expansion with respect to the ports of the gadget and extending some of the ideas in~\cite{HLW} (see Theorem~\ref{thm:expansion-port}).
Our bounds for the edge expansion of random bipartite graphs may be of independent interest; see Section~\ref{subsection:key-gadget-fact}.

\subsection{Consequences of main result for the Ising model: proofs of Theorems~\ref{thm:hardness-Ising} and~\ref{thm:main:intro:eth}}
\label{subsection:cor:proofs}

In this section we show how to derive
Theorems~\ref{thm:hardness-Ising} and~\ref{thm:main:intro:eth} from
Theorem~\ref{thm:main}.
For Theorem~\ref{thm:hardness-Ising}
we also use the fact that there is no randomized algorithm for $\MC$ with polynomial running time unless $RP=NP$.
(We recall that
$RP$ is the class of problems that can be solved in polynomial time by a randomized algorithm.)
For Theorem~\ref{thm:main:intro:eth} we use a stronger assumption, namely the (randomized) exponential time hypothesis (or $rETH$) \cite{IP,CIKP}.

\begin{proof}[Proof of Theorem~\ref{thm:hardness-Ising}]
	Suppose there is an identity testing algorithm for $\mathcal{M}(n,d)$ with
	$\poly(n)$ running time and sample complexity; that is, $L \le T = \poly(n)$.
	Since $|\beta|d = \omega(\ln n)$,
	$$L\le\frac{\exp\left(|\beta| d/c\right)}{30n}.$$
	Hence, Theorem~\ref{thm:main} implies there is a randomized algorithm for \MC~on graphs of size
	$N = \Theta(n^{\min\{\frac \rho 4,\frac{1}{14}\}})$
	that succeeds with probability at least $11/20$ and has running time $O(T+Ln) = \poly(n)$.
	This implies that \MC~is in $BPP$.
	($BPP$ is the class of all decision problems solvable in polynomial time with
	success probability greater than $1/2$ on both ``yes'' and ``no'' instances; in contrast, $RP$ only allows errors on ``no'' instances.)
	Since \MC~is $NP$-complete, then $NP \subseteq BPP$, and
	the result follows from the standard fact that if $NP \subseteq BPP$, then $RP = NP$; see, e.g.,~\cite{Ko}.
\end{proof}


\begin{proof}[Proof of Theorem~\ref{thm:main:intro:eth}]
	Suppose there exists an identity testing algorithm
	with running time $T$ and sample complexity $L$.
	If $L>\frac{\exp\left(|\beta| d/c\right)}{30n}$, then
	$$
	T \ge L > \frac{\exp\left(|\beta| d/c\right)}{30n}
	$$
	and the result follows.
	Otherwise, when $|\beta| d \ge c \ln n$ for a suitable constant $c = c(\rho) > 0$,
	Theorem~\ref{thm:main} implies that there exists a randomized algorithm for \MC~on graphs with $N = \Theta(n^{\min\{\frac \rho 4,\frac{1}{14}\}})$ vertices
	with running time at most $O(T+Ln)$ and success probability at least $11/20$. However,
	under the assumption that $rETH$ is true, there is no randomized algorithm for \MC~ in such graphs with running time
	${\e}^{o(n^{\alpha})}$, where $\alpha = {\min\{\frac \rho 4,\frac{1}{14}\}}$. Thus, there exist constants $\delta,\gamma > 0$ such that
	$$
	\delta(T+Ln) \ge {\e}^{\gamma n^{\alpha}}.
	$$
	Consequently, if $L \le \frac{{\e}^{\gamma n^{\alpha}}}{2\delta n} $, then $T \ge \frac{{\e}^{\gamma n^{\alpha}}}{2\delta} $; otherwise $T \ge L \ge  \frac{{\e}^{\gamma n^{\alpha}}}{2\delta n}$.
	Putting these bounds together we get
	\[
	T \ge \min\left\{\frac{\exp(|\beta| d/c)}{30n},
	\frac{\exp(\gamma n^{\alpha})}{2\delta n}
	\right\},
	\]
	and the result follows.
\end{proof}


Finally, we note that
Theorem~\ref{thm:main} also implies a polynomial (in $n$) lower bound for the running time of any identity testing algorithm
when $|\beta| d = \Theta(\log n)$.
This regime is not covered by Theorem~\ref{thm:hardness-Ising}, where the assumption is that $|\beta| d = \omega(\log n)$,
and Theorem~\ref{thm:main:intro:eth}
applies to this setting, but under the stronger $rETH$ assumption.
Our next theorem shows that
the weaker complexity theoretic assumption $RP \neq NP$ suffices.

\begin{thm}
	\label{thm:main:poly}
	Suppose
	$n$, $d$ are positive integers  such
	$3 \le d \le n^{1- \rho}$ for some constant $\rho \in (0,1)$
	and $\beta < 0$ is such that $|\beta| d >  c \ln n$, where $c = c(\rho)$ is the constant from Theorem \ref{thm:main}.
	If $RP \neq NP$,
	then, for all $n$ sufficiently large,
	any algorithm that solves the identity testing problem for $\mathcal{M}(n,d)$
	for the antiferromagnetic Ising model
	has running time $T= \Omega(n^{\Delta})$, where $\Delta = \frac{|\beta|d}{c\ln n} -1 $.
\end{thm}

\begin{proof}
	Suppose there is an identity testing algorithm for $\mathcal{M}(n,d)$
	with running time $T$ and sample complexity $L$.
	We consider two cases.
	First, if $L \le \frac{\exp\left(|\beta| d/c\right)}{30n}$, then Theorem \ref{thm:main}
	implies that there is a randomized algorithm for \MC~on graphs with $N = \Theta(n^{\min\{\frac{\rho}{4},\frac{1}{14}\}})$ vertices
	that has running time $O(T+Ln)$ and success probability $11/20$.
	Therefore,
	$T+Ln = n^{\omega(1)}$  since otherwise $NP \subseteq BPP$ and thus $RP = NP$~\cite{Ko}.
	Hence, if $Ln =O(\poly(n))$, then $T = \Omega(n^{\omega(1)})$;
	otherwise $T \ge L = \Omega(n^{\omega(1)})$.
	For the second case, when  $L >\frac{\exp\left(|\beta| d/c\right)}{30n}$, we have
	$$
	T \ge L > \frac{\exp\left(|\beta| d/c\right)}{30n} = \frac{n^\Delta}{30},
	$$
	and the result follows.
\end{proof}

\section{Hardness of the $\UMC$ problem}
\label{sec:two-large-cuts}

In this section we prove Lemma~\ref{lemma:unique-mc:alg}, where the hardness of the $\UMC$ problem is established. 
We formally define the $\UMC$ problem next.

%


\begin{defn}
	\label{def:umc}
	Let $H=(V,E)$ be a graph and  let $w\in\N$. Let $\ol{H}_{w}=(\ol{V},\ol{E})$ be the \emph{multigraph} defined as follows:
	\begin{enumerate}
		\item $\ol{V}$ contains all vertices in $V$ and two new vertices $s$ and $t$; i.e., $\ol{V} = V \cup \{s,t\}$;
		\item $\ol{E}$ contains all edges in $E$, $N$ copies of edges $\{s,v\}$ and $\{t,v\}$ for each $v\in V$, and $w$ copies of the edge $\{s,t\}$.
	\end{enumerate}
\end{defn}
\noindent
Observe that the cut $(\{s,t\},V)$ contains exactly $2N^2$ edges.
\begin{defn}
		\label{def:tlc}
	In the $\UMC$ problem, given a graph $H$ and $w\in\N$,
	the goal is to
	determine whether there are at least two cuts in $\ol{H}_w$ of size at least $2N^2$.
\end{defn}
\noindent
Lemma~\ref{lemma:unique-mc:alg} is a direct corollary of the following lemma, which implies that $\MC$ can be reduced to $\UMC$.

\begin{lemma}
	\label{lemma:key-umc-fact}
	Let $H=(V,E)$ be an $N$-vertex  graph and let $w\in\N$.
	The cut
	$(\{s,t\},V)$ is the \emph{unique} maximum cut of $\ol{H}_w$ if and only if
	$
	\MC(H) < N^2 - w.
	$
\end{lemma}	
Consequently, to solve $\MC$ on inputs $H$ and $k$,
it is sufficient to solve the $\UMC$ problem for $\ol{H}_w$ with $w = N^2-k$.
Hence, Lemma~\ref{lemma:unique-mc:alg} is a direct corollary of Lemma~\ref{lemma:key-umc-fact}. (Note that Lemma~\ref{lemma:key-umc-fact} also implies that the $\UMC$ problem is NP-complete.)



\begin{proof}[Proof of Lemma~\ref{lemma:key-umc-fact}]
	Let $(S,T)$ be a cut of $\hat{H}_w$  (i.e., $S\cup T = \hat{V}$ and $S\cap T=\emptyset$) and let
	$E_{\hat{H}_w}(S,T) \subseteq \ol{E}$ be the set of edges between $S$ and $T$ in $\ol{H}_w$.
	Similarly, for $S',T' \subseteq V$, let  $E_{H}(S',T') \subseteq E$ be the set of edges between $S'$ and $T'$ in $H$.
	
	Let us consider first the cuts $(S,T)$ where $s$ and $t$ belong to the same set.
	Without loss of generality assume $s,t\in S$, and let $S_0 = S\setminus\{s,t\}$. Then, $(S_0,T)$ is a cut of $H$ and so
	\begin{align*}
		\left|E_{\hat{H}_w}(S,T)\right| &= \left|E_{\hat{H}_w}(S_0,T)\right| + \left|E_{\hat{H}_w}(\{s,t\},T)\right|
		= \left|E_{H}(S_0,T)\right| + 2N|T|
		\leq 
		(N-|T|)|T| + 2N|T|.
	\end{align*}
	The quadratic function $f(x) = (N-x)x+2Nx$ is maximized at $x=N$ for $0\leq x \leq N$ and $f(N)= 2N^2$. Thus, $\left|E_{\hat{H}_w}(S,T)\right| \leq 2N^2$, and the maximum value $2N^2$ can be attained only when $|T|=N$; i.e., $S_0 = \emptyset$ and $(S,T) = (\{s,t\},V)$.
	
	Now, for the cuts where $s$ and $t$ belong to distinct sets of the cut,
	let us assume without loss of generality that $s\in S$ and $t\in T$. Let $S_0=S\backslash\{s\}$ and $T_0=T\setminus \{t\}$. Then, $(S_0,T_0)$ is a cut of $H$, and
	\begin{align*}
		\left| E_{\hat{H}_w}(S,T) \right| &= \left| E_{\hat{H}_w}(S_0,T_0) \right| + \left| E_{\hat{H}_w}(S_0,\{t\}) \right| + \left| E_{\hat{H}_w}(\{s\},T_0) \right| + \left| E_{\hat{H}_w}(\{s\},\{t\}) \right|\\
		&= \left| E_H(S_0,T_0) \right| + N^2 + w.
	\end{align*}
	Hence, the maximum cut of this class corresponds to the case when $(S_0,T_0)$ is a maximum cut of $H$, and
	$$
	\left| E_{\hat{H}_w}(S,T) \right|  =  \MC(H) + N^2 + w.
	$$
	
	Combining the above two cases, we conclude that $(\{s,t\},V)$ is the unique maximum cut of $\hat{H}_w$ if and only if
	$
	2N^2 > \MC(H) + N^2 + w,
	$
	and the result follows.
\end{proof}

\section{Proof of  main result for the Ising model: Theorem~\ref{thm:main}}
\label{section:main-proof}

\noindent\textbf{The Ising gadget.} \ \ Suppose $m,p,d,\din,\dout \in\N^+$ are positive integers such that
$m \ge p$, $d \ge 3$ and $\din + \dout = d$.
Let $G = (V_G,E_G)$ be the random bipartite graph defined as follows:
\begin{enumerate}
	\item Set $V_G=L\cup R$, where $|L|=|R|=m$ and $L\cap R=\emptyset$;
	\item Let $P$ be subset of $V_G$ chosen uniformly at random among all the subsets such that $|P\cap L|=|P\cap R|=p$;
	\item Let $M_1,\dots,M_{\din}$ be $\din$ random perfect matchings between $L$ and $R$;
	\item Let $M_1',\dots,M_{\dout}'$ be $\dout$ random perfect matchings between $L\backslash P$ and $R\backslash P$;
	\item Set $E_G = \left(\bigcup_{i=1}^{\din} M_i \right) \cup \left(\bigcup_{i=1}^{\dout} M_i'\right)$;
	\item Make the graph $G$ simple by replacing multiple edges with single edges.
\end{enumerate}
We use $\Gr(m,p,\din,\dout)$ to denote the resulting distribution; that is, $G \sim \Gr(m,p,\din,\dout)$.
Vertices in $P$ are called \textit{ports}. Every port has degree at most $\din$ while every non-port vertex has degree at most $d$.

In our proofs, we use instances of this random graph model
with two different choices of parameters.
For the case when $d$ is such that $3 \le d = O(1)$, we choose
$p = \floor{m^{1/4}}$, $\din = d-1$ and $\dout=1$; otherwise
we take $p = m$ (i.e., every vertex is a port), $\din = \floor{\theta d}$ and $\dout = d -  \floor{\theta d}$ for a suitable constant $\theta \in (0,1)$.
For both parameter choices we establish that
the random graph $G$ is a good expander with high probability; see Section~\ref{subsection:key-gadget-fact}. Using this, we can show that there are only two ``typical''
configurations for the Ising model on $G$, even in the presence of an external configuration (i.e., a boundary condition) exerting influence on the configuration of $G$ via its ports.

We present some notation next that will allow us to formally state these facts.
Let $\sigma^+(G)$ be the configuration of $G = (L\cup R,E_G)$ where every vertex in $L$
is assigned ``$+$'' and every vertex in $R$ is assigned ``$-$'';
similarly, define $\sigma^-(G)$ by interchanging ``$+$'' and ``$-$''.
To capture the notion of an
 external configuration for the bipartite graph $G$,
we assume that 
$G$ is an induced subgraph of a larger graph $G' = (V_{G'},E_{G'})$.
Let $\partial P  = V_{G'} \setminus V_{G}$. Assume that
every vertex in $P \subseteq V_G$ is connected to up to $\dout$ vertices in $\partial P$
and that
there are no edges between $V_G \setminus P$ and $\partial P$ in $G'$.
We use $\{\partial P = \tau\}$ for the event that the configuration in $G'$ of $\partial P$ is $\tau \in \{+,-\}^{\partial P}$.
We can show that for any $\tau$, with high probability over the choice of the random graph~$G$,
the Ising configuration of $V_G$ on $G'$ conditioned on $\{\partial P = \tau\}$ will likely be $\sigma^+(G)$ or $\sigma^-(G)$.

\begin{thm}
	\label{thm:gadget:fact-ld}
	Suppose $\beta < 0$, $3 \le d = O(1)$, $\din = d-1$, $\dout = 1$ and $p = \floor{m^\alpha}$, where $\alpha \in (0,\frac{1}{4}]$ is a constant independent of $m$.
	Then, there exists a constant $\delta > 0$
	such that
	with probability $1-o(1)$ over the choice of the random graph $G$
	the following holds for every configuration $\tau$ on $\partial P$:
	$$
	\mu_{G',\beta} (\{\sigma^+(G),\sigma^-(G)\} \mid \partial P = \tau) \ge 1 - \frac{2 m}{{\e}^{\delta |\beta|   d}}.
	$$
\end{thm}

\begin{thm}
	\label{thm:gadget:fact-hd}
	Suppose $\beta < 0$, $p = m$ and $4 + \frac{1200}{\rho} \le d \le m^{1-\rho}$
	for some constant $\rho \in (0,1)$ independent of $m$.
	Then, there exist constants $\delta = \delta(\rho) > 0$ and $\theta = \theta(\rho) \in (0,1)$
	such that
	when $\din = \floor{\theta d}$ and $\dout = d - \floor{\theta d}$
	the following holds for every configuration $\tau$ on $\partial P$ with probability $1-o(1)$ over the choice of the random graph $G$:
	$$
	\mu_{G',\beta} (\{\sigma^+(G),\sigma^-(G)\} \mid \partial P = \tau) \ge 1 - \frac{2 m}{{\e}^{\delta |\beta|   d}}.
	$$
\end{thm}
\noindent
The proofs of these theorems are given in Section~\ref{subsection:key-gadget-fact}.

\bigskip\noindent\textbf{Testing instance construction.} \  \ 
Let $H=(V,E)$ be a simple $N$-vertex graph and
for $w  \le N^2$
let
$\ol{H}_w$ be the multigraph from Definition~\ref{def:umc}.
We use an instance of the random bipartite graph $\Gr(m,p,\din,\dout)$
as a gadget to define a simple graph $\ol{H}_w^{\Gamma}$, where
$\Gamma$ denotes the set parameters $\{m,p,\din,\dout,\ell\}$; $\ell > 0$ is assumed to be an integer divisible by $\dout$.
The graph $\ol{H}_w^{\Gamma}$ is constructed as follows:
\begin{enumerate}
	\item Generate an instance $G = (L \cup R, E_G)$ of the random graph model $\Gr(m,p,\din,\dout)$;
	\item Replace every vertex of $\ol{H}_w$ by a copy $G_v =  (L_v \cup R_v, E_{G_v})$ of the generated instance $G$;
	\item For every edge $\{v,u\} \in \ol{H}_w$,
	choose $\ell/\dout$ unused ports in $L_v$
	and $\ell/\dout$ unused ports in $R_u$ and
	connect them with a simple bipartite $\dout$-regular graph;
	\item Similarly, for every edge $\{v,u\} \in \ol{H}_w$,
	choose $\ell/\dout$ unused ports in $R_v$
	and $\ell/\dout$ unused ports in $L_u$ and  connect them with a simple bipartite $\dout$-regular graph.
\end{enumerate}
Observe that our construction requires:
\begin{align}
	\din+\dout &= d \le m,\label{eq:cons-cond-1}\\
	\dout &\,\,|\,\, \ell,\label{eq:cons-cond-2}\\
	\ell(N^2+w) &\le p \cdot \dout, \label{eq:cons-cond-3}\\
	\dout^2 &\leq \ell. \label{eq:cons-cond-4}
\end{align}
To see that \eqref{eq:cons-cond-3}  is necessary, note that
the maximum degree of $\ol{H}_w$ is $N^2+w$ (this is the degree of vertices  $s$ and $t$), and so the total out-degree of the ports should be large enough to accommodate $\ell(N^2+w)$ edges.
Observe also that when condition~\eqref{eq:cons-cond-4} holds, there is always a simple bipartite $\dout$-regular graph with $\ell/\dout$ vertices on each side for steps 3 and 4.

The number of vertices in $\ol{H}_w^{\Gamma}$ is $2m(N+2)$
and its maximum degree is $d = \din+\dout$; thus, $\ol{H}_w^{\Gamma} \in \mathcal{M}(2m(N+2),d)$.
Let $I$ be an independent set with $N$ vertices. By setting $H=I$ and $w=0$, we can analogously define the graphs $\ol{I}_0$ and $\ol{I}_0^{\Gamma}$ so that $\ol{I}_0^{\Gamma} \in \mathcal{M}(2m(N+2),d)$.
Let $M$ and $M^*$ denote the Ising models $(\ol{H}_w^{\Gamma},\beta)$ and $(\ol{I}_0^{\Gamma},\beta)$, respectively.
We show next that the models $M$ and $M^*$ 
are statistically close if and only if $(\{s,t\},V)$ is the unique large cut of $\ol{H}_w$.
To formally state this fact we require some additional notation.


For a configuration $\sigma$ on $\ol{H}_w^{\Gamma}$,
we say that the gadget $G_v \!=\!  (L_v \cup R_v, E_{G_v})$ is in the plus (resp., minus) \emph{phase} if all the vertices in $L_v$ (resp., $R_v$)
are assigned ``$+$'' in $\sigma$ and all the vertices in $R_v$ (resp., $L_v$) are assigned ``$-$''.
Let $\Omega_{\mathrm{good}}$ be the set of configurations of $\ol{H}_w^{\Gamma}$ where the gadget of every vertex is either in the plus or the minus phase.
The set of Ising configurations of $\ol{H}_w^{\Gamma}$ and $\ol{I}_0^{\Gamma}$ is the same and is denoted by $\Omega$.
We use $Z_M$, $Z_{M^*}$ for the partition functions of $M$, $M^*$, and $Z_M(\Lambda)$, $Z_{M^*}(\Lambda)$ for their restrictions to a subset of configurations $\Lambda \subseteq \Omega$.
That is,
$
Z_M = \sum_{\sigma \in \Omega} w_M(\sigma)
$
and
$
Z_M(\Lambda) = \sum_{\sigma \in \Lambda} w_M(\sigma)
$
where $w_M(\sigma) :=  {\e}^{\beta A(\sigma)}$ is called the \emph{weight} of the configuration $\sigma$ in $M$; see \eqref{eq:ising-def}. When $\beta < 0$, 
$w_M(\sigma) =  {\e}^{-|\beta| A(\sigma)}$.

The Ising models $M$ and $M^*$ are related as follows.

\begin{lemma}\label{lemma:main-reduction-phase}
	Let $N \ge 1$, $w \ge 0$ be integers and let $\beta < 0$.
	Let $\Gamma = (m,p,\din,\dout,\ell)$ be such that $|\beta|( \ell - d) \ge N$ and conditions \eqref{eq:cons-cond-1}--\eqref{eq:cons-cond-4}
	are satisfied.
	If for the Ising model $M = (\ol{H}_w^\Gamma,\beta)$ we have
	$Z_{M}(\Omega_{\rm good}) \ge (1-\varepsilon) Z_{M}$ for some $\varepsilon \in (0,1)$, then  with probability $1-o(1)$ over the choice of the random graph $G$  the following holds:
	\begin{enumerate}
		\item 	If $(\{s,t\},V)$ is the unique maximum cut of $\ol{H}_w$, then
		$$
		\TV{\mu_M}{\mu_{M^*}} \le 2(\varepsilon + {\e}^{- 2|\beta| d}).
		$$
		\item If $(\{s,t\},V)$ is not the unique maximum cut of $\ol{H}_w$, then
		$$
		\TV{\mu_M}{\mu_{M^*}} > \frac{1}{2} - \varepsilon - {\e}^{- 2|\beta| d}.
		$$
		\item 	If there is a cut in $\ol{H}_w$ with strictly more edges than
		$(\{s,t\},V)$, then
		$$
		\TV{\mu_M}{\mu_{M^*}} \ge 1 - \varepsilon - 2 {\e}^{- 2|\beta| d}.
		$$
	\end{enumerate}
\end{lemma}

The next lemma states that we can easily generate samples from the simpler model $M^*$; this will be crucial
in our proof of Theorem~\ref{thm:main}.

\begin{lemma}
	\label{lemma:main-reduction-sampling}
	Let $N \ge 1$ be an integer and let $\beta < 0$.
	Let $\Gamma = (m,p,\din,\dout,\ell)$ be such that $|\beta| (\ell N-d) \ge N$ and conditions \eqref{eq:cons-cond-1}--\eqref{eq:cons-cond-4}
	are satisfied.
	If for the Ising model $M^* = (\ol{I}_0^\Gamma,\beta)$ we have
	$Z_{M^*}(\Omega_{\rm good}) \ge (1-\varepsilon) Z_{M^*}$ for some $\varepsilon \in (0,1)$, then
	there exists
	a sampling
	algorithm with running time $O(mN)$
	such that with probability $1-o(1)$ over the choice of the random graph $G$,
	the distribution $\mu_{\textsc{alg}}$ of its output satisfies:
	$$
	\TV{\mu_{M^*}}{\mu_{\textsc{alg}}} \le \varepsilon + {\e}^{- 2|\beta| d}.
	$$
\end{lemma}

The proofs of Lemmas~\ref{lemma:main-reduction-phase} and~\ref{lemma:main-reduction-sampling} are provided
in Section~\ref{subsec:M}. We are now ready to prove Theorem~\ref{thm:main}.

\begin{proof}[Proof of Theorem~\ref{thm:main}]
	Let us assume first that $3 \le d = O(1)$.
	In this case, we take
	\begin{align*}
	N &= \floor{n^{1/14}}-2,~~\text{and} ~~
	m = \floor{\frac{n^{13/14}}{2}}.
	\end{align*}
	
	If
	$\lfloor n^{1/4} \rfloor$ and  $\lfloor \frac{n^{13/14}}{2} \rfloor$ are both integers, then  $n = 2m(N+2)$.
	For simplicity and
	without much loss of generality, we assume that this is indeed the case.
	See Remark~\ref{rmk:non-int} for
	a brief explanation on how to extend the current proof to the case when 
	$\lfloor n^{1/4} \rfloor$ or  $\lfloor \frac{n^{13/14}}{2} \rfloor$  are not integers.

	Let $H=(V,E)$ be a
	graph such that $|V|=N$.
	We show that an identity testing  algorithm for $\mathcal{M}(n,d)$ with running time  $T$ and sample complexity
	$
	L  \le \frac{\exp(|\beta| d/c)}{30n}
	$,  henceforth called the \textsc{Tester},
	can be used to solve the $\UMC$ problem
	on inputs $H$ and $w \in \N$
	in $O(T + Ln )$ time.
	
	We recall that in the $\UMC$ problem the goal is to determine whether $(\{s,t\},V)$ is the unique maximum cut of the graph $\ol{H}_w$; see Definitions~\ref{def:umc} and~\ref{def:tlc}.
	For this, we construct the two Ising models $M = (\ol{H}_w^{\Gamma},\beta)$ and $M^* =  (\ol{I}_0^{\Gamma},\beta)$,
	as described at the beginning of this section.
	When $3 \le d = O(1)$,
	we choose
	$p = \lfloor m^{1/4} \rfloor$, $\din = d-1$, $\dout = 1$
	and $\ell = \Theta(n^{9/112})$.
	That is,
	$$
	\Gamma = \{m, \lfloor m^{1/4} \rfloor,d-1,1,\Theta(n^{9/112})\}.
	$$
	Recall that $\ell$ is an integer divisible by $\dout$ by assumption.
	Moreover, $\din + \dout = d$ and $\ol{H}_w^{\Gamma}$, $\ol{I}_0^{\Gamma}$
	have exactly $n$ vertices; hence, $\ol{H}_w^{\Gamma},\ol{I}_0^{\Gamma} \in \mathcal{M}(n,d)$.
	
	Suppose $\sigma$ is sampled according from $\mu_M$.
	Theorem~\ref{thm:gadget:fact-ld} implies that with probability $1-o(1)$ over the choice of the random gadget $G$, if the configuration in the gadget $G_v$ for vertex $v \in \ol{V}$
	is re-sampled in $\sigma$, conditional on the configuration of $\sigma$ outside of $G_v$, then the new configuration in $G_v$ will be in either the plus or minus phase with probability
	at least $1 - \frac{2m}{\e^{\delta|\beta| d}}$, for suitable constant $\delta > 0$.
	A union bound then implies that after re-sampling the configuration in every gadget one by one, the resulting configuration $\sigma'$ is in the set $\Omega_{\rm good}$ with probability $1 - \frac{2m(N+2)}{\e^{\delta|\beta| d}}$.
	The same is true if $\sigma$ were sampled from $\mu_{M^*}$ instead.
	Thus,
	\begin{align}
		\mu_M(\Omega_{\rm good}) &= \frac{Z_M(\Omega_{\rm good})}{Z_M} \ge 1 - \frac{2m(N+2)}{\e^{\delta|\beta| d}},~\text{and}	\label{eq:key-bound-M} 	\\
		\mu_{M^*}(\Omega_{\rm good})  &= \frac{Z_{M^*}(\Omega_{\rm good})}{Z_{M^*}} \ge 1 - \frac{2m(N+2)}{\e^{\delta|\beta| d}} \label{eq:key-bound-M*}.	
	\end{align}
	
	Our choices for $N$ and $\Gamma$ satisfy conditions \eqref{eq:cons-cond-1}--\eqref{eq:cons-cond-4}.
	It can also be checked that  $|\beta| (\ell N-d) \ge N$ when $|\beta| d \ge c \ln n$.
	Then,  \eqref{eq:key-bound-M*} and Lemma~\ref{lemma:main-reduction-sampling} imply that we can generate $L$ samples $\mathcal{S} = \{\sigma_1,\dots,\sigma_L\}$ from a distribution $\mu_{\textsc{alg}}$ in $O(nL)$ time such that
	\begin{equation}
		\label{eq:M*-alg}
		\TV{\mu_{M^*}}{\mu_{\textsc{alg}}} \leq \frac{2m(N+2)}{\e^{\delta|\beta| d}} + \frac{1}{\e^{2 |\beta| d}} \leq  \frac{2n}{\e^{\gamma|\beta| d}},
	\end{equation}
	where $\gamma = \min\{2,\delta\}$.
	
	Our algorithm for $\UMC$
	inputs the Ising model $M$ and the $L$ samples $\mathcal{S}$ to the \textsc{Tester} and outputs the  negation of the \textsc{Tester}'s output.
	Recall that  the \textsc{Tester} returns \textsc{Yes} if it regards the samples in $\mathcal{S}$ as samples from $\mu_{M}$; it returns \textsc{No} if it regards them to be from some other distribution $\nu$ such that
	$\TV{\mu_M}{\nu} > 1/3$.
	
	If $(\{s,t\},V)$ is the unique maximum cut of $\ol{H}_w$, then  \eqref{eq:key-bound-M} and part 2 of Lemma~\ref{lemma:main-reduction-phase}  imply:
	$$
	\TV{\mu_M}{\mu_{M^*}} \le  2\left(\frac{2m(N+2)}{\e^{\delta|\beta| d}} + \frac{1}{\e^{2 |\beta| d}}\right) \leq  \frac{4n}{\e^{\gamma|\beta| d}}.
	$$	
	The triangle inequality and \eqref{eq:M*-alg} imply:
	$$
	\TV{\mu_M}{\mu_{\textsc{alg}}} \leq \TV{\mu_M}{\mu_{M^*}} + \TV{\mu_{M^*}}{\mu_{\textsc{alg}}} \leq   \frac{6n}{\e^{\gamma|\beta| d}}.
	$$
	
	Let $\mu_{M}^{\otimes L}$, $\mu_{M^*}^{\otimes L}$ and $\mu_{\textsc{alg}}^{\otimes L}$ be the product distributions corresponding to $L$ independent samples from $\mu_{M}$, $\mu_{M^*}$ and $\mu_{\textsc{alg}}$ respectively.
	When $c > 1/ \gamma$, we have
	$$
	\TV{\mu_{M}^{\otimes L}}{\mu_{\textsc{alg}}^{\otimes L}} \leq L \TV{\mu_M}{\mu_{\textsc{alg}}} \leq  \frac{\e^{|\beta| d/c}}{30n}\cdot  \frac{6n}{\e^{\gamma|\beta| d}} \le \frac{1}{5}.
	$$
	If $\mathbb{P}$ is the optimal coupling of the distributions $\mu_{M}^{\otimes L}$ and $\mu_{\textsc{alg}}^{\otimes L}$, and $\mathcal{S}'$ is sampled from~$\mathbb{P} (\cdot | \mathcal{S})$, then $\mathcal{S}' = \mathcal{S}$ with probability at least $4/5$.
	Hence, the input to the $\textsc{Tester}$ (i.e., $\mathcal{S}$)
	is distributed according to $\mu_{M}^{\otimes L}$ with probability at least $4/5$.
	Let $\mathcal{F}_\mathcal{S}$ be the event that this is indeed the case.
	Recall that the \textsc{Tester} makes a mistake with probability at most $1/4$.
	Moreover,
	if $\mathcal{F}_\mathcal{S}$ occurs and the \textsc{Tester} does not make a mistake, then the \textsc{Tester} would output \textsc{Yes}.
	Therefore,
	\begin{align*}
		\Pr\left[ \textsc{Tester}~\text{outputs}~\textsc{No} \right]
		&\leq \Pr\left[\neg \mathcal{F}_\mathcal{S} ~~\text{or}~~ \textsc{Tester}~\text{makes}~\text{a}~\text{mistake} \right]\\
		&\leq \Pr\left[ \neg \mathcal{F}_\mathcal{S} \right] + \Pr\left[\textsc{Tester}~\text{makes}~\text{a}~\text{mistake}\right]\\
		&\le \frac{1}{5} + \frac{1}{4} = \frac{9}{20}.
	\end{align*}
	Hence, the \textsc{Tester} returns \textsc{Yes} with probability at least $11/20$ in this case.
	
	When $(\{s,t\},V)$ is not the unique maximum cut of $\ol{H}_w$,
	\eqref{eq:key-bound-M} and the second part of Lemma~\ref{lemma:main-reduction-phase} imply
	\begin{equation}
		\label{eq:error-lb}
		\TV{\mu_M}{\mu_{M^*}} > \frac{1}{2} - \frac{n}{\e^{\delta|\beta| d}} - \frac{1}{{\e}^{2|\beta| d}}  > \frac{1}{3},
	\end{equation}
	where the last inequality holds for $n$ large enough, since by assumption that $|\beta| d\geq c\ln n$
	and we chose $c > 1/\gamma$. Moreover, from \eqref{eq:M*-alg} we get
	$$
	\TV{\mu_{M^*}^{\otimes L}}{\mu_{\textsc{alg}}^{\otimes L}} \leq L \TV{\mu_{M^*}}{\mu_{\textsc{alg}}} \leq  \frac{1}{15}.
	$$
	Thus,
	with probability at least $14/15$ the samples in  $\mathcal{S}$ have distribution $\mu_{M^*}^{\otimes L}$.
	Let $\mathcal{F}^*_\mathcal{S}$ be the event that this is the case. Then,
	$$
	\Pr\left[ \textsc{Tester}~\text{outputs}~\textsc{Yes} \right] \leq \Pr\left[ \neg \mathcal{F}^*_\mathcal{S} \right] + \Pr\left[\textsc{Tester}~\text{makes}~\text{a}~\text{mistake}\right] \leq \frac{19}{60}.
	$$
	Hence, the \textsc{Tester} returns \textsc{No} with probability at least $2/3$.
	
	Therefore, our algorithm can solve the $\UMC$ problem on $\ol{H}_w$ in $O(T+L n)$ time with probability at least $11/20$.
	The results for the case when $3 \le d = O(1)$ then follows from Corollary \ref{lemma:unique-mc:alg}
	and the fact that $|V| = N =  \lfloor n^{1/14} \rfloor-2 \ge \lfloor n^{\min\{\frac \rho 4,\frac{1}{14}\}} \rfloor - 2$.
	
	Now, for $d$ such that $d \le n^{1-\rho}$ but $d \rightarrow \infty$,
	we take
	\begin{align*}
	N = \floor{n^{\rho/4}}-2 ,~~
	m = \floor{\frac{n^{1- \rho/4}}{2}},~~\text{and}~~\Gamma = \{m, m,\floor{\theta d},d-\floor{\theta d},\Theta(n^{1-\frac{3\rho}{4}})\},
	\end{align*}
	where $\theta = \theta(\rho)$ is a suitable constant.
	That is,
	$p =m$, $\din = \floor{\theta d}$, $\dout = d -  \floor{\theta d}$ and $\ell = \Theta(n^{1-\frac{3\rho}{4}})$.
	These choices for $N$, $m$ and $\Gamma$ satisfy conditions \eqref{eq:cons-cond-1}--\eqref{eq:cons-cond-4}.
	Hence,
	\eqref{eq:key-bound-M} and 	\eqref{eq:key-bound-M*} can be deduced in similar fashion using Theorem~\ref{thm:gadget:fact-hd} instead.
	The rest of the proof remains unchanged for this case. Note that for this choice of parameters, $|V| = N =  \lfloor n^{\rho/4} \rfloor-2 \ge \lfloor n^{\min\{\frac \rho 4,\frac{1}{14}\}} \rfloor - 2$.
\end{proof}

\begin{rmk}
	\label{rmk:non-int}
	When either $\lfloor n^{1/4} \rfloor$ or  $\lfloor \frac{n^{13/14}}{2} \rfloor$  is not an integer, then
	$2(m+1)(N+3) \ge n > 2m(N+2)$, and
	an identity testing algorithm for $\mathcal{M}(n,d)$
	with running time $T$ and sample complexity $L$
	can be used to solve the same problem for $\mathcal{M}(2m(N+2),d)$ by simply
	``padding'' the graph from $\mathcal{M}(2m(N+2),d)$ with
	$n-2m(N+2)$ isolated vertices.
	The samples from the hidden distribution can be extended by
	adding isolated vertices and independently assigning ``$+$'' or ``$-$'' with probability $1/2$  to each of them.
	Hence, the resulting algorithm for identity testing in $\mathcal{M}(2m(N+2),d)$ has running time $T' = T(n)+O(n)$ and sample complexity $L(n)$.
	In this case, the proof of Theorem~\ref{thm:main} gives that there is an algorithm for the $\UMC$ problem on graphs with $N$ vertices
	with running time $O(T' + Nm \cdot L(n)) = O(T(n)+n\cdot L(n))$. Hence, Theorem~\ref{thm:main} holds for all sufficiently large $n$.
\end{rmk}

\begin{rmk}
	\label{rmk:error}
	From~\eqref{eq:error-lb} we see that our proof works when $\|\mu_{G,\beta}-\mu_{\Gu,\bu}{\|}_\textsc{tv}> \varepsilon$
	for any constant $\varepsilon \in (0, 1/2)$.
	With a minor modification to the proof, we can extend our result
	to all constant values of $\varepsilon \in (0,1)$, provided $n$ is sufficiently large.
	Specifically, if we assume that
	the starting graph $H$ has a maximum cut of \emph{odd} size,
	then it is straightforward to verify that either $(\{s,t\},V)$ is the unique maximum cut of $\ol{H}_w$
	or there is some other cut with \emph{strictly} more edges. If this is the case, then we can use part 3 of Lemma~\ref{lemma:main-reduction-phase} (instead of part 2) and deduce that the bound in~\eqref{eq:error-lb} becomes
	$\|\mu_{M}-\mu_{M^*}{\|}_\textsc{tv}> 1-\varepsilon$
	for any desired constant $\varepsilon \in (0,1)$.
	It can be easily checked that the $\UMC$ problem restricted to graphs with odd maximum cuts is still hard; in particular, any algorithm for $\UMC$ that works for this type of input can be used to solve the $\MC$ problem efficiently.
\end{rmk}

\subsection{Relating the Ising models $M$ and $M^*$: proof of Lemmas \ref{lemma:main-reduction-phase} and \ref{lemma:main-reduction-sampling}}
\label{subsec:M}

Let $\sigma^+ = \sigma^+(\ol{H}_w^{\Gamma})$ be the configuration of $\ol{H}_w^{\Gamma}$
such that the gadgets for $s$ and $t$ are in the plus phase
and every other gadget is in the minus phase; define $\sigma^- = \sigma^-(\ol{H}_w^{\Gamma})$ in similar manner but interchanging ``$+$'' and ``$-$'' everywhere.
(Recall that a gadget is in the plus (resp., minus) phase if all the vertices in $L$ (resp., $R$)
are assigned ``$+$'' in $\sigma$ and all the vertices in $R$ (resp., $L$) are assigned ``$-$''.)
Let $\Omega^0 = \Omega^0(\ol{H}_w^{\Gamma})= \{\sigma^+,\sigma^-\}$.
The following fact will be used in the proof of
Lemma~\ref{lemma:main-reduction-phase}.

\begin{fact}
	\label{fact:redunction:is}
	Let $N \ge 1$ be an integer and let $\beta < 0$.
	Let $\Gamma = (m,p,\din,\dout,\ell)$ be such that $|\beta| (\ell N-d) \ge N$ and conditions \eqref{eq:cons-cond-1}--\eqref{eq:cons-cond-4}
	are satisfied.
	If for the Ising model $M^* = (\ol{I}_0^\Gamma,\beta)$ we have
	$Z_{M^*}(\Omega_{\rm good}) \ge (1-\varepsilon) Z_{M^*}$ for some $\varepsilon \in (0,1)$, then
	$
	\mu_{M^*}(\Omega^0) \ge 1 - \varepsilon - {\e}^{- 2|\beta| d}.
	$
\end{fact}
\begin{proof}
	The weight of the configurations $\sigma^+$, $\sigma^-$ satisfy:
	$
	w_{M^*}(\sigma^+)=w_{M^*}(\sigma^-) = 1.
	$
	If $\sigma\in \Omega_{\rm good} \setminus \Omega^0$,
	then
	the gadget for either $s$ or $t$ is connected
	to the gadget of at least one other vertex in the same phase by $2\ell N$ edges. Hence,
	$
	w_{M^*}(\sigma) \leq \e^{2 \beta \ell N} = \e^{-2|\beta| \ell N}
	$ and
	\begin{align*}
		Z_{M^*}(\Omega_{\rm good} \setminus \Omega^0) &= \sum_{\sigma\in \Omega_{\rm good} \setminus \Omega_0} w_{M^*}(\sigma)
		\leq |\Omega_{\rm good}|\cdot \e^{-2|\beta|\ell N}
		=  2^{N+2}\cdot \e^{-2|\beta|\ell N}
		\leq {\e}^{- 2|\beta| d},
	\end{align*}
	where in the last inequality we used the fact that $|\beta| (\ell N-d) \ge N$ by assumption.
	Then,
	$$
	Z_{M^*}(\Omega_{\rm good}) \le 2 + {\e}^{- 2|\beta| d}
	$$
	and so
	$$
	\mu_{M^*}(\Omega^0) = \frac{2}{Z_{M^*}(\Omega_{\rm good})} \cdot \frac{Z_{M^*}(\Omega_{\rm good})}{Z_{M^*}} \ge \left(1 - {\e}^{- 2|\beta| d}\right)  \frac{Z_{M^*}(\Omega_{\rm good})}{Z_{M^*}} \ge 1 - {\e}^{- 2|\beta| d}- \varepsilon,
	$$
	as claimed.
\end{proof}
We are now ready to prove Lemmas~\ref{lemma:main-reduction-phase} and~\ref{lemma:main-reduction-sampling}.
\begin{proof}[Proof of Lemma~\ref{lemma:main-reduction-phase}]
	We show that when $(\{s,t\},V)$ is the unique maximum cut of $\ol{H}_w$, then
	\begin{equation}
		\label{eq:key-fact-1}
		\mu_{M}(\Omega^0) \ge 1 - \varepsilon-{\e}^{- 2|\beta| d}.
	\end{equation}
	Since by symmetry $\mu_{M}(\sigma^+)=\mu_{M}(\sigma^-)$ and $\mu_{M^*}(\sigma^+)=\mu_{M^*}(\sigma^-)$,
	Fact~\ref{fact:redunction:is} implies
	$$
	\TV{\mu_M}{\mu_{M^*}} \le \left|\mu_M(\sigma^+)-\mu_{M^*}(\sigma^+)\right| +\frac{\mu_M(\Omega\setminus\Omega^0)+\mu_{M^*}(\Omega\setminus\Omega^0)}{2} \le 2(\varepsilon + {\e}^{- 2|\beta| d})
	$$
	and part 1 follows. (Recall that $\Omega$ is the set of Ising configurations of the graphs $\ol{H}_w^\Gamma$ and $\ol{I}_0^\Gamma$.)
	
	To establish \eqref{eq:key-fact-1}, observe that
	\begin{equation}
		\label{eq:prob-lb}
		\mu_{M}(\Omega^0) = \frac{Z_M(\Omega^0)}{Z_M(\Omega_{\rm good})} \cdot \frac{Z_M(\Omega_{\rm good})}{Z_M} \ge \frac{(1-\varepsilon) Z_M(\Omega^0)}{Z_M(\Omega_{\rm good})},
	\end{equation}
	where the last inequality follows from the assumption that $Z_M(\Omega_{\rm good}) \ge (1-\varepsilon)Z_M$.
	
	For $\sigma \in \Omega_{\rm good}$, let $\mathcal{I}(\sigma)$
	be the number of edges $\{u,v\}$ of $\ol{H}_w$ such that the gadgets corresponding to vertices $u$ and $v$ in $\ol{H}_w^\Gamma$ are in the same phase in $\sigma$.
	Since every edge of $\ol{H}_w$ correspond to exactly $2\ell$ edges in $\ol{H}_w^\Gamma$, we have
	$w_M(\sigma) =  {\e}^{2\beta\ell \mathcal{I}(\sigma)} = {\e}^{-2|\beta|\ell \mathcal{I}(\sigma)}$.
	Moreover, $\mathcal{I}(\sigma^+) = \mathcal{I}(\sigma^-) = w+|E|$, where $E$ is the set of edges of the graph $H$.
	When $(\{s,t\},V)$ is the unique maximum cut of $\ol{H}_w$,
	$\mathcal{I}(\sigma) \ge w+|E|+1$
	for all $\sigma\in \Omega_{\rm good} \setminus \Omega_0$.
	Therefore,
	$Z_M(\Omega^0) = 2 {\e}^{-2|\beta|\ell(w + |E|)}$ and for $\sigma\in \Omega_{\rm good} \setminus \Omega_0$
	$$
	w_{M}(\sigma) \leq \e^{-2|\beta|\ell(w+|E|+1)} = \frac{Z_M(\Omega^0)}{2\e^{2|\beta|\ell}}.
	$$
	Then,
	\begin{align*}
		Z_{M}(\Omega_{\rm good}) &=Z_M(\Omega^0) + \sum_{\sigma\in \Omega_{\rm good} \setminus \Omega_0} w_{M}(\sigma)
		\leq Z_M(\Omega^0) + |\Omega_{\rm good}|\cdot\frac{Z_M(\Omega^0)}{2\e^{2|\beta|\ell}}
		= Z_M(\Omega^0) \left(1 + \frac{2^{N+1}}{\e^{2|\beta|\ell}}\right).
	\end{align*}
	By assumption $|\beta|( \ell-d) \ge N$, so
	$
	Z_{M}(\Omega_{\rm good}) \leq Z_M(\Omega^0) \left(1 + {\e}^{- 2|\beta| d}\right).
	$
	Thus, we deduce that
	$$
	\frac{Z_M(\Omega^0)}{Z_M(\Omega_{\rm good})} \geq \frac{1}{1 + {\e}^{- 2|\beta| d}} \geq 1 - {\e}^{- 2|\beta| d}.
	$$
	Plugging this bound into \eqref{eq:prob-lb} gives \eqref{eq:key-fact-1} and the proof of part 1 of the lemma is complete.
	
	For the second part we show that when $(\{s,t\},V)$ is not the unique maximum cut of $\ol{H}_w$, then
	\begin{equation}
		\label{eq:conditional-bound}
		\mu_{M}(\Omega^0) \le \frac{1}{2}.
	\end{equation}
	By Fact~\ref{fact:redunction:is},
	$
	\mu_{M^*} (\Omega^0) \geq 1 - \varepsilon - {\e}^{- 2|\beta| d};
	$
	hence,
	$$
	\TV{\mu_M}{\mu_{M^*}} \geq \left| \mu_{M^*}(\Omega^0) - \mu_{M} (\Omega^0) \right| \geq \frac{1}{2} - \varepsilon - {\e}^{- 2|\beta| d}
	$$
	and part 2 follows.
	
	To establish \eqref{eq:conditional-bound},
	let $(S,\hat{V}\setminus S)\neq (\{s,t\},V)$ be a maximum cut of the graph $\hat{H}_w$. Let $\sigma_*^+$ (resp., $\sigma_*^-$) be the Ising configuration of the graph $\ol{H}_w^\Gamma$ where the gadgets corresponding to vertices in $S$ are in the plus phase (resp., minus phase), and the remaining gadgets are in the minus phase (resp., plus phase).
	Since $(S,\hat{V}\setminus S)$ is a maximum cut of $\hat{H}_w$,
	$\mathcal{I}(\sigma_*^+) = \mathcal{I}(\sigma_*^-) \le \mathcal{I}(\sigma^+) = \mathcal{I}(\sigma^-)$
	and so
	$
	w_M(\{\sigma_*^+,\sigma_*^-\}) \geq w_M(\Omega^0).
	$
	It follows that
	$$
	Z_M \geq w_M(\Omega^0) + w_M(\sigma_*^+,\sigma_*^-) \geq 2w_M(\Omega^0)
	$$
	and
	$
	\mu_{M}(\Omega^0) = {w_M(\Omega^0)}/{Z_M} \leq 1/2;
	$
	this gives \eqref{eq:conditional-bound} and part 2 follows.

	Part 3 follows in similar fashion.
	Let $(S,\hat{V}\setminus S)$ be a cut of $\hat{H}_w$
	with strictly more edges than the cut $(\{s,t\},V)$.
	Let $\sigma_*^+$ (resp., $\sigma_*^-$) be Ising configuration of $\ol{H}_w^\Gamma$
	determined by $(S,\hat{V}\setminus S)$
	as in the proof of part 2.
	Then,
	$\mathcal{I}(\sigma_*^+) = \mathcal{I}(\sigma_*^-) < \mathcal{I}(\sigma^+) = \mathcal{I}(\sigma^-)$
	and
	$
	w_M(\{\sigma_*^+,\sigma_*^-\}) \geq {\e}^{2|\beta| \ell} w_M(\Omega^0).
	$
	It follows that
	$$
	Z_M \geq w_M(\Omega^0) + w_M(\sigma_*^+,\sigma_*^-) \geq (1+ {\e}^{2|\beta| \ell})w_M(\Omega^0)
	$$
	and
	$$
	\mu_{M}(\Omega^0) = \frac{w_M(\Omega^0)}{Z_M} \leq \frac{1}{1+ {\e}^{2|\beta| \ell}} \le  {\e}^{-2|\beta| \ell} \le {\e}^{-2|\beta| d},
	$$
	where in the last inequality we use the assumption that $|\beta|( \ell-d) \ge N$ and so $\ell \ge d$.
	This bound and Fact~\ref{fact:redunction:is} imply
	$$
	\TV{\mu_M}{\mu_{M^*}} \geq \left| \mu_{M^*}(\Omega^0) - \mu_{M} (\Omega^0) \right| \geq 1 - \varepsilon - 2 {\e}^{- 2|\beta| d},
	$$
	as claimed.
	\end{proof}

\begin{proof}[Proof of Lemma~\ref{lemma:main-reduction-sampling}]
	By Fact~\ref{fact:redunction:is},
	$
	\mu_{M^*}(\Omega^0) \ge 1 -  \varepsilon - {\e}^{- 2|\beta| d}.
	$
	Also, $\mu_{M^*}(\sigma^+)=\mu_{M^*}(\sigma^-) = \mu_{M^*}(\Omega^0)/2$.
	Hence, if $\mu_{\textsc{alg}}$ is the uniform distribution over $\{\sigma^+,\sigma^-\}$, we have
	$$
	\TV{\mu_{M^*}}{\mu_{\textsc{alg}}} =  \left|\mu_{M^*}(\sigma^+)-\frac 12\right| + \frac{1-\mu_{M^*}(\Omega^0)}{2} \le \varepsilon + {\e}^{- 2|\beta| d}.
	$$
	The results follows from the fact that a sample from $\mu_{\textsc{alg}}$ can be generated in $O(mN)$ time.	
\end{proof}

\section{Properties of the Ising gadget}
\label{subsection:key-gadget-fact}

In this section we prove the key properties of the random bipartite graph $\Gr(m,p,\din,\dout)$ that were used
in Section~\ref{section:main-proof} to establish our main result.
In particular, we establish Theorems~\ref{thm:gadget:fact-ld} and~\ref{thm:gadget:fact-hd}.
Throughout this section we let $G= (V_G = L \cup R, E_G)$ be an instance of $\Gr(m,p,\din,\dout) $ as defined in Section~\ref{section:main-proof}.
Recall that $|L|=|R|=m$ and that there is a set $P$ of ports such that
$|P \cap L| = |P \cap R| = p$.
For $S,T\subset V_G$  define
\begin{align*}
E(S,T) &= \left\{ \{u,v\}\in E_G: u\in S, v\in T \right\}.
\end{align*}
In the proof of Theorems~\ref{thm:gadget:fact-ld} and~\ref{thm:gadget:fact-hd} we crucially use the following facts about the edge expansion of the random graph $G$.
\begin{thm}
	\label{thm:expansion-generic-hd}
	Suppose $p=m$ and $3 \le \din \le d \le m^{1-\rho}$
	where $\rho \in (0,1)$ is a constant independent of $m$.
	Then, with probability $1-o(1)$ over the choice of the random graph $G$:
	$$
	\min_{\substack{S\subset V_G:\\0<|S|\leq m}} \frac{|E(S,V_G\backslash S)|}{|S|} \ge  \frac{\rho \din}{300}.
	$$
\end{thm}
\begin{thm}
	\label{thm:expansion-generic}
	Suppose $3 \le d = O(1)$,
	$p = \floor{m^\alpha}$ with	
	$\alpha \in (0,\frac{1}{4}]$,
	$\din=d-1$
	and
	$\dout = 1$.
	Then, there exists a constant $\gamma > 0$  independent of $m$ such that with probability $1-o(1)$ over the choice of the random graph $G$:
	$$
	\min_{\substack{S\subset V_G:\\0<|S|\leq m}} \frac{|E(S,V_G\backslash S)|}{|S|} \ge \gamma d.
	$$
\end{thm}
\begin{thm}
	\label{thm:expansion-port}
	Suppose $3 \le d = O(1)$,
	$p = \floor{m^\alpha}$ with	
	$\alpha \in (0,\frac{1}{4}]$,
	$\din=d-1$
	and
	$\dout = 1$. Then, there exists a constant $\gamma > 0$ independent of $m$ such that with probability $1-o(1)$ over the choice of the random graph $G$: $$
	\min_{\substack{S\subset V_G:\\0<|P\cap S|\leq |S|\leq m}} \frac{|E(S,V_G\backslash S)|}{|P\cap S|} > 1 + \gamma.
	$$
\end{thm}
\begin{proof}[Proof of Theorems~\ref{thm:gadget:fact-ld} and~\ref{thm:gadget:fact-hd}]
	Let $\sigma$ and $\tau$ be Ising configurations on $V_G$ and $\partial P$ respectively.
	Let $P^+ \subseteq \partial P$ be the set of vertices of $\partial P$ that are assigned ``$+$'' by $\tau$ and let $P^-$ be those that are assigned ``$-$''.
	Let $L_\sigma^+\subseteq L$ and $L_\sigma^- \subseteq L$ be the set of vertices of $L$ that are assigned ``$+$'' and ``$-$'', respectively, in $\sigma$ and define
	$R_\sigma^+,R_\sigma^- \subseteq R$ similarly.
	Let $S_\sigma$ denote the set of smaller cardinality between $L_\sigma^+ \cup R_\sigma^-$ and $L_\sigma^- \cup R_\sigma^+$; hence $S_\sigma \le m$. Suppose $|S_\sigma| > 0$; i.e., $\sigma \neq \sigma^+(G)$ and $\sigma \neq \sigma^-(G)$, where $\sigma^+(G)$ (resp., $\sigma^-(G)$) is the configuration of $G$ in which every vertex in $L$
	is assigned ``$+$'' (resp., ``$-$'') and every vertex in $R$ is assigned ``$-$'' (resp., ``$+$'').
	
	For $S,T \subseteq V_G \cup \partial P$, we use $[S,T]$ for the number of edges between $S$ and $T$ in the graph $G' = (V_{G} \cup \partial P,E_G \cup E(P,\partial P))$.
	Observe that the weights of $\sigma^+(G)$ and $\sigma^-(G)$ in $G'$ conditional on $\tau$ are:
	\begin{align*}
	w^+ := w^\tau_{G',\beta}(\sigma^+(G)) &=  {\e}^{-|\beta|([L, P^+] + [R , P^-])},~\text{and}\\
	w^- := w^\tau_{G',\beta}(\sigma^-(G)) &=  {\e}^{-|\beta|([L,P^-] + [R , P^+])}.
	\end{align*}
	Henceforth we use $w^\tau(\cdot)$ for $w^\tau_{G',\beta}(\cdot)$.
	
	We consider first the case when $S_\sigma = L_\sigma^+ \cup R_\sigma^-$.
	Then,
	\begin{align}
	w^\tau(\sigma) &= \exp\left[ -|\beta| (|E(S_\sigma,V_G\setminus S_\sigma)|  + [L^+_\sigma , P^+]+ [L^-_\sigma , P^-]+ [R^+_\sigma , P^+]+ [R^-_\sigma , P^-])\right]\notag\\
	&= w^- \cdot \exp\left[ -|\beta| (|E(S_\sigma,V_G\setminus S_\sigma)|  + [L^+_\sigma , P^+] + [R^-_\sigma , P^-] - [L^+_\sigma , P^-]- [R^-_\sigma , P^+])\right]\notag\\
	&\le w^- \cdot  \exp\left[ -|\beta| (|E(S_\sigma,V_G\setminus S_\sigma)|  - [L^+_\sigma , \partial P] - [R^-_\sigma ,\partial P] )\right] \notag\\
	&\le w^-\cdot  \exp\left[ -|\beta| (|E(S_\sigma,V_G\setminus S_\sigma)|  - [S_\sigma , \partial P])\right].
	\label{eq:weight-bound-1}
	\end{align}
	where the first inequality follows from $[L_\sigma^+ , P^-]-[L_\sigma^+ , P^+] \le [L_\sigma^+ , \partial P]$ and $[R_\sigma^- , P^+]-[R_\sigma^- , P^-] \le [R_\sigma^- , \partial P]$.
	
	In Theorem~\ref{thm:gadget:fact-ld}
	we assume that $3 \le d = O(1)$, $p = \floor{m^\alpha}$ with $\alpha \in (0,\frac{1}{4}]$, $\din = d-1$ and $\dout = 1$.
	Hence, $ [S_\sigma , \partial P] = |S_\sigma \cap P|$ and
	since $0 < |S_\sigma| \le m$,
	Theorems~\ref{thm:expansion-generic} and~\ref{thm:expansion-port} imply that  exists a constant $\gamma > 0$ such that with probability $1-o(1)$ over the choice of the random graph $G$ we have
	\begin{align*}
	\frac{|E(S_\sigma,V_G \setminus S_\sigma)|}{|S_\sigma|} &\ge \gamma d,~\textrm{and}\\
	\frac{|E(S_\sigma,V_G\backslash S_\sigma)|}{|S_\sigma \cap P|} &\ge 1 + \gamma.
	\end{align*}
	Combining these two inequalities we get for
	$\delta = \frac{\gamma^2}{1+\gamma}$ that
	\begin{equation*}
	\label{eq:exr-imeq}
	|E(S_\sigma,V_G \setminus S_\sigma)| \ge |S_\sigma \cap P| + \delta d |S_\sigma| = [S_\sigma , \partial P] + \delta d |S_\sigma| .
	\end{equation*}
	Plugging this bound into \eqref{eq:weight-bound-1},
	\begin{equation}
	\label{eq:weight-bound-2}
	w^\tau(\sigma) \le w^-\cdot  {\exp}\left[ -\delta |\beta| d |S_\sigma|\right].
	\end{equation}
	Under the assumptions in Theorem~\ref{thm:gadget:fact-hd},
	we can also establish \eqref{eq:weight-bound-2} as follows.
	When $m^{1-\rho} \ge d \ge \din = \floor{\theta d} \ge 3$,  Theorem \ref{thm:expansion-generic-hd} implies that
	$$
	|E(S_\sigma,V_G\setminus S_\sigma)|  \ge \frac{\rho \din}{300} |S_\sigma| = \frac{\rho \floor{\theta d}}{300}   |S_\sigma|.
	$$
	Moreover,
	$$
	[S_\sigma , \partial P] \le \dout |S_\sigma| = (d-\floor{\theta d}) |S_\sigma|.
	$$
	Hence,  taking
	$$
	\theta = \frac{300 + 0.75 \rho}{300+\rho}
	$$
	we get  that when $d \ge 4 + \frac{1200}{\rho}$:
	$$
	\frac{\rho \floor{\theta d}}{300}   - (d-\floor{\theta d})  \ge \frac{\rho d}{600}.
	$$
	Together with \eqref{eq:weight-bound-1} this implies
	$$
	w^\tau(\sigma) \le   w^-\cdot  {\exp}\left[ -\frac{\rho |\beta| d |S_\sigma|}{600}\right],
	$$
	which gives \eqref{eq:weight-bound-2} for $\delta \leq \rho/600$.
	Observe that our choice of $\theta$ guarantees $d-1 \ge \din = \floor{\theta d} \ge 3$ for all $d \ge 4$.
	
	For the case when $S_\sigma = L_\sigma^- \cup R_\sigma^+$ we deduce analogously that for a suitable $\delta > 0$
	\begin{align}
	w^\tau(\sigma) \le w^+ \cdot {\exp}\left[ -\delta |\beta| d |S_\sigma|\right].    \label{eq:weight-bound-3}
	\end{align}
	
	Let $\Omega_G$ be the set of Ising configurations of the graph $G$.
	By definition, the partition function $Z_{G',\beta,\tau}$
	for the conditional distribution $\mu_{G',\beta}(\cdot \mid \partial P = \tau)$
	satisfies:
	$$
	Z_{G',\beta,\tau} = \sum_{\sigma \in \Omega_G} w^\tau(\sigma) \le \sum_{\sigma:\; 0\leq |L_\sigma^+ \cup R_\sigma^-| \leq m} w^\tau(\sigma) + \sum_{\sigma:\; 0\leq |L_\sigma^- \cup R_\sigma^+| \leq m} w^\tau(\sigma).
	$$
	From \eqref{eq:weight-bound-2} we get
	\begin{align*}
	\sum_{\sigma:\; 0\leq |L_\sigma^+ \cup R_\sigma^-| \leq m} w^\tau(\sigma)
	&\leq \sum_{\sigma:\; 0\leq |L_\sigma^+ \cup R_\sigma^-| \leq m}  w^-  \!\cdot\!{\e}^{-\delta |\beta| d |L_\sigma^+ \cup R_\sigma^-|}\\
	&= w^-  \sum_{k=0}^m \binom{2m}{k} {\e}^{ -\delta |\beta| d k }
	\leq w^- \!\left(1+\e^{-\delta|\beta| d}\right)^{2m}.
	\end{align*}
	Similarly, we deduce from \eqref{eq:weight-bound-3}
	$$
	\sum_{\sigma:\; 0\leq |L_\sigma^- \cup R_\sigma^+| \leq m} w^\tau(\sigma) \leq w^+ \left(1+\e^{-\delta|\beta| d}\right)^{2m}.
	$$
	Hence,
	$$
	Z_{G',\beta,\tau} \leq \left(w^- +w^+\right) \left(1+\e^{-\delta|\beta| d}\right)^{2m},
	$$
	and
	$$
	\mu_{G',\beta}( \{\sigma^+(G),\sigma^-(G)\} \mid \partial P = \tau) = \frac{w^++w^- }{Z_{G,\beta,\tau}} \geq \frac{1}{\left(1+\e^{-\delta|\beta| d}\right)^{2m}} \geq 1 - \frac{2m}{\e^{\delta |\beta| d}},
	$$
	where in the last inequality we use $(\frac{1}{1+x})^{2m}\geq (1-x)^{2m} \geq 1-2mx$ for all $x\in[0,1)$.
\end{proof}

\subsection{Gadget expansion when $d\rightarrow \infty$: proof of Theorem~\ref{thm:expansion-generic-hd}}
\label{subsec:ge-hd}

We derive~Theorem~\ref{thm:expansion-generic-hd} as a consequence of the following
two properties of the random bipartite multigraphs obtained as the union of perfect matchings.
\begin{lemma}[{\cite[Theorem 4]{BDH}}]
	\label{lemma:spectral}
	Let $\hat{G} = (L_{\hat{G}} \cup R_{\hat{G}},E_{\hat{G}})$ be
	a random bipartite graph obtained
	as the union of $d$ random perfect matchings between  $L_{\hat{G}}$ and $R_{\hat{G}}$.
	Let $\lambda_2(\hat{G})$ denote the second largest eigenvalue of its adjacency matrix.
	Then,
	for $3 \le d=O(1)$ and any constant $\delta>0$ (independent of $m$), with probability $1-o(1)$, the following holds:
	$$
	\lambda_2(\hat{G}) < 2\sqrt{d-1}+\delta.
	$$
\end{lemma}
\begin{lemma}\label{lemma:multiedges}
	Let $\hat{G} = (L_{\hat{G}} \cup R_{\hat{G}},E_{\hat{G}})$ be a random bipartite graph obtained as the union of $d$
	random perfect matchings between $L_{\hat{G}}$ and $R_{\hat{G}}$.
	Suppose $|L_{\hat{G}}|=|R_{\hat{G}}| = m$ and that $d \le m^{1-\rho}$ for some constant $\rho \in (0,1)$
	independent of $m$. Then, the probability that an edge between $L_{\hat{G}}$ and $R_{\hat{G}}$ is chosen by more than $\ceil{ 3/\rho}$ random perfect matchings is $O(m^{-1})$.
\end{lemma}
\begin{proof}
	Let $L_{\hat{G}} = \{v_1,\dots,v_m\}$ and $R_{\hat{G}} = \{u_1,\dots,u_m\}$. Let $X_{ij}$ be the random variable corresponding to the number of perfect matchings that use the edge $\{v_i,u_j\}$ and let $\kappa = \ceil{ 3/\rho}$. Then,
	\begin{align*}
	\Pr[X_{ij} \ge \kappa] = \sum_{a = \kappa}^d {d \choose a} \frac{1}{m^a} \left(1 - \frac 1m\right)^{d-a} \le \sum_{a = \kappa}^d \left(\frac{\e d}{a m}\right)^a \le  \left(\frac{\e}{\kappa m^\rho}\right)^\kappa \sum_{a = 0}^{d-\kappa} \left(\frac{ \e}{\kappa m^\rho}\right)^a \leq O(m^{-3}).
	\end{align*}
	The result follows by a union bound over the pairs $\{v_i,u_j\}$.
\end{proof}
\begin{proof}[Proof of Theorem~\ref{thm:expansion-generic-hd}]
	Let $G  = (V_G = L \cup R, E_G)$ be a random bipartite graph sampled from $\Gr(m,p,\din,\dout)$.
	For $k \ge 3$
	let $F_{k} =(L\cup R,E(F_{k}))$ be a random bipartite graph obtained as the union of $k$ random perfect matchings
	between $L$ and $R$.
	For $S \subseteq V_G$ let $E_{F_{k}}(S,V_G\setminus S) \subseteq E(F_{k})$ be set edges between $S$ and $V_G\setminus S$ in $F_{k}$.
	Lemma \ref{lemma:multiedges} implies that for $\kappa = \ceil{3/\rho}$,  with probability $1-o(1)$,
	we have
	\begin{align}
	\label{eq:linear-exp:ld-1}
	\min_{\substack{S\subset V_G: \\ 0 <|S|\leq m}} \frac{|E(S,V_G\backslash S)|}{|S|}
	\geq \min_{\substack{S\subset V_{{G}}: \\ 0<|S|\leq m}} \frac{|E_{F_{\din}}(S,V_{{G}}\setminus S)|}{\kappa |S|}.
	\end{align}
	Let $d' \ge 3$ be the unique integer divisible by $3$ such that  $\din \ge d' \ge \din-2$. Then,
	\begin{align}
	\label{eq:linear-exp:ld-2}
	\min_{\substack{S\subset V_{{G}}: \\ 0<|S|\leq m}} \frac{|E_{F_{\din}}(S,V_{{G}}\setminus S)|}{\kappa |S|} \ge \min_{\substack{S\subset V_{{G}}: \\ 0<|S|\leq m}} \frac{|E_{F_{d'}}(S,V_{{G}}\setminus S)|}{\kappa |S|}.
	\end{align}
	The random graph $F_{d'}$ can also be obtained as the union of $d'/3$ independent instances of the random graph $F_3$;
	let $F_3^{(1)},\dots,F_{3}^{(d'/3)}$ be these instances.
	For each $i \in\{1,\dots,d'/3\}$, Cheeger's inequality and Lemma \ref{lemma:spectral} (with $\delta = 0.01$) imply that with probability $r = 1-o(1)$:	
	$$
	\min_{\substack{S\subset V_{{G}}: \\ 0<|S|\leq m}} \frac{\left|E_{F_3^{(i)}}(S,V_{{G}}\setminus S)\right|}{ |S|} \ge \frac{3-\lambda_2(F_3^{(i)})}{2} \ge 0.08.
	$$
	Let $Z$ be number of $F_3^{(i)}$'s that satisfy this property. We have $\E[Z] = rd'/3$, $\Var(Z) = r(1-r)d'/3$ and by Chebyshev's inequality for sufficiently large $m$
	$$
	\Pr\left[Z \le \frac{3d'}{10}\right] \leq \Pr\left[ \big|Z-\E[Z]\big|\geq \left(\frac{r}{3}-\frac{3}{10}\right)d'  \right] \le \frac{r(1-r)}{3\left(\frac{r}{3}-\frac{3}{10}\right)^2d'} = o(1).
	$$
	Therefore, with probability $1-o(1)$ we have
	$$
	\min_{\substack{S\subset V_{{G}}: \\ 0<|S|\leq m}} \frac{|E_{F_{d'}}(S,V_{{G}}\setminus S)|}{ |S|} = \min_{\substack{S\subset V_{{G}}: \\ 0<|S|\leq m}}  \sum_{i=1}^{\frac {d'}{3}} \frac{ \left|E_{F_3^{(i)}}(S,V_{{G}}\setminus S)\right|}{|S|} \geq 0.08 Z \ge 0.024 d'.
	$$
	This bound, combined with \eqref{eq:linear-exp:ld-1} and \eqref{eq:linear-exp:ld-2}, implies that with probability $1-o(1)$:
	$$
	\min_{\substack{S\subset V_G: \\ 0 <|S|\leq m}} \frac{|E(S,V_G\backslash S)|}{|S|}
	\geq  \frac{0.024 d'}{\kappa} \ge \frac{\rho \din}{300},
	$$
	where in the last inequality we use $\kappa \leq 3/\rho +1 \leq 4/\rho$ and $d'\geq \frac{3}{5}\din$ for $\din\geq 3$.
\end{proof}	

\subsection{Gadget expansion when $d = O(1)$: proof of Theorems~\ref{thm:expansion-generic} and~\ref{thm:expansion-port}}
\label{subsec:ge-ld}

Let $m$, $p$ and $d$ be positive integers such that
$3 \le d = O(1)$ and
$p = \floor{m^\alpha}$ for some constant	
$\alpha \in (0,1)$.
Throughout this section we let $G = (V_G=L\cup R,E_G)$ be a random bipartite graph distributed according to $\Gr(m,p,d-1,1)$; that is, $\din=d-1$ and $\dout = 1$.
The random graph $G$ can be equivalently generated as follows.
\begin{lemma}
	\label{lemma:equivalent-gen-2}
	Let $m,p,d \in\N^+$ be positive integers such that
	$m \ge p$ and $d \ge 3$.
	Let $G' = (V_G=L\cup R,E_{G'})$ be the random bipartite graph generated as follows:
	\begin{enumerate}
		\item Let $M_1,M_2,\dots,M_{d}$ be $d$ random perfect matchings between $L$ and $R$;
		\item Let $P_1$ be a subset of $L$ chosen uniformly at random among all the subsets of $L$ such that $|P_1|=p$;
		\item Let $P_2 \subset R$ be the set of vertices in $R$ that are matched to $P_1$ in $M_d$, and let $A \subset M_d$ be the set of edges between $P_1$ and $P_2$;
		\item Let $P = P_1 \cup P_2$ and $E_{G'} = \bigcup_{i=1}^d M_i \setminus A$.
	\end{enumerate}
	Then $G'$ has distribution $\Gr(m,p,d-1,1)$.
\end{lemma}
Additionally, the edge expansion of the random graph $G$ satisfies the following bounds.
\begin{lemma}\label{lem:small-expansion}
	For $3\le d = O(1)$, $\alpha \in (0,1)$ and $\delta>0$, there exists $\varepsilon>0$ such that with probability $1-o(1)$:
	$$
	\min_{\substack{S\subset V_G:\\0<|S|\leq \varepsilon m}} \frac{|E(S,V_G\backslash S)|}{|S|} > d-2-\alpha-\delta.
	$$
\end{lemma}
\begin{lemma}\label{lem:small-expansion-P}
	For $3\le d = O(1)$, $\alpha \in (0,\frac 14]$, $\delta>0$ and  $\xi\in(0,1)$, it holds with probability $1-o(1)$:
	$$
	\min_{\substack{S\subset V_G:\\0<\xi|S|\leq |P\cap S|}} \frac{|E(S,V_G\backslash S)|}{|S|} > \xi(d-\alpha-\delta)-1.
	$$
\end{lemma}
We are now ready to prove Theorems~\ref{thm:expansion-generic} and~\ref{thm:expansion-port}.
\begin{proof}[Proof of Theorem~\ref{thm:expansion-generic}]
	By Lemma~\ref{lem:small-expansion},	for $\alpha \in (0,\frac 14]$ and $\delta \le \frac 14$ there exists $\varepsilon>0$ such that with probability $1-o(1)$:
	$$
	\min_{\substack{S\subset V_G:\\0<|S|\leq \varepsilon m}} \frac{|E(S,V_G\backslash S)|}{|S|} > d-\frac 52 \geq \frac{d}{6}.
	$$
	For $S\subset V_G$ with $\eps m<|S|\leq m$, we consider
	the random bipartite graph $\hat{G}$ obtained
	as the union of $d$ random perfect matchings $M_1,\dots,M_d$ between $L$ and $R$.
	Let $\lambda_2(\hat{G})$ denote the second largest eigenvalue of the adjacency matrix of $\hat{G}$.
	Lemma \ref{lemma:spectral} implies that for any constant $\delta>0$ (independent of $m$), with probability $1-o(1)$, the following holds:
	$$
	\lambda_2(\hat{G}) < 2\sqrt{d-1}+\delta.
	$$
	Now, let $\hat{G}=(V_{G},E_{\hat{G}})$ and let $\hat{E}(S,V_{{G}} \setminus S)$
	be the set of edges between $S$ and $V_{{G}}\setminus S$ in $\hat{G}$.
	Cheeger's inequality implies
	$$
	\min_{\substack{S\subset V_G:\\0<|S|\leq m}} \frac{|\hat{E}(S,V_{{G}}\backslash S)|}{|S|} \geq \frac{d-\lambda_2(\hat{G})}{2} > \frac{d}{2}-\sqrt{d-1}-\frac{\delta}{2} \ge \frac{d}{40},
	$$
	where the last inequality holds for $\delta \le 0.01$.
	
	We use this bound on the edge expansion of $\hat{G}$ to deduce a bound for the edge expansion of $G$.
	First note that	by Lemma~\ref{lemma:equivalent-gen-2}, $G$ can be obtained from $\hat{G}$ as follows:
	\begin{enumerate}
		\item Choose $P_1 \subset L$ uniformly at random among all the subsets of $L$ of size $p$;
		\item Let $P_2 \subset R$ be the set of vertices matched to $P_1$ in $M_d$, and let $A \subset M_d$
		be the set of edges between $P_1$ and $P_2$;
		\item Set $P = P_1 \cup P_2$ and $E_G = \bigcup_{i=1}^d M_i \setminus A$;
		\item Replace all the multiedges in $E_G$ by single edges.
	\end{enumerate}
	Moreover, since $3\leq d = O(1)$, Lemma \ref{lemma:multiedges} implies there exists a constant $\kappa$ such that with probability $1-O(m^{-1})$ the multiplicity of every edge  in $\hat{G}$ is at most $\kappa$. Hence,
	in order to obtain $G$ from $\hat{G}$, $p$ edges are removed and the multiplicity of every edge may decrease by a factor of at most $\kappa$. Therefore,
	for every $S \subset V_G$
	$$
	E(S,V_G \setminus S) \ge \frac{\hat{E}(S,V_G \setminus S) - p}{\kappa},
	$$
	and
	\begin{align*}
	\min_{\substack{S\subset V_G:\\\eps m<|S|\leq m}} \frac{|E(S,V_G\backslash S)|}{|S|}
	\geq \min_{\substack{S\subset V_{{G}}:\\\eps m<|S|\leq m}} \frac{|\hat{E}(S,V_{{G}}\setminus S)| - p}{\kappa |S|}
	\ge \frac{d}{40 \kappa}- \frac{1}{\eps \kappa m^{3/4}},
	\end{align*}
	and the result is established.
\end{proof}	
\begin{proof}[Proof of Theorem~\ref{thm:expansion-port}]
	Let $S\subset V_G$ such that $0 < |S| \le m$.
	Let $\hat{d} = d - \alpha -\delta$, $\delta = 0.01$ and let $\varepsilon>0$ be the constant from Lemma~\ref{lem:small-expansion}.
	We consider three cases, depending on the sizes of $S$ and $S \cap P$.
	Suppose first that
	$|S|>\varepsilon m$. Then, Theorem~\ref{thm:expansion-generic} implies that with probability $1-o(1)$, there exists $\gamma' > 0$ such that
	$$
	|E(S,V\backslash S)| > \gamma' |S| > \gamma' \varepsilon m \ge  (1+\gamma)m^\alpha > (1+\gamma)|P\cap S|,
	$$
	where the second to last inequality holds for sufficiently large $m$ and a suitable constant $\gamma > 0$.
	
	For the second case, suppose that $|S|\leq\varepsilon m$ and $|P\cap S|<\xi |S|$, where
	$$
	\xi =  \frac{\sqrt{4\hat{d}(\hat{d}-2)+1}+1}{2\hat{d}}.
	$$
	Then, by Lemma~\ref{lem:small-expansion}, with probability $1-o(1)$:
	$$
	|E(S,V\backslash S)| > (\hat{d}-2)|S| > \frac{\hat{d}-2}{\xi}|P\cap S| = \frac{\sqrt{4\hat{d}(\hat{d}-2)+1}-1}{2} |P\cap S| \ge (1+\gamma)|P\cap S|,
	$$
	where the last inequality holds for sufficiently small $\gamma$ since $\hat{d} \ge 2.74$.
	
	Finally, suppose  $|S|\leq\varepsilon m$ and $|P\cap S|\geq \xi |S|$. In this case, Lemma \ref{lem:small-expansion-P} also implies that with probability $1-o(1)$:
	$$
	|E(S,V\backslash S)| > \left( \xi \hat{d}-1 \right)|S| = \frac{\sqrt{4\hat{d}(\hat{d}-2)+1}-1}{2} |P\cap S| \ge (1+\gamma)|P\cap S|,
	$$
	and the result follows.	
\end{proof}

\subsection{Gadget expansion when $d = O(1)$: auxiliary facts}
In this section, we give the proof of our two key bounds on the edge expansion of the random graph $\Gr(m,p,{d-1},1)$.
In particular, we prove Lemmas~\ref{lem:small-expansion} and~\ref{lem:small-expansion-P}.
Suppose $3 \le d = O(1)$,
$p = \floor{m^\alpha}$ for some constant	
$\alpha \in (0,1)$,
and let  $G = (V_G=L\cup R,E_G)$ be a random bipartite graph distributed according to $\Gr(m,p,d-1,1)$.
Our proofs of Lemmas~\ref{lem:small-expansion} and~\ref{lem:small-expansion-P} will be based on the following bound for the
\textit{vertex expansion} of small subsets of $L$ (or $R$).
Recall that the vertex expansion is defined as:
$$
\partial S = \{ v\in V_G\backslash S: \exists u\in S, \{u,v\}\in E_G\}.
$$
\begin{lemma}\label{lem:small-vx-expansion}
	For $3\le d = O(1)$, $\alpha \in (0,1)$ and $\delta>0$, there exists $\varepsilon>0$ such that with probability $1-o(1)$ the following holds:
	$$
	\min_{\substack{S\subset L:\\0<|S|\leq \varepsilon m}} \frac{|\partial S|}{|S|} > d-1-\alpha-\delta.
	$$	
\end{lemma}
The proof of this lemma has similar flavor to that of Theorem 4.16 in~\cite{HLW} for random regular bipartite graphs, and it is provided in Section \ref{subsection:vertex-exp}.
\begin{proof}[Proof of Lemma~\ref{lem:small-expansion}]
	By Lemma~\ref{lem:small-vx-expansion}, for $3 \le d = O(1)$, $\alpha \in (0,1)$ and $\delta>0$ there exists $\varepsilon>0$ such that for all
	$T \subset V_G$ with $0<|T|\leq \varepsilon m$ and either $T\subset L$ or $T\subset R$, with probability $1-o(1)$ we have
	$$
	|\partial T| > (d-1-\alpha-\delta)|T|.
	$$
	Suppose this holds for every such $T$. Let $S\subset V_G$ such that $0<|S|\leq \varepsilon m$ and
	let $S_L = S\cap L$ and $S_R = S\cap R$. Then, $S_L\subset L$, $S_R\subset R$ and $\max\{|S_L|,|S_R|\} \leq |S|\leq \varepsilon m$. Hence,
	\begin{align*}
	|E(S,V\backslash S)| &\ge |\partial S| \\
	&= |\partial S_L \backslash S_R| + |\partial S_R \backslash S_L|\\
	&\geq |\partial S_L| - |S_R| + |\partial S_R| - |S_L|\\
	&> (d-1-\alpha-\delta) |S_L| - |S_R| + (d-1-\alpha-\delta) |S_R| - |S_L|\\
	&= (d-2-\alpha-\delta) |S|.
	\end{align*}
	Therefore, this holds for all such $S$ with probability $1-o(1)$ and the result follows.
\end{proof}
To prove Lemma~\ref{lem:small-expansion-P} we also need the following fact.
\begin{fact}
	\label{fact:port-edges}
	For $3 \le d = O(1)$ and $\alpha \in (0,\frac 14]$, the event $E(L \cap P, R\cap P) = \emptyset$ occurs with probability $1-O(m^{-1/2})$.
\end{fact}
\begin{proof}
	By definition,
	no edges are added between $L \cap P$ and $R\cap P$ in  $M_1'$. Thus,
	it is sufficient to show that $E(L \cap P, R\cap P) = \emptyset$ after adding the $d-1$ random perfect matchings $M_1,\dots,M_{d-1}$. We may generate a random matching by choosing for each vertex of $L$, in any order, a uniformly random  unmatched vertex in $R$.
	Say $M_1,\dots,M_{d-1}$ are generated in this manner always matching the vertices in $L \cap P$ first.
	Then, the probability of the event  $E(L \cap P, R\cap P) = \emptyset$ is:
	\begin{align*}
	\left[\prod_{i=0}^{p-1} \left(1 - \frac{p}{m-i}\right)\right]^{d-1} \ge \left(1 - \frac{p}{m-p+1}\right)^{p(d-1)} \ge 1 - O(m^{-1/2}).\tag*{\qedhere}
	\end{align*}
\end{proof}
We now have all the ingredients for the proof of Lemma~\ref{lem:small-expansion-P}.
\begin{proof}[Proof of Lemma~\ref{lem:small-expansion-P}]
	By Lemma~\ref{lem:small-vx-expansion},  for $3 \le d = O(1)$ and $\delta>0$ there exists $\varepsilon>0$ such that, with  probability $1-o(1),$ for all $S \subset V$ with $0<|S|\leq \varepsilon m$ and either $S\subset L$ or $S\subset R$,
	$$
	|\partial S| > \left( d-1-\alpha-\delta \right)|S|.
	$$
	Also, by Fact~\ref{fact:port-edges}, $P$ is an independent set with probability $1-O(m^{-1/2})$.
	Hence, by a union bound, both of these events occur with probability $1-o(1)$. Suppose this is the case.
	
	For $\xi\in(0,1)$, let $S\subset V$ such that $0<\xi |S|\leq |P\cap S|$. Then, for sufficiently large $m$
	$$
	|S| \leq \frac{|P\cap S|}{\xi} \leq \frac{2p}{\xi} \leq \eps m.
	$$
	Let $S_L=S\cap L$ and $S_R=S\cap R$. Since there is no edge between any pair of vertices in $P$, we have
	$$
	\partial S_L \backslash S_R \supset \partial (P\cap S_L) \backslash S_R = \partial (P\cap S_L) \backslash (S_R\backslash P),
	$$
	and similarly $\partial S_R \backslash S_L \supset \partial (P\cap S_R) \backslash (S_L\backslash P)$. Moreover, $S_L\subset L$, $S_R\subset R$ and $\max\{|S_L|,|S_R|\} \leq |S|\leq \varepsilon m$. It follows that
	\begin{align*} 	
	|E(S,V\backslash S)|  &\ge
	|\partial S| \\
	&= |\partial S_L \backslash S_R| + |\partial S_R \backslash S_L|\\
	&\geq |\partial (P\cap S_L) \backslash (S_R \backslash P)| + |\partial (P\cap S_R) \backslash (S_L \backslash P)|\\
	&\geq |\partial (P\cap S_L)| - |S_R \backslash P| + |\partial (P\cap S_R)| - |S_L \backslash P|\\
	&> (d-1-\alpha-\delta) |P\cap S_L| + (d-1-\alpha-\delta)|P\cap S_R| - |S\backslash P|\\
	&= (d-\alpha-\delta) |P\cap S| - |S|\\
	&\geq \left( \xi(d-\alpha-\delta) - 1 \right) |S|.
	\end{align*}
	Thus, this holds for all $S\subset V$ such that $0<\xi |S|\leq |P\cap S|$ with probability $1-o(1)$ as claimed.
\end{proof}

\subsubsection{Vertex expansion for the Ising gadget when $d=O(1)$}
\label{subsection:vertex-exp}

In this section, we prove our lower bound in Lemma~\ref{lemma:equivalent-gen-2} for the \emph{vertex expansion}
of a random bipartite graph $G=(V_G=L \cup R,E_G)$ with distribution $\Gr(m,p,d-1,1)$
when $p = \floor{m^\alpha}$ for some $\alpha \in (0,\frac 14]$.
As mentioned earlier, our proof of this bound follows closely the approach in \cite{HLW}.
We also need the following fact regarding yet another equivalent way of generating a graph with distribution $\Gr(m,p,d-1,1)$; see also Lemma~\ref{lemma:equivalent-gen-2}.
\begin{lemma}
	\label{lemma:equivalent-gen-1}
	Let $m,p,d \in\N^+$ be positive integers such that
	$m \ge p$ and $d \ge 3$.
	Let $G'' = (V_G=L\cup R,E_G)$ be the random bipartite graph generated as follows:
	\begin{enumerate}
		\item Let $M_1,M_2,\dots,M_{d-1}$ be $d-1$ random perfect matchings between $L$ and $R$;
		\item Let $P_1$ be a subset of $L$ chosen uniformly at random among all the subsets of $L$ such that $|P_1|=p$;
		\item Let $M_1'$ be a random complete matching between $L\backslash P_1$ and $R$, and let
		$P_2 \subset R$ be the set of unmatched vertices in $R$; hence  $|P_2|=p$.
		\item Let $P = P_1 \cup P_2$ and $E_{G''} = \left(\bigcup_{i=1}^{d-1} M_i\right) \bigcup M_1'$.
	\end{enumerate}
	Then $G''$ has distribution $\Gr(m,p,d-1,1)$.
\end{lemma}
Both Lemmas~\ref{lemma:equivalent-gen-1} and \ref{lemma:equivalent-gen-2} are proved in Section \ref{subsection:alternative-gen},
\begin{proof}[Proof of Lemma~\ref{lem:small-vx-expansion}]
	Let
	$$
	\eta = d-1-\alpha-\delta.
	$$
	For $S\subset L$ and $T\subset R$, let $X_{S,T}$ be an indicator random variable for the event $\partial S\subseteq T$. For some $\eps\in(0,1/\eta)$ to be chosen later, let
	$$
	X = \sum_{\substack{S\subset L:\\0<|S|\leq \eps m}} \sum_{\substack{T\subset R:\\|T|=\floor{\eta|S|}}} X_{S,T}.
	$$
	Since every set $T \subset R$ of size less than $\floor{\eta|S|}$ is included in some subset of $R$ of size exactly
	$\floor{\eta|S|}$, it suffices to show that
	$$
	\Pr[X>0] \leq O\left( \frac{(\ln m)^\delta}{m^\delta} \right).
	$$
	Write $|S|=s$, $|T|=t$ and $|P\cap S|=r$. Then, Lemma~\ref{lemma:equivalent-gen-1} implies
	$$
	\Pr\left[ X_{S,T}=1 \right] = \left( \frac{t(t-1)\dots(t-s+1)}{m(m-1)\dots(m-s+1)} \right)^{d-1} \left( \frac{t(t-1)\dots(t-(s-r)+1)}{m(m-1)\dots(m-(s-r)+1)} \right) \leq \left( \frac{t}{m} \right)^{ds-r}.
	$$
	Using  Markov's inequality, we get
	\begin{align*}
	\Pr\left[ X > 0 \right] = \Pr\left[ X\geq 1 \right] &\leq \E[X],
	\end{align*}
	and so
	\begin{align*}
	\Pr\left[X > 0 \right] &\le \sum_{\substack{S\subset L:\\0<|S|\leq \eps m}} \sum_{\substack{T\subset R:\\|T|=\floor{\eta|S|}}} \E [X_{S,T}]\\
	&\leq \sum_{s=1}^{\floor{\eps m}}   \sum_{r=0}^s   \binom{p}{r}\binom{m-p}{s-r} \binom{m}{t} \left(\frac{t}{m}\right)^{ds-r}.
	\end{align*}
	From the inequality $\binom{m}{k}\leq \left(\frac{\e m}{k}\right)^k$, we get
	\begin{align*}
	\Pr\left[ X > 0 \right] &\leq \sum_{s=1}^{\floor{\eps m}}  \sum_{r=0}^s   \left(\frac{\e p}{r}\right)^r \left(\frac{\e (m-p)}{s-r}\right)^{s-r} \left(\frac{\e m}{t}\right)^t \left(\frac{t}{m}\right)^{ds-r}\\
	&= \sum_{s=1}^{\floor{\eps m}} \sum_{r=0}^s  \Big[ (\e p)^r (\e(m-p))^{s-r} \Big] \cdot \left[ \left(\frac{1}{r}\right)^r\left(\frac{1}{s-r}\right)^{s-r} \right] \cdot \left[ \left(\frac{\e m}{t}\right)^t \left(\frac{t}{m}\right)^{ds-r} \right].
	\end{align*}
	Since $p \le m^\alpha$, the first term is be bounded by
	\begin{equation}\label{eqn:expansion-term1}
	(\e p)^r (\e(m-p))^{s-r} \leq \e^s p^r m^{s-r} = \left( \e m^{1-\frac {(1-\alpha)r}{s}} \right)^s.
	\end{equation}
	Also, the AM-GM inequality yields
	\begin{equation}\label{eqn:expansion-term2}
	\left(\frac{1}{r}\right)^r\left(\frac{1}{s-r}\right)^{s-r} \leq \left( \frac{r\cdot\frac{1}{r} + (s-r)\cdot\frac{1}{s-r}}{s} \right)^s = \left(\frac{2}{s}\right)^s.
	\end{equation}
	Finally, since for $\varepsilon \in (0,1/\eta)$, $t \leq \eta s \le \eta \varepsilon m < m$ we have
	\begin{equation}\label{eqn:expansion-term3}
	\left(\frac{\e m}{t}\right)^t \left(\frac{t}{m}\right)^{ds-r} \leq \left(\frac{\e m}{t}\right)^{\eta s} \left(\frac{t}{m}\right)^{ds-r} = \left[ \e^{\eta} \left(\frac{t}{m}\right)^{d-\frac rs-\eta} \right]^s \leq \left[ \e^{\eta} \left(\frac{\eta s}{m}\right)^{d-\frac rs-\eta} \right]^s.
	\end{equation}
	Combining the inequalities \eqref{eqn:expansion-term1}, \eqref{eqn:expansion-term2} and \eqref{eqn:expansion-term3} we deduce
	\begin{align*}
	\Pr\left[ X > 0 \right]
	&\leq \sum_{s=1}^{\floor{\eps m}} \sum_{r=0}^s\left[ \e m^{1-(1-\alpha)\frac rs} \cdot \frac{2}{s} \cdot \e^{\eta} \left(\frac{\eta s}{m}\right)^{d-\frac rs-\eta} \right]^s\\
	&= \sum_{s=1}^{\floor{\eps m}} \sum_{r=0}^s \left[ 2\e^{\eta+1}\eta^{d-\frac rs-\eta} \cdot \left(\frac{s}{m}\right)^{d-1-\alpha \frac rs-\eta} \cdot s^{-(1-\alpha)\frac rs} \right]^s.
	\end{align*}
	We recall that $\eta = d-1-\alpha-\delta$, so for $0\leq r\leq s$ we have
	\begin{equation*}
	\eta^{d-\frac rs-\eta} \leq (d-1)^{d-\eta}, \quad\quad \left(\frac{s}{m}\right)^{d-1-\alpha \frac rs-\eta} \leq \left(\frac{s}{m}\right)^\delta \quad\text{and}\quad s^{-(1-\alpha)\frac rs}\leq 1.
	\end{equation*}
	Thus, taking $c=c(d,\alpha,\delta) = 2  \e^{\eta+2}(d-1)^{d-\eta}$ we obtain
	\begin{align*}
	\Pr\left[X > 0 \right]
	&\le \sum_{s=1}^{\floor{\eps m}} \sum_{r=0}^s \left[ 2\e^{\eta+1} (d-1)^{d-\eta} \cdot \left(\frac{s}{m}\right)^{\delta} \right]^s\\
	&\le \sum_{s=1}^{\floor{\eps m}} (s+1) \left[ \frac{c}{{\e}} \cdot \left(\frac{s}{m}\right)^\delta \right]^s \\
	&\leq \sum_{s=1}^{\floor{\eps m}} \left[ c \left(\frac{s}{m}\right)^\delta\right]^s,
	\end{align*}
	where in the last inequality we use the fact that $(s+1)^{1/s}\leq {\e}$ for all $s> 0$.
	
	Now given $\delta>0$, since $c \le 2 d^{2+\delta} {\e}^{d+1}$, we can choose $\eps\in(0,1/\eta)$ such that
	$
	c \eps^\delta < {\e}^{-1}.
	$
	For this choice of $\varepsilon$ we have
	$$
	\sum_{s=\floor{\delta\ln m}+1}^{\floor{\eps m}} \left[ c \left(\frac{s}{m}\right)^\delta\right]^s \leq \sum_{s=\floor{\delta\ln m}+1}^{\floor{\eps m}} \left[ c \eps^\delta\right]^s \leq \sum_{s=\floor{\delta\ln m}+1}^{\floor{\eps m}} \frac{1}{\e^s} = O\left( \frac{1}{m^\delta} \right)
	$$
	and
	$$
	\sum_{s=1}^{\floor{\delta\ln m}} \left[ c \left(\frac{s}{m}\right)^\delta\right]^s \leq \sum_{s=1}^{\floor{\delta\ln m}} \left[ c \left(\frac{\delta\ln m}{m}\right)^\delta\right]^s = O\left( \frac{(\ln m)^\delta}{m^\delta} \right).
	$$
	Hence,
	$$
	\Pr\left[ X > 0 \right]  \leq O\left( \frac{(\ln m)^\delta}{m^\delta} \right)
	$$
	and the lemma follows.
\end{proof}
\subsection{Ising gadget: equivalent generation}
\label{subsection:alternative-gen}
We previously stated, in Lemmas~\ref{lemma:equivalent-gen-2} and~\ref{lemma:equivalent-gen-1}, two alternative procedures to generate a random graph $G$ with distribution $\Gr(m,p,d-1,1)$. We conclude this section with a proof of these facts.
\begin{proof}[Proof of Lemmas~\ref{lemma:equivalent-gen-2} and~\ref{lemma:equivalent-gen-1}]
	Denote the random bipartite graph defined in Lemma~\ref{lemma:equivalent-gen-2} by $G'$ and the one in Lemma~\ref{lemma:equivalent-gen-1} by $G''$. We need to show that both $G'$ and $G''$ have distribution $\Gr(m,p,d-1,1)$. Recall that $P_1=P\cap L$, $P_2=P\cap R$ and $M'_1$ is the perfect matching between $L\backslash P_1$ and $R\backslash P_2$. By the definitions of $\Gr(m,p,d-1,1)$, $G'$ and $G''$, it suffices to show that the joint distributions of $(P_2, M'_1)$ in these three models are the same. We recall that:
	\begin{itemize}
		\item In $\Gr(m,p,d-1,1)$, the joint distribution $\rho$ of $(P_2, M'_1)$ is:
		\begin{enumerate}
			\item $P_2$ is a subset of $R$ chosen uniformly at random among all the subsets of $R$ such that $|P_2|=p$;
			\item $M'_1$ is a random perfect matching between $L\backslash P_1$ and $R\backslash P_2$.
		\end{enumerate}
		\item In $G'$, the joint distribution $\rho'$ of $(P_2, M'_1)$ is:
		\begin{enumerate}
			\item $M_d$ is a random perfect matching between $L$ and $R$;
			\item $P_2 \subset R$ is the set of vertices in $R$ that are matched to $P_1$;
			\item $A \subset M_d$ is the set of edges between $P_1$ and $P_2$, and let $M'_1=M_d\backslash A$.
		\end{enumerate}
		\item In $G''$, the joint distribution $\rho''$ of $(P_2, M'_1)$ is:
		\begin{enumerate}
			\item $M'_1$ is a random complete matching between $L\backslash P_1$ and $R$;
			\item $P_2 \subset R$ is the set of unmatched vertices in $R$.
		\end{enumerate}
	\end{itemize}
	
	We first show that $\rho''=\rho$. In $\rho''$, the set $P_2$ of unmatched vertices in $R$ is a uniformly random subset of $R$ over all subsets such that $|P_2|=p$. Also, given $P_2\subset R$, a random complete matching between $L\backslash P_1$ and $R$ conditioned on that vertices in $P_2$ are unmatched is a random perfect matching between $L\backslash P_1$ and $R\backslash P_2$. This gives $\rho''=\rho$.	
	To see that $\rho'=\rho''$, we observe that a random complete matching between $L\backslash P_1$ and $R$ can be obtained by first drawing a random perfect matching between $L$ and $R$ and then removing all edges incident to $P_1$.
\end{proof}

\section{Identity testing algorithm for the ferromagnetic Ising model}
\label{sec:alg}

Let $G = (V,E)\in\mathcal{M}(n,d)$ be an $n$-vertex graph of maximum degree at most $d$. In this section, we 
focus on the ferromagnetic (attractive) setting.
We will allow each edge to have distinct but \emph{positive} interaction parameter which may depend on $n$.
In setting, $\beta = \{\beta(v,w)\}_{\{v,w\} \in E}$ with $\beta(v,w) > 0$, and
the Gibbs distribution becomes
\[
\mu_{G,\beta}(\sigma) = \frac{1}{Z_{G,\beta}} \exp\left( \sum_{\{v,w\}\in E} \beta(v,w)\Ind\{\sigma(v) = \sigma(w)\} \right),  
\]
for every $\sigma \in \{+,-\}^V$; cf., \eqref{eq:ising-def}. 
With slight abuse of notation, we also use $\beta$ for the largest $\beta(v,w)$; i.e., $\beta = \max_{\{v,w\} \in E} \beta(v,w)$.
We remark that the Ising model is also well-defined when $G$ is a multigraph. 
Indeed, an Ising model on a multigraph can be transformed into an equivalent model on a simple graph by collapsing all parallel edges and setting $\beta(v,w)$ to be the sum of the weights of all the edges between $v,w \in V$.
We restrict attention again in this section to the simpler case where there is no external magnetic field. 
This simplification is not actually necessary 
and is done only for the sake of clarity in our proofs; for a discussion about how our algorithmic results extend
to models with external field see Remark~\ref{rmk:field}.


In \cite{DDK}, the authors give an algorithm for identity testing for $\mathcal{M}(n,d)$ (see Algorithm~2 in \cite{DDK}).
We call this algorithm the \emph{DDK algorithm}. 
As discussed in the introduction, the running time and sample complexity of this algorithm 
for the ferromagnetic Ising model, 
where we can sample in polynomial time, 
is $\poly(n,d,\beta,\eps^{-1})$. 
We provide here an algorithm whose running time and sample complexity is polynomial in $n$, $d$ and $\eps^{-1}$ but independent of $\beta$. 

Our algorithm will use as a subroutine the DDK algorithm for multigraphs,
which extends straightforwardly to this more general setting.
We will also use the fact that we can generate samples from the ferromagnetic Ising distribution in polynomial time; see~\cite{JS:Ising,RW,GuoJ,CGHT} for various methods. 
These two facts are rigorously stated in the following theorems.
For positive integers 
$n$ and $m$, let $\mathcal{M}_\mathrm{multi}(n,m)$ denote the family of all $n$-vertex multigraphs with at most $m$ edges. 

\begin{thm}
	\label{thm:sample-ferro-Ising}
	Let $H\in \mathcal{M}_\mathrm{multi}(n,m)$ where $n$, $m$ are positive integers.
	Then, for all $\beta,\delta>0$,  there exists an algorithm that generates a sample from a distribution $\mu_{\textsc{alg}}$ satisfying:
	\[
	\TV{\mu_{H,\beta}}{\mu_{\textsc{alg}}} \leq \delta,
	\]
	with running time $\poly(m,\log(1/\delta))$. 
\end{thm}


\begin{thm}[\cite{DDK}]\label{thm:DDK}
	The DDK algorithm for the identity testing problem in $\mathcal{M}_\mathrm{multi}(n,m)$ has sample complexity ${O}(m^2\beta^2 \eps^{-2}\log n)$, running time 
	$\poly(m,\beta,\eps^{-1})$ and success probability at least~$4/5$.
\end{thm}

Before presenting our algorithm, we first introduce some necessary notations and definitions.
For a set of vertices $A \subset V$, let $\binom{A}{2}$ denote the collection of all pairs of vertices in $A$; i.e., 
\[
\binom{A}{2} = \{\{v,w\}: v,w\in A, v\neq w\}.
\]
Suppose $P = \{C_1,\dots,C_k\}$ is a partition of $V$; that is, $\cup_{i=1}^k C_i = V$ and $C_i\cap C_j = \emptyset$ for $1\leq i<j\leq k$. Let
\[
E(P) = \bigcup_{i=1}^k \binom{C_i}{2}.
\]
The \textit{quotient graph} $G_{\textsc{p}} = (V(G_{\textsc{p}}),E(G_{\textsc{p}}))$ is a multigraph defined as follows:
\begin{enumerate}
	\item Every vertex of $G_{\textsc{p}}$ is a partition class $C_i$ from $P$ where $1\leq i\leq k$; i.e., $V(G_{\textsc{p}}) = P$.
	\item Every edge between $C_i$ and $C_j$ of $G_{\textsc{p}}$ corresponds to an edge $\{v,w\}$ of $G$ where $v\in C_i$ and $w\in C_j$, allowing parallel edges.  The number of edges between $C_i$ and $C_j$ in $G_{\textsc{p}}$ is equal to the size of the set $\{\{v,w\} \in E: v\in C_i,\, w\in C_j \}$.
\end{enumerate}

Observe also that there is a one-to-one correspondence between the edge sets $E(G_{\textsc{p}})$ and $E \setminus E(P)$, which we represent by the bijective map $\varphi: E(G_{\textsc{p}}) \to E\backslash E(P)$.
Using $\varphi$, we can define an Ising model $(G_{\textsc{p}},\beta_{\textsc{p}})$ on the quotient graph $G_{\textsc{p}}$; the parameter $\beta_{\textsc{p}}$ is given by
\[
\beta_{\textsc{p}} (e) = \beta(\varphi(e))
\]
for every $e \in E(G_{\textsc{p}})$. 
We call $(G_{\textsc{p}},\beta_{\textsc{p}})$ the \emph{quotient model}. 

Suppose $\tau\in \{+,-\}^V$ is an Ising configuration of the original model $(G,\beta)$ satisfying $\tau(v) = \tau(w)$ for all $\{v,w\} \in E(P)$; that is, every vertex in the same partition class $C_i$ shares the same spin. Then, we can define a corresponding configuration $\tau_{\textsc{p}} \in \{+,-\}^P$ of the quotient model $(G_{\textsc{p}},\beta_{\textsc{p}})$ as follows. For each $C_i\in P$ and $v\in C_i$
\[
\tau_{\textsc{p}} (C_i) = \tau(v);
\]
our assumption on $\tau$ guarantees that the configuration $\tau_{\textsc{p}}$ is well-defined.

Recall that in the identity testing problem for $\mathcal{M}(n,d)$, 
we are given a graph $G=(V,E) \in \mathcal{M}(n,d)$,
the parameter $\beta$ 
and sample access to an unknown Ising distribution on a graph ${G}^* \in \mathcal{M}(n,d)$.
We will reduce this problem to identity testing for the corresponding quotient models.

For ease of notation, we set $\mu=\mu_{G,\beta}$, $\mu_{\textsc{p}}=\mu_{G_{\textsc{p}},\beta_{\textsc{p}}}$, 
${\mu}^*=\mu_{{G}^*,{\beta}^*}$ and $\mu^*_{\textsc{p}}={\mu}_{G^*_{\textsc{p}},\beta^*_{\textsc{p}}}$.
For a partition $P$ of $V$, 
define $\mathcal{P}$ to be the event that vertices from the same partition class of $P$ receive the same spin; i.e.,
\[
\mathcal{P} =  \{ X_v = X_w, \,\forall \{v,w\} \in E(P) \},
\]
where recall that $X_v ,X_w \{+1,-1\}$ are the random variables for the spins at vertices $v$ and $w$ respectively.
Interchangeably, we also use $\mathcal{P}$ for the set
\begin{align*}
\mathcal{P}  &= \left\{ \sigma\in\{+,-\}^V: \sigma(v) = \sigma(w), \,\forall \{v,w\} \in E(P) \right\}.
\end{align*}
We remark that the conditional distributions $\mu(\cdot|\mathcal{P})$ and ${\mu}^*(\cdot|\mathcal{P})$ are equivalent to the Gibbs distributions $\mu_{\textsc{p}}$ and ${\mu}^*_{\textsc{p}}$, respectively, for the quotient models.

Given $L$ samples $\{\sigma_1,\dots,\sigma_L\}$ from $\mu$, we define $\muemp$ to be the empirical distribution of these samples; in particular,
\[
\muemp(\mathcal{P}) = \frac{1}{L}\, \sum_{i=1}^L \Ind\left\{ \sigma_i \in \mathcal{P} \right\} 
\]
Observe that $\muemp(\mathcal{P}) = 1$ if and only if in all the $L$ samples every vertex from the same partition class of $P$ has the same spin. Similarly, given $L$ samples $\{\tau_1,\dots,\tau_L\}$ from ${\mu}^*$, we also define the empirical distribution $\muemph$ and the empirical probability $\muemph(\mathcal{P})$.

\begin{algorithm}[t]
	\SetKwInOut{Input}{input}\SetKwInOut{Output}{output}
	\Input{An Ising model $(G,\beta)$, $L$ samples $\{\tau_i\}_{i=1}^L$ from an unknown Ising model 
		$({G}^*,{\beta}^*)$ and a parameter $\eps > 0$.}
	\Output{\textsc{Yes} if it regards$\{\tau_i\}_{i=1}^L$ as samples from  $\mu_{G,\beta}$; \\ \textsc{No} if it regards $\{\tau_i\}_{i=1}^L$ as samples from $\mu_{G^*,\beta^*}$ such that~${\|\mu_{G,\beta} - \mu_{G^*, \beta^*}\|}_{\textsc{tv}}> \eps$.}
	\BlankLine
	
	$P = \{V\}$\;\label{algstep:find-P}
	
	\For{$i \leftarrow 1$ \KwTo $L$}{
		\ForEach{$C \in P$}{
			$C_+ = \{v\in C: \tau_i(v) = +\}$\;
			$C_- = \{v\in C: \tau_i(v) = -\}$\;
			$P \leftarrow P \backslash \{C\} \cup \{C_+,C_-\}$\;\label{algstep:find-P-2}
		}
	}
	
	\medskip
	Generate $L$ independent $(\eps/16)$-approximate samples $\{\sigma_i\}_{i=1}^L$ from $\mu_{G,\beta}$\;\label{algstep:sample-from-G}
	\If{$\muemp(P) \leq 1-\eps/4$\label{algstep:bad-event1-if}}{
		return \textsc{No}\;\label{algstep:bad-event1-then}
	}
	
	\medskip
	\ForEach{$\{v,w\} \in E\backslash E(P)$\label{algstep:bad-event2-for}}{
		\If{$\beta(v,w) \geq \ln(20n^2L)$\label{algstep:bad-event2-if}}{
			return \textsc{No}\;\label{algstep:bad-event2-then}
		}	
	}
	
	\medskip
	Run the DDK algorithm on the quotient model $(G_{\textsc{p}},\beta_{\textsc{p}})$, with samples $\{(\tau_i)_{\textsc{p}}\}_{i=1}^L$ and parameter $\eps' = \eps/2$.\label{algstep:DDK}
	
	Return the output of the DDK algorithm.
	
	\caption{Identity testing for ferromagnetic Ising models}\label{alg-testing}
\end{algorithm}

Suppose we are given a known Ising model $(G,\beta)$, $L$ samples $\{\tau_1,\ldots,\tau_L\}$ from an unknown Ising model $({G}^*,\beta^*)$ and a parameter $\eps>0$, where $G,G^*\in\mathcal{M}(n,d)$ and $\beta,\beta^*>0$.
Our algorithm (see Algorithm~\ref{alg-testing}) 
tests whether $(G,\beta) = (G^*,\beta^*)$ or ${\|\mu_{G,\beta}-\mu_{G^*,\beta^*}\|}_\textsc{tv} > \eps$.
For this, it first identifies ``heavy'' edges.
These are the edges $\{v,w\}$ whose interaction parameter $\beta(v,w)$ is very large,
and thus its endpoints $v$ and $w$ are likely to have the same spin.
To identify the heavy edges, the algorithm looks for the coarsest partition $P$ of $V$ such that 
for every partition class $C$ of $P$, all vertices from $C$ have the same spin in each of the $L$ samples $\{\tau_1,\dots,\tau_L\}$; i.e., $\muemph(\mathcal{P})=1$. The algorithm will regard the edges in $E(P)$
as the heavy ones. 

The algorithm then generates $L$ samples from the given Ising model $(G,\beta)$ and computes ${\muemp}(P)$ from these samples.
If ${\muemp}(P)$ is small, then it means that there is substantial disagreement between the sets of heavy edges in the known and hidden models, so the algorithm will output \textsc{No}.
Otherwise, the algorithm has found a partition $P$ such that vertices from the same partition class are very likely to receive the same spin in both $(G,\beta)$ and $( G^*,\beta^*)$.
The algorithm outputs \textsc{No} if any heavy edge of $G$ is not included in $E(P)$.

As a result, after Steps \ref{algstep:bad-event2-for}-\ref{algstep:bad-event2-then},
the weights of the  edges in $E(G_\textsc{p})$ are guaranteed to be $O(\log n)$ 
and those in $E({G}^*_\textsc{p})$ are also $O(\log n)$ with high probability.
Consequently, the original identity testing problem for $(G,\beta)$ and $(G^*,\beta^*)$ reduces to the same problem for the quotient models $(G_\textsc{p},\beta_\textsc{p})$ and $({G}^*_\textsc{p},{\beta}^*_\textsc{p})$ where $\beta_\textsc{p},{\beta}^*_\textsc{p} = O(\log n)$. 
Hence, when we run the DDK algorithm 
on the quotient models, the dependence on $\beta$ can be replaced by a $\poly(\log n)$ term.
The precise sample complexity of Algorithm~\ref{alg-testing} is given in the following theorem.



\begin{thm}
	\label{thm:alg-main}
	Suppose $G,G^* \in \mathcal{M}(n,d)$ and $\beta,\beta^*>0$. For all sufficiently large $n$, Algorithm~\ref{alg-testing} outputs the correct answer for the identity testing problem with probability at least $3/4$ and has sample complexity $L =  O(n^2 d^2 \eps^{-2}\log^3n)$.
	Moreover, the running time of Algorithm~\ref{alg-testing} is $\poly(n,d,\varepsilon^{-1})$.
\end{thm}

\begin{rmk}
	\label{rmk:field}
	Our guarantees for Algorithm~\ref{alg-testing}  in Theorem~\ref{thm:alg-main} extend without significant modification to the case
	where the Gibbs distribution $\mu_{G,\beta}$ includes
	a \emph{consistent} magnetic field or vertex potential.
	(Recall that a magnetic field is consistent if it has the same sign in every vertex.)
	We believe that identity testing for the ferromagnetic Ising model with inconsistent fields is actually hard
	since sampling is already known to be \BIShard in this setting~\cite{GJfield}.
	Observe that the DDK algorithm is not guaranteed to run in polynomial time with inconsistent fields since in this case we do not know
	how to compute the pairwise covariances efficiently. 
\end{rmk}

Before proving Theorem~\ref{thm:alg-main}, we first show that, with high probability, the empirical distributions $\muemp$, $\muemph$ are close to the corresponding Gibbs distributions $\mu$, ${\mu}^*$.
Define $\mathcal{F}$ to be the event that the following two events occur:
\begin{enumerate}
	\item For every partition $P$ of $V$,
	$
	|\muemp(\mathcal{P}) - \mu(\mathcal{P})| < \frac{\eps}{8}$ and $|\muemph(\mathcal{P}) - {\mu^*}(\mathcal{P})| < \frac{\eps}{8};
	$
	\item For every $v,w\in V$, 	if
	${\mu^*}(X_v = X_w) \geq 1 - \frac{1}{20n^2L}$, then $\tau_i(v) = \tau_i(w)$ for all $1\leq i\leq L$.
\end{enumerate}
The probability space associated with the event $\mathcal{F}$ is determined by the random samples $\{\sigma_1,\dots,\sigma_L\}$ and $\{\tau_1,\dots,\tau_L\}$. 
If $\mathcal{F}$ occurs, then the empirical and true distributions are close to each other. 
We can prove that the event $\mathcal{F}$ occurs with probability at least $19/20$.

\begin{lemma}
	\label{lem:F-holds}
	Suppose $L\geq 800n^2\eps^{-2}$. Then for $n$ sufficiently large we have
	$
	\Pr[\mathcal{F}] \geq \frac{19}{20}.
	$	
\end{lemma}

We justify next Steps~\ref{algstep:bad-event1-if}-\ref{algstep:bad-event2-then} of Algorithm~\ref{alg-testing}. 
Let $P$ be the partition of $V$ we get after Step~\ref{algstep:find-P-2} of the algorithm.
Let $\mathcal{E}$ be the event that the following two events occur:
\begin{enumerate}
	\item $
	\muemp(\mathcal{P}) > 1-\frac{\eps}{4};
	$
	\item For every edge $\{v,w\}\in E\backslash E(P)$,
	$
	\beta(v,w) < \ln(20n^2L).
	$
\end{enumerate}
Like $\mathcal{F}$, the probability space associated with $\mathcal{E}$ is determined by the random samples $\{\sigma_1,\dots,\sigma_L\}$ and $\{\tau_1,\dots,\tau_L\}$. 
Observe that if the event $\mathcal{E}$ does not occur, then Algorithm~\ref{alg-testing} will output \textsc{No}.
The following lemma justifies this.
\begin{lemma}\label{lem:E-is-true}
	Assume the event $\mathcal{F}$ occurs. If $(G,\beta) = (G^*,\beta^*)$, then the event $\mathcal{E}$ always occurs.
\end{lemma}

Now, suppose $\mathcal{F}$ and $\mathcal{E}$ occur; then, the algorithm outputs the answer by running the DDK algorithm on the quotient models.
The next two lemmas,
in which we establish several useful properties of the quotient models $(G_\textsc{p},\beta_\textsc{p})$ and $({G}^*_\textsc{p},{\beta}^*_\textsc{p})$,
will be used to guarantee the correctness of Algorithm~\ref{alg-testing} in this case.

\begin{lemma}\label{lem:quotient-close}
	Assume both of the events $\mathcal{F}$ and $\mathcal{E}$ occur. Let $P$ be the partition of $V$ we get after Step~\ref{algstep:find-P-2} of Algorithm~\ref{alg-testing}. Then the following events always occur:
	\begin{enumerate}
		\item The original Gibbs distribution $\mu$ and the conditional Gibbs distribution $\mu(\cdot |\mathcal{P})$ are close:
		\[
		\TV{\mu}{\mu(\cdot |\mathcal{P})} < \frac{3\eps}{8};
		\]
		\item The original Gibbs distribution ${\mu}^*$ and the conditional Gibbs distribution ${\mu^*}(\cdot |\mathcal{P})$ are close:
		\[
		\TV{{\mu^*}}{{\mu}^*(\cdot |\mathcal{P})} < \frac{\eps}{8}.
		\]
	\end{enumerate}	
\end{lemma}

\begin{lemma}\label{lem:quotient-bound}
	Assume both of the events $\mathcal{F}$ and $\mathcal{E}$ occur. Let $P$ be the partition of $V$ we get after Step~\ref{algstep:find-P-2} of Algorithm~\ref{alg-testing}. Then the following events always occur:
	\begin{enumerate}
		\item For every edge $\{v,w\}\in E\backslash E(P)$,
		$
		\beta(v,w) < \ln(20n^2L);
		$
		\item For every edge $\{v,w\}\in {E^*}\backslash E(P)$,
		$
		{\beta^*}(v,w) < \ln(20n^2L).
		$
	\end{enumerate}	
\end{lemma}

The proof of Lemmas~\ref{lem:F-holds},~\ref{lem:E-is-true},~\ref{lem:quotient-close} and~\ref{lem:quotient-bound} are provided in Section~\ref{subsec:alg-lemmas}.
We are now ready to prove Theorem~\ref{thm:alg-main}.

\begin{proof}[Proof of Theorem~\ref{thm:alg-main}]
	Assume the event $\mathcal{F}$ occurs. If the event $\mathcal{E}$ does not occur, then by Lemma~\ref{lem:E-is-true} we have $(G,\beta) \neq ( G^*, \beta^*)$, and Algorithm~\ref{alg-testing} will accordingly return \textsc{No}. Let us assume that the event $\mathcal{E}$ occurs. 
	Recall that we denote the Gibbs distributions of the quotient models $(G_\textsc{p}, \beta_\textsc{p})$ and $({G}^*_\textsc{p}, {\beta}^*_\textsc{p})$ by $\mu_\textsc{p}$ and ${\mu}^*_\textsc{p}$ respectively. 
	If $\mu={\mu}^*$, then $\mu(\cdot |\mathcal{P}) = {\mu}^*(\cdot |\mathcal{P})$ and therefore $\mu_\textsc{p} = {\mu}^*_\textsc{p}$.
	Otherwise, if $\TV{\mu}{{\mu}^*}>\eps$, then we deduce from Lemma~\ref{lem:quotient-close} that
	\begin{align*}
	\TV{\mu_\textsc{p}}{{\mu}^*_\textsc{p}} &= \TV{\mu(\cdot |\mathcal{P})}{{\mu}^*(\cdot |\mathcal{P})}\\
	&\geq \TV{\mu}{{\mu}^*} - \TV{\mu}{\mu(\cdot |\mathcal{P})} - \TV{{\mu}^*}{{\mu}^*(\cdot |\mathcal{P})}\\
	&> \eps - \frac{3\eps}{8} - \frac{\eps}{8} = \frac{\eps}{2}.
	\end{align*}
	Since in every sample $\tau_i$, vertices from the same partition class of $P$ always receive the same spin, we can regard $\{\tau_i\}_{i=1}^L$ as independent samples from the conditional distribution ${\mu}^*(\cdot |\mathcal{P})$. Therefore, $\{(\tau_i)_\textsc{p}\}_{i=1}^L$ are independent samples from the Gibbs distribution ${\mu}^*_\textsc{p}$ of the quotient model $({G}^*_\textsc{p}, {\beta}^*_\textsc{p})$. Thus, we can run the DDK algorithm on inputs $(G_\textsc{p}, \beta_\textsc{p})$, $\{(\tau_i)_\textsc{p}\}_{i=1}^L$ and $\eps/2$. By Theorem~\ref{thm:DDK} and Lemma~\ref{lem:quotient-bound}, the number of samples needed is $L =  O(n^2d^2\eps^{-2}\log^3 n)$.
	
	Consequently, Algorithm~\ref{alg-testing} fails only if the event $\mathcal{F}$ does not occur or the DDK algorithm makes a mistake. By Theorem~\ref{thm:DDK}, the DDK algorithm has success probability $4/5$. Thus, Lemma~\ref{lem:F-holds} implies that the failure probability of Algorithm~\ref{alg-testing} is at most $1/20+1/5 = 1/4$,  provided $L\geq 800n^2\eps^{-2}$.
	Finally, Theorems~\ref{thm:sample-ferro-Ising} and \ref{thm:DDK} imply that the overall running time of Algorithm~\ref{alg-testing} is $\poly(n,d,\varepsilon^{-1})$ as claimed.
\end{proof}

\subsection{Proofs of auxiliary lemmas}
\label{subsec:alg-lemmas}

In this section we provide the missing proofs of Lemmas~\ref{lem:F-holds},~\ref{lem:E-is-true},~\ref{lem:quotient-close} and~\ref{lem:quotient-bound}. 

\begin{proof}[Proof of Lemma~\ref{lem:F-holds}]
	Let $P$ be a partition of $V$. 
	Since the samples $\{\sigma_1,\dots,\sigma_L\}$ from $(G,\beta)$ are $(\eps/16)$-approximate, we have
	\[
	|\mu_{\textsc{alg}}(\mathcal{P}) - \mu(\mathcal{P})| \leq \TV{\mu_{\textsc{alg}}}{\mu} \leq \frac{\eps}{16}.
	\]
	Then, by the triangle inequality
	$$
	|\muemp(\mathcal{P}) - \mu(\mathcal{P})|  \le 	|\muemp(\mathcal{P}) - \mu_{\textsc{alg}}(\mathcal{P})|  + 	|\mu_{\textsc{alg}}(\mathcal{P})- \mu(\mathcal{P})|  \le |\muemp(\mathcal{P}) - \mu_{\textsc{alg}}(\mathcal{P})|  + \frac{\eps}{16}
	$$	
	A Chernoff bound then implies
	\begin{align*}
	\Pr \left[|\muemp(\mathcal{P}) - \mu(\mathcal{P})| \geq \frac{\eps}{8}\right] &\leq \Pr \left[|\muemp(\mathcal{P}) - \mu_{\textsc{alg}}(\mathcal{P})| \geq \frac{\eps}{16}\right]\\
	&= \Pr \left[ \left| \sum_{i=1}^L \Ind\{\sigma_i \in \mathcal{P}\} - L\mu_{\textsc{alg}}(\mathcal{P}) \right| \geq \frac{\eps L}{16}\right]\\
	&\leq 2 \exp \left( -\frac{\eps^2 L}{768\mu_{\textsc{alg}}(\mathcal{P})} \right)\\
	&\leq 2 \exp \left( -\frac{\eps^2 L}{768} \right)
	\leq 2\e^{-n^2},
	\end{align*}
	where the last inequality holds when $\eps^2 L \geq 800n^2$. 
	The total number of partitions of $V$ is at most $n^n$. It then follows from the union bound that
	\[
	\Pr \left[ \exists ~\text{a partition}~ P ~\text{of}~ V: |\muemp(\mathcal{P}) - \mu(\mathcal{P})| \geq \frac{\eps}{8} \right] \leq n^n \cdot 2\e^{-n^2} = 2\e^{n\ln n-n^2} \leq \frac{1}{80},
	\]
	for large enough $n$.
	
	In similar fashion, we deduce that the same holds for $|\muemph(\mathcal{P}) - {\mu}^*(\mathcal{P})|$. Namely,
	\[
	\Pr \left[ \exists ~\text{a partition}~ P ~\text{of}~ V: |\muemph(\mathcal{P}) - {\mu}^*(\mathcal{P})| \geq \frac{\eps}{8} \right] \leq \frac{1}{80}.
	\]
	
	Finally, for each $v,w\in V$ such that ${\mu}^*(X_v=X_w) \geq 1 - (20n^2L)^{-1}$, we obtain from a union bound over the samples that
	\[
	\Pr\left[\exists i,\, 1\leq i \leq L: \tau_i(v)\neq \tau_i(w)\right] \leq L \cdot {\mu}^*(X_v\neq X_w) \leq \frac{1}{20n^2}.
	\]
	Another union bound, this time over the pairs of vertices, implies 
	\[
	\Pr\left[\exists v,w\in V,\, \hat{\mu}(X_v=X_w) \geq 1 - (20n^2L)^{-1},\, \exists i,\, 1\leq i \leq L: \tau_i(v)\neq \tau_i(w) \right] \leq \frac{n^2}{2} \cdot \frac{1}{20n^2} = \frac{1}{40}.
	\]
	
	Combining all bounds above, we obtain from another union bound that
	\[
	\Pr[\neg \mathcal{F}] \leq \frac{1}{80} + \frac{1}{80} + \frac{1}{40} = \frac{1}{20},
	\]
	as desired.
\end{proof}

\begin{proof}[Proof of Lemma~\ref{lem:E-is-true}]
	Assume $(G,\beta) = (G^*,\beta^*)$. 
	Since $\mathcal{F}$ occurs, we have
	\[
	\muemp(\mathcal{P}) > \mu(\mathcal{P}) - \frac{\eps}{8} = {\mu}^*(\mathcal{P}) - \frac{\eps}{8} > \muemph(\mathcal{P}) - \frac{\eps}{4} = 1 - \frac{\eps}{4}.
	\]
	Suppose next that there exists some $\{v,w\}\in E\backslash E(P)$ such that $\beta(v,w) \geq \ln(20n^2L)$. Then, ${\beta}^*(v,w) = \beta(v,w) \geq \ln(20n^2L)$.
	We deduce from Lemma 11 in \cite{DDK} that
	\[
	{\mu}^*(X_v\neq X_w) \leq \frac{1}{\e^{{\beta}^*(v,w)} + 1} \leq \e^{-{\beta}^*(v,w)} \leq \frac{1}{20n^2L}.
	\]
	The event $\mathcal{F}$ implies that $\tau_i(v) = \tau_i(w)$ for all $1\leq i\leq L$. It follows that $v$ and $w$ must belong to the same partition class of $P$; i.e., $\{u,v\}\in E(P)$, which leads to a contradiction. Hence, the event $\mathcal{E}$ always occurs when $\mathcal{F}$ occurs.
\end{proof}

\begin{proof}[Proof of Lemma~\ref{lem:quotient-close}]
	Since both of the events $\mathcal{F}$ and $\mathcal{E}$ occur, we have
	\[
	\mu(\mathcal{P}) > \muemp(\mathcal{P}) - \frac{\eps}{8} > 1 - \frac{3\eps}{8}.
	\]
	It follows immediately that
	$
	\TV{\mu}{\mu(\cdot |\mathcal{P})} = 1 - \mu(\mathcal{P}) < 3\eps/8.
	$
	Similarly, since the event $\mathcal{F}$ occurs, we have
	\[
	{\mu}^*(\mathcal{P}) > \muemph(\mathcal{P}) - \frac{\eps}{8} = 1 - \frac{\eps}{8},
	\]
	and
	$
	\TV{{\mu}^*}{{\mu}^*(\cdot |\mathcal{P})} = 1 -{\mu^*}(\mathcal{P}) <3\eps/8. 
	$
\end{proof}

\begin{proof}[Proof of Lemma~\ref{lem:quotient-bound}]
	When $\mathcal{E}$ occurs, 
	for all $\{v,w\}\in E\backslash E(P)$, we have $\beta(v,w) < \ln(20n^2L)$. Suppose $\{v,w\}\in {E^*}\backslash E(P)$ and ${\beta}^*(v,w) \geq \ln(20n^2L)$. 
	By Lemma 11 in \cite{DDK} we have 
	\[
	{\mu}^*(X_v\neq X_w) \leq \frac{1}{\e^{{\beta}^*(v,w)} + 1} \leq \e^{-{\beta}^*(v,w)} \leq \frac{1}{20n^2 L}.
	\]
	When $\mathcal{F}$ occurs, $\tau_i(v)=\tau_i(w)$ for all $1\leq i\leq L$ and thus $v$ and $w$ belong to the same partition class of $P$. It follows that $\{v,w\}\in E(P)$ which is a contradiction. 
\end{proof}

\section{Lower bounds for proper $q$-colorings}
\label{sec:BIS-hardness}

Let $d$ and $q$ be positive integers
and let $G=(V,E) \in \mathcal{M}(n,d)$, where, as in the previous sections, $ \mathcal{M}(n,d)$ denotes the family of all $n$-vertex graphs of maximum degree at most $d$.
We use $\Omega_G$ for the set of all proper $q$-colorings of $G$
and $\mu_G$ for the uniform distribution on $\Omega_G$.
We recall that a coloring of the vertices of $G$ using colors $\{1,\dots,q\}$ is proper
if the endpoints of every edge in $G$ are assigned different colors. 
The proper $q$-colorings model is one of the easiest combinatorial examples of a hard-constraint model.

The identity testing problem for proper $q$-colorings in $\mathcal{M}(n,d)$ is described as follows: given $q$, a graph $G\in \mathcal{M}(n,d)$ and sample access to random $q$-colorings of an unknown graph $G^*\in \mathcal{M}(n,d)$, distinguish with probability at least $3/4$ whether $\mu_G=\mu_{G^*}$ or $\norm{\mu_G-\mu_{G^*}}>1/3$. 
We establish lower bounds for this problem, thus initiating the study of identity testing in the context of hard-constraint spin systems.

Our lower bounds will crucially use the presumed hardness of the {\BIS} problem.
This is the problem of counting independent sets in bipartite graphs.  
{\BIS} is believed not to have an FPRAS, and
it
is widely used
in the study of the complexity of
approximate counting
problems; see, e.g.,~\cite{DGGJ,GJ,DGJ,BDGJM,CDGJLMR,CGGGJSV,GGJ}.
Specifically, we utilize the hardness of the problem of counting proper $3$-colorings in bipartite graphs,
which we denote by {\BIPthreeCOLOR} and is known to be no easier than {\BIS}.
\begin{thm}[\cite{DGGJ}]\label{thm:BIS}
	If {\BIPthreeCOLOR} admits an FPRAS, then {\BIS} admits an FPRAS.
\end{thm}

We show that when $d \ge q+\sqrt{q}+\Theta(1)$, 
any identity testing algorithm for proper $q$-colorings for $\mathcal{M}(n,d)$
with running time $T(n)$ and sample complexity $L(n)$
provides a randomized algorithm for {\BIPthreeCOLOR}~on graphs of $\poly(n)$ size with running time 
$\poly(T(n),L(n))$.
This will allow us to establish Theorem~\ref{thm:colorings-hardness} from the introduction, since if $T(n)$ and $L(n)$ were polynomials in $n$, then one would obtain an FPRAS for {\BIPthreeCOLOR} and for {\BIS} by Theorem~\ref{thm:BIS}.
For $q\in\N^+$, let
\begin{equation}
\label{eq:dc}
d_c(q) = q + \ceil{\sqrt{q-\frac{3}{4}}-\frac{1}{2}}.
\end{equation}

\begin{thm}
	\label{thm:main-q-coloring}
	Let $d$ and $q$ be positive integers such that  $q\geq 3$ and $d\geq d_c(q)$. 
	Suppose that, for all sufficiently large $n$, there is an identity testing algorithm for proper $q$-colorings in $\mathcal{M}(n,d)$ with running time $T(n)$ and sample complexity $L(n)$. Then,
	for every integer $N$ sufficiently large, $\delta \in(0,1)$ and $\varepsilon \in(0,1)$,
	there exists an integer $n = \Theta(\varepsilon^{-2} N^{4})$
	such that if $L(n) \leq 2^{N-4} $, 
	then
	there is an algorithm 
	that with probability at least $1-\delta$ computes an $\eps$-approximation 
	for {\BIPthreeCOLOR} 
	on bipartite graphs with $N$ vertices.
	The running time of this algorithm is 
	$$O\left( [nL(n)+T(n)]N\ln(N/\delta) + \eps^{-8} \right).$$
\end{thm}

Theorem~\ref{thm:colorings-hardness} is a direct corollary of this result.

\begin{proof}[Proof of Theorem~\ref{thm:colorings-hardness}]
	Suppose there is an identity testing algorithm for $q$-colorings in $\mathcal{M}(n,d)$ with $\poly(n)$ running time and sample complexity; i.e., $L(n)\leq T(n)=\poly(n)$. Then, by Theorem~\ref{thm:main-q-coloring}, for any $\eps,\delta\in(0,1)$ there is an algorithm for {\BIPthreeCOLOR} on an $N$-vertex bipartite graph that outputs an $\eps$-approximation solution with probability at least $1-\delta$ in time $\poly(N,\eps^{-1},\ln(\delta^{-1}))$. That is, there is an FPRAS for {\BIPthreeCOLOR}
	and thus also one for {\BIS} by Theorem~\ref{thm:BIS}. This leads to a contradiction and the result follows.
\end{proof}


The proof of Theorem~\ref{thm:main-q-coloring} 
is divide into two cases: $q \ge 4$ and $q = 3$.
Conceptually, these two cases are proved in the same manner but in the $q=3$ case the construction of the testing instance requires some additional ideas. 
The proof for $q\geq 4$ is provided in Section~\ref{subsec:proof:col:main:q>=4}. 
The $q=3$ case is considered in in Section~\ref{subsec:proof:col:main:q=3}.
Before that, we provide a proof sketch containing the high level ideas of our proof in Section~\ref{sec:main-result-col},
we introduce our gadget $G(m,q,t)$ in Section~\ref{subsec:col:gadget}, and we describe the construction of the coloring instance in Section~\ref{subsec:testing:instance}. 

\subsection{Lower bounds for proper colorings: proof overview}
\label{sec:main-result-col}


As mentioned, we crucially use in our proof
the hardness of {\BIPthreeCOLOR}, the problem of counting proper $3$-colorings in bipartite graphs. 
We show that when $d \ge d_c(q)$, 
an identity testing algorithm for proper $q$-colorings in $\mathcal{M}(n,d)$,
with running time $T(n)$ and sample complexity $L(n)$,
can be turned into a randomized algorithm for {\BIPthreeCOLOR}~on graphs of $\poly(n)$ size with running time 
$\poly(T(n),L(n))$; see Theorem~\ref{thm:BIS}.
Theorem~\ref{thm:main-q-coloring} follows from the fact that if $T(n)$ and $L(n)$ were both polynomials in $n$, then we would obtain an algorithm
that computes an $\eps$-approximation 
for {\BIPthreeCOLOR} in polynomial time.

To derive an algorithm for {\BIPthreeCOLOR} we proceed as follows.
Let $H$ be an $N$-vertex connected bipartite graph, and suppose we want to compute an $\varepsilon$-approximation for 
the number of $3$-colorings $Z_3(H)$ of $H$.
Let $B$ be the \emph{complete} $N$-vertex bipartite graph with the same bipartition as $H$, and let $Z_3(B)$ denote the number of $3$-colorings of $B$. Then, $Z_3(H) \in [Z_3(B),3^N]$.
We converge to an $\varepsilon$-approximation of $Z_3(H)$ via binary search in the interval $[Z_3(B),3^N]$. Specifically, for $\hat{Z} \in [Z_3(B),3^N]$ we construct a suitable identity testing instance and run the identity testing algorithm to determine whether we should consider larger or smaller values than $\hat{Z}$.

The testing instance is constructed as follows. For integers $k,\ell \ge 1$, we define the graph $\hat{H}_{k,\ell}$ that consists of $k$ copies 
$H_1,\dots,H_k$ 
of the original graph $H$
and a complete $(q-3)$-partite graph $J$ in which each cluster has $\ell$ vertices.
In addition to the edges in $J$ and in the $k$ copies of $H$, $\hat{H}_{k,\ell}$ also contains edges between every vertex in $J$ and every vertex in $H_i$ for $i = 1,\dots,k$. 
(Our definition of $\hat{H}_{k,\ell}$ requires $q\geq 4$; the case when $q\!=\!3$ requires a slightly more complicated construction which is provided in Section~\ref{subsec:proof:col:main:q=3}.)
For any $\hat{Z} \in [Z_3(B),3^N]$, we choose 
$k$ and $\ell$ in a way so that the output of the identity testing algorithm on $\hat{H}_{k,\ell}$
can be interpreted as feedback on whether  or not $\hat{Z} >Z_3(H)$.

We set $k = \ceil{N/\eps}$ where $\varepsilon$ is the accuracy parameter.
The choice of $\ell$ is more subtle.
There are only two types of colorings for $\hat{H}_{k,\ell}$: (i) those where $J$ uses $q-3$ colors and (ii) those where $J$ uses $q-2$ colors.
It can be easily checked that there are $|\Omega_1| =\Theta(Z_3(H)^k)$ colorings of the first type
and $|\Omega_2| = \Theta(2^{\ell+k})$ of the second type. Hence, the choice of $\ell$ will determine which of these two types of colorings dominates in the uniform distribution 
$\mu_{k,\ell}$
over the proper colorings of~$\hat{H}_{k,\ell}$.

%

To compare $\hat Z$ and $Z_3(H)$,
we could set $\ell$ so that $\hat{Z}^k \!=\! |\Omega_2| \!=\! \Theta(2^{\ell+k})$ and draw a sample from $\mu_{k,\ell}$. If we get a coloring of the first kind, we may presume that 
$|\Omega_1| \gg |\Omega_2|$, or equivalently that 
$Z_3(H) \!>\! \hat{Z}$. Conversely, if the coloring is of the second kind, then it is likely that $|\Omega_1| \ll |\Omega_2|$ and $Z_3(H) \!<\! \hat{Z}$. 
Sampling from $\mu_{k,\ell}$ is hard,
but we can emulate this approach with a testing algorithm.

Specifically, we construct a simpler graph $\hat{B}_{k,\ell}$ such that:  (i) we can easily generate samples from
$\hat{\mu}_{k,\ell}$, the uniform distribution over the proper
$q$-colorings $\hat{B}_{k,\ell}$; and (ii) $\mu_{k,\ell}$ and $\hat{\mu}_{k,\ell}$ are close in total variation distance 
if and only if the dominant colorings in the Gibbs distributions are those of the second type. Then, we pass $q$, $\hat{H}_{k,\ell}$ and samples from $\hat{\mu}_{k,\ell}$ as input to the tester.
Its output then reveals the dominant color class and hence whether  $\hat Z$ is larger or smaller than $Z_3(H)$.

Our final obstacle is that 
the maximum degree of the graph $\hat{H}_{k,\ell}$ depends on $N$, $k$ and $\ell$,
and could be much larger than $d$. 
To reduce the degree of $\hat{H}_{k,\ell}$ so that it belongs to $\mathcal M(n,d)$, we design a degree reducing gadget, which is inspired by the gadgets used to establish the hardness of the decision and 
structure learning problems \cite{EHK,MR,BCSVsl}. 

\subsection{The colorings gadget}
\label{subsec:col:gadget}

In this section, we present our construction of the coloring gadget, which is 
inspired by similar constructions in \cite{EHK,MR,BCSVsl} for establishing the computational hardness of the decision and (equivalent) structure learning problems for proper $q$-colorings. 

For $m,q,t\in\N^+$ with $t<q$, the graph $G(m,q,t) = (V(m,q,t),E(m,q,t))$ is defined as follows. Let $C_1,\dots,C_m$ be cliques of size $q-1$ and let $I_1,\dots,I_m$ be independent sets of size $t$. Then, set
\[
V(m,q,t) = \bigcup_{i=1}^m \big(V(C_i) \cup V(I_i)\big)
\]
where $V(C_i)$ and $V(I_i)$ are the vertex sets of $C_i$ and $I_i$ respectively for $1\leq i\leq m$. Vertices in the independent sets $I_1,\dots,I_m$ are called \textit{ports}. The cliques $C_i$'s and the independent sets $I_i$'s are connected in the following way:
\begin{enumerate}
	\item For $1\leq i\leq m$, there is a complete bipartite graph between $C_i$ and $I_i$. That is, for $u\in C_i$ and $v\in I_i$, $\{u,v\}\in E(m,q,t)$.
	\item For $2\leq i\leq m$, each $C_i$ is partitioned into $t$ almost-equally-sized disjoint subsets $C_{i,1},\dots,C_{i,t}$ of size either $\floor{(q-1)/t}$ or $\ceil{(q-1)/t}$. Then, the $j$-th vertex of $I_{i-1}$ is connected to every vertex in $C_{i,j}$.
\end{enumerate}
Together with the edges in the cliques $C_i$ for $1\leq i\leq m$, these edges constitute the edge set $E(m,q,t)$. See Figure~\ref{fig:G(m,q,t)} for an illustration of the graph $G(m,q,t)$ and Figure~\ref{fig:G(3,3,2)} for $G(3,3,2)$ as an example.

\begin{figure}[t]
	\centering
	\begin{tikzpicture}[scale=0.96]
	\node(C1) [shape=circle, draw=black, minimum size=1.5cm] at (0,0) {$C_1$};
	\node(I1) [shape=circle, draw=black, minimum size=0.9cm] at (2,0) {$I_1$};
	\node(C2) [shape=circle, draw=black, minimum size=1.5cm] at (4,0) {$C_2$};
	\node(I2) [shape=circle, draw=black, minimum size=0.9cm] at (6,0) {$I_2$};
	\node(Im1)[shape=circle, draw=black, minimum size=0.9cm] at (9,0) {$I_{m-1}$};
	\node(Cm) [shape=circle, draw=black, minimum size=1.5cm] at (11,0) {$C_m$};
	\node(Im) [shape=circle, draw=black, minimum size=0.9cm] at (13,0) {$I_m$};
	
	\path
	(C1) edge (I1)
	(I1) edge [dashed] (C2)
	(C2) edge (I2)
	(Im1) edge [dashed] (Cm)
	(Cm) edge (Im)
	
	(I2) edge[white] node[black]{$\cdots\cdots$} (Im1);
	\end{tikzpicture}
	\caption{The graph $G(m,q,t)$. Each of $C_1,\ldots,C_m$ is a clique of size $q-1$ and each of $I_1,\ldots,I_m$ is an independent set of size $t<q$. Solid lines between $C_i$ and $I_i$ mean that every vertex in $C_i$ is adjacent to every vertex in $I_i$. Dashed lines between $I_{i-1}$ and $C_i$ mean that every vertex in $I_{i-1}$ is adjacent to roughly $(q-1)/t$ vertices in $C_i$ with no two vertices in $I_{i-1}$ sharing a common neighbor in $C_i$.}
	\label{fig:G(m,q,t)}
	
	\vspace{1em}
	\begin{tikzpicture}[scale=0.9]
	\filldraw
	(0,0) circle (2pt) node {}
	(2,0) circle (2pt) node {}
	(4,0) circle (2pt) node {}
	(6,0) circle (2pt) node {}
	(8,0) circle (2pt) node {}
	(10,0) circle (2pt) node {}
	(0,2) circle (2pt) node {}
	(2,2) circle (2pt) node {}
	(4,2) circle (2pt) node {}
	(6,2) circle (2pt) node {}
	(8,2) circle (2pt) node {}
	(10,2) circle (2pt) node {}
	;
	
	\path
	(0,0) edge (0,2)
	(4,0) edge (4,2)
	(8,0) edge (8,2)
	
	(0,0) edge (2,0)
	(2,0) edge (4,0)
	(4,0) edge (6,0)
	(6,0) edge (8,0)
	(8,0) edge (10,0)
	
	(0,2) edge (2,2)
	(2,2) edge (4,2)
	(4,2) edge (6,2)
	(6,2) edge (8,2)
	(8,2) edge (10,2)
	
	(0,0) edge (2,2)
	(0,2) edge (2,0)
	(4,0) edge (6,2)
	(4,2) edge (6,0)
	(8,0) edge (10,2)
	(8,2) edge (10,0)
	;
	
	\draw[dashed]
	(0,1) ellipse (15pt and 45pt)
	(2,1) ellipse (15pt and 45pt)
	(4,1) ellipse (15pt and 45pt)
	(6,1) ellipse (15pt and 45pt)
	(8,1) ellipse (15pt and 45pt)
	(10,1) ellipse (15pt and 45pt)
	;
	
	\draw
	(0,-1) node {$C_1$}
	(2,-1) node {$I_1$}
	(4,-1) node {$C_2$}
	(6,-1) node {$I_2$}
	(8,-1) node {$C_3$}
	(10,-1) node {$I_3$}
	;
	\end{tikzpicture}
	\caption{The graph $G(3,3,2)$ for $m=3$, $q=3$ and $t=\big\lceil\sqrt{q-3/4}+1/2\big\rceil=2$. All ports (i.e., vertices in $I_1\cup I_2 \cup I_3$) have degree at most 3. All non-ports (i.e., vertices in $C_1\cup C_2 \cup C_3$) have degree at most 4. Recall that $d_c(3)=3+\big\lceil\sqrt{3-3/4}-1/2\big\rceil=4$.}
	\label{fig:G(3,3,2)}
\end{figure}

The following key fact of the gadget $G(m,q,t)$ follows from its definition.
\begin{lemma}\label{lem:gadget-phase-coloring}
	Let $m,q,t\in\N^+$ with $t<q$. Then in every proper $q$-coloring of $G(m,q,t)$, all ports (i.e., vertices in the independent sets $I_1,\dots,I_m$) have the same color and all non-ports (i.e., vertices in the cliques $C_1,\dots,C_m$) are assigned the remaining $q-1$ colors.
\end{lemma}

\begin{proof}
	Consider a proper $q$-coloring $\sigma$ of $G(m,q,t)$. Since each $C_i$ is a clique of size $q-1$ for $1\leq i\leq m$, it receives $q-1$ colors in $\sigma$. For $1\leq i\leq m$, for each $v\in I_i$, $v$ is adjacent to all vertices in $C_i$; hence, $v$ receives the only color that is not used by $C_i$ in $\sigma$. That means, for each $i$ all vertices in $I_i$ have the same color which does not appear in $C_i$. Next, for $2\leq i\leq m$, each vertex in $C_i$ is adjacent to some vertex in $I_{i-1}$. Since all vertices in $I_{i-1}$ have the same color, we deduce that $I_{i-1}$ receives the color which does not appear in $C_i$. It follows immediately that all independent sets $I_1,\dots,I_n$ have the same color and the cliques $C_1,\dots,C_m$ use the remaining $q-1$ colors.
\end{proof}

The following lemma shows that when $d\geq d_c(q)$ the maximum degree of the gadget $G(m,q,t)$ is at most $d$ for a certain choice of $t$; moreover, every port has degree at most $d-1$. See Figure~\ref{fig:G(3,3,2)} for an example.
\begin{lemma}\label{lem:gadget-degree-coloring}
	Suppose $q\geq 3$ and $d\geq d_c(q)$. If 
	$
	t = \lceil\sqrt{q-3/4}+1/2\rceil,
	$
	then every port of the graph $G(m,q,t)$ has degree at most $d-1$ and every non-port of $G(m,q,t)$ has degree at most $d$.
\end{lemma}

\begin{proof}
	The degree of a port in $G(m,q,t)$ is bounded by
	\[
	(q-1) + \ceil{\frac{q-1}{t}} \leq (q-1) + \ceil{\frac{q-1}{\sqrt{q-\frac{3}{4}}+\frac{1}{2}}} = q + \ceil{\sqrt{q-\frac{3}{4}}-\frac{1}{2}} - 1 \leq d-1.
	\]
	Meanwhile, the degree of a non-port in $G(m,q,t)$ is at most
	\[
	(q-2) + t + 1 = q + \ceil{\sqrt{q-\frac{3}{4}}+\frac{1}{2}} - 1 = q + \ceil{\sqrt{q-\frac{3}{4}}-\frac{1}{2}} \leq d.\tag*{\qedhere}
	\]
\end{proof}

We define the \textit{phase} of a proper $q$-coloring of $G(m,q,t)$ to be the color of its ports.
In the following lemma we bound the number of $q$-colorings with a given phase, which is used later in the proof of Theorem~\ref{thm:main-q-coloring}.
\begin{lemma}\label{lem:gadget-No-coloring}
	Let $m,q,t\in\N^+$ with $t<q$. Then, the number of proper $q$-colorings of $G(m,q,t)$ with a given phase is $[(q-1)!]^m$.
\end{lemma}

\begin{proof}
	By Lemma~\ref{lem:gadget-phase-coloring}, in every proper $q$-colorings  the vertices in the independent sets $I_1,\dots,I_m$ are assigned the same color, which is given by its phase. With the coloring of these vertices fixed, the number of $q$-colorings of each clique $C_i$ is $(q-1)!$. Since the cliques are disjoint the total number of $q$-colorings is $[(q-1)!]^m$.
\end{proof}

\subsection{Testing instance construction: the $q\geq 4$ case}
\label{subsec:testing:instance}

Let $H=(V,E)$ be a bipartite connected graph on $N$ vertices and suppose we want to approximately count the number of $3$-colorings of $H$.
In this section we show how to construct the testing instance from $H$ when $q\geq 4$. Our construction uses the gadget from Section~\ref{subsec:col:gadget}.

For integers $k,\ell \ge 1$, define the simple graph $\hat{H}_{k,\ell}=(\hat{V},\hat{E})$ as follows:
\begin{enumerate}
	\item Let $H_1 = (V(H_1),E(H_1)),\dots,H_k=(V(H_k),E(H_k))$ be $k$ copies of the graph $H$;
	\item Let $J$ be a complete $(q-3)$-partite graph in which each cluster has $\ell$ vertices; 
	\item Set $\hat{V} = \left( \bigcup_{i=1}^k V(H_i) \right) \cup V(J)$;
	\item In addition to the edges in $J$ and those in $H_i$ for $1\leq i\leq k$, $\hat{E}$ also contains edges between every vertex in $H_i$ for $1\leq i\leq k$ and every vertex in $J$; i.e., for $u\in H_i$ and $v\in J$, we have $\{u,v\}\in \hat{E}$.
\end{enumerate}
We remark that our definition of $\hat{H}_{k,\ell}$ requires $q\geq 4$; when $q=4$, $J$ is simply an independent set with $\ell$ vertices. Next, we use the graph $G(m,q,t)$ from Section~\ref{subsec:col:gadget} as a gadget to construct a simple graph $\hat{H}_{k,\ell}^\Gamma$ based on $\hat{H}_{k,\ell}$ where $\Gamma=\{m,q,t\}$. We proceed as follows:
\begin{enumerate}
	\item Replace every vertex $v$ of $\hat{H}_{k,\ell}$ by a copy $G_v$ of $G(m,q,t)$;
	\item For every edge $\{u,v\}\in \hat{E}$, pick an unused port in $G_u$ and an unused port in $G_v$ and connect them; in this way, every port is connected with at most one port from another gadget.
\end{enumerate}
The number of ports in a gadget $G(m,q,t)$ is $mt$, and the total number of vertices in $\hat{H}_{k,\ell}$ is $kN+\ell(q-3)$. 
The graph $\hat{H}_{k,\ell}^\Gamma$ is well-defined only if we have enough ports in every gadget $G_v$ to connect them with ports from other gadgets. 
For this, it suffices that
\[
mt \geq kN+\ell(q-3),
\]
and so we set
\begin{equation}
\label{eq:param:m}
m = kN+\ell(q-3) \qquad\text{and}\qquad t=\lceil\sqrt{q-3/4}+1/2\rceil \ge 1.
\end{equation}

Let $B$ be a complete bipartite graph with the same vertex bipartition as $H$. By setting $H=B$, we can define the graphs $\hat{B}_{k,\ell}$ and $\hat{B}_{k,\ell}^\Gamma$. Given $k,\ell\in\N^+$, we write $G = \hat{H}_{k,\ell}^\Gamma$ and $G^*=\hat{B}_{k,\ell}^\Gamma$ for our choice of $m$ and $t$. Suppose $q\geq 3$ and $d\geq d_c(q)$. Then, Lemma~\ref{lem:gadget-degree-coloring} implies that $G,G^* \in \mathcal{M}(n,d)$ for
\begin{equation}
\label{eq:param:n}
n=[kN+\ell(q-3)] \cdot [m(q-1+t)] = m^2(q-1+t).
\end{equation}

Let $Z_3(H)$ and $Z_3(B)$ denote the number of $3$-colorings of $H$ and $B$, respectively.
The two uniform distributions $\mu_G$ and $\mu_{G^*}$ over the $q$-colorings of $G$ and $G^*$ are related as follows.
\begin{lemma}\label{lem:G-and-Gstar-coloring}
	Let $k,\ell\in\N^+$ with $\ell\geq 2$. Define
	$
	\psi(k,\ell) = (q-3)^{1/k} \; 2^{1+\ell/k}.
	$
	Then the following holds:
	\begin{enumerate}
		\item If $Z_3(H) < \psi(k,\ell)$, then
		\[
		\TV{\mu_G}{\mu_{G^*}} \leq \frac{4}{3} \left( \frac{Z_3(H)}{\psi(k,\ell)} \right)^k.
		\]
		\item If $Z_3(H) \geq \psi(k,\ell)$, then
		\[
		\TV{\mu_G}{\mu_{G^*}} \geq \frac{2}{5} \left( 1 - \left( \frac{Z_3(B)}{Z_3(H)} \right)^k \right).
		\]
	\end{enumerate}
\end{lemma}

Finally, we note that we can generate random $q$-colorings of $G^*$ in polynomial time.
\begin{lemma}\label{lem:sampling-coloring}
	There exists an algorithm with running time $O(n)$ that generates a sample from the distribution $\mu_{G^*}$.
\end{lemma}

The proof of both of these lemmas are provided in Section~\ref{subsec:coloring-facts}.

\subsection{Proof of Theorem~\ref{thm:main-q-coloring}: the $q\geq 4$ case}
\label{subsec:proof:col:main:q>=4}

In this section, we prove Theorem~\ref{thm:main-q-coloring} for the case when $q\geq 4$. Our proof relies on Lemmas~\ref{lem:G-and-Gstar-coloring} and \ref{lem:sampling-coloring}. 
We converge to a good approximation for the number of $3$-colorings  $Z_3(H)$ of a bipartite graph $H$ using the presumed algorithm for the identity testing problem.
In each round, we choose $k,\ell$ and generate the graph $G = \hat{H}^\Gamma_{k,\ell}$ as described in Section~\ref{subsec:testing:instance};
the size of the graph $G$ depends on $k,\ell$ and thus it varies in each round. 
We then generate samples from $\mu_{G^*}$ in polynomial time by Lemma~\ref{lem:sampling-coloring} where $G^* = \hat{B}^\Gamma_{k,\ell}$. These samples and the graph $G$ are passed as input to the identity testing algorithm.
If $Z_3(H)< \psi(k,\ell)$, then $\mu_G$ and $\mu_{G^*}$ are close in total variation distance (see Lemma~\ref{lem:G-and-Gstar-coloring}), and the tester would return \textsc{Yes}. Otherwise, if $Z_3(H)\geq  \psi(k,\ell)$, then $\mu_G$ and $\mu_{G^*}$ are statistically far from each other, and the tester would return \textsc{No}. Thus, using binary search over $k,\ell$ we can obtain a good approximation for $Z_3(H)$.

\begin{proof}[Proof of Theorem~\ref{thm:main-q-coloring} for $q\geq 4$]
	Let $H=(V(H), E(H))$ be an $N$-vertex connected bipartite graph 
	with $N \ge 5$.
	Suppose we want to approximately count the number of $3$-colorings of $H$. 
	Recall that $B$ is the complete bipartite graph with the same vertex bipartition as $H$. Then, $Z_3(B) \leq Z_3(H) \leq 3^N$ where the upper bound corresponds to the independent set on $N$ vertices.
	
	Fix $\eps,\delta\in(0,1)$. Our goal is to find an integer $\hat{Z}  \in [Z_3(B),3^N]$ such that with probability at least $1-\delta$
	\begin{equation}\label{eq:FPRAS}
	(1-\eps)\hat{Z} \leq Z_3(H) \leq (1+\eps)\hat{Z}.
	\end{equation}
	We assume first that $\eps\geq 2^{-N/4}$.
	The case when $\eps < 2^{-N/4}$ is much simpler and will be considered at the end of the proof.
	We give an algorithm that with probability at least $1-\delta$ outputs an integer $\hat{Z} \in [Z_3(B),3^N]$ such that 
	\begin{equation}\label{eqn:bounds-on-Z}
	2^{-\eps} (\hat{Z}-1) \leq Z_3(H) \leq 2^\eps \hat{Z}.
	\end{equation}
	Then, \eqref{eq:FPRAS} follows from the following fact.
	
	\begin{fact}
		\label{fact:fpras}
		For all $\varepsilon \in [2^{-N/4},1)$,
		if $\hat{Z}$ is such that $2^{-\eps} (\hat{Z}-1) \leq Z_3(H) \leq 2^\eps \hat{Z}$, 	then $(1-\eps)\hat{Z} \leq Z_3(H) \leq (1+\eps)\hat{Z}$.
	\end{fact}
	
	Let $k = \ceil{N/\eps}$. Recall that $\psi(k,\ell) = (q-3)^{1/k} \; 2^{1+\ell/k}$. For any $\hat{Z}\in \N^+$ for which we would like to test if \eqref{eqn:bounds-on-Z} hold, we choose an integer $\ell$ satisfying
	$$
	2^{-\eps} \psi(k,\ell) \leq \hat{Z} \leq \psi(k,\ell).
	$$
	Such an $\ell$ would exist if and only if it satisfies 
	$$
	k\log_2\hat{Z} - \log_2(q-3) -k \leq \ell \leq k\log_2\hat{Z} - \log_2(q-3) -k + k\eps.
	$$
	Since the difference between the upper and lower bounds is $k\eps \geq N\geq 1$, there is always at least one possible value for $\ell$. 
	Note also that $\ell\leq k\log_2(3^N) \leq 2kN$ as $\hat{Z}\leq 3^N$. 
	
	After choosing $k$ and $\ell$, 
	which depend on $N$, $q$, $\varepsilon$ and $\hat{Z}$, 
	we construct the graphs $G = \hat{H}_{k,\ell}^\Gamma$ and $G^*=\hat{B}_{k,\ell}^\Gamma$ as defined in Section~\ref{subsec:testing:instance}. 
	Then, the graphs $G$ and $G^*$ belong to $\mathcal{M}(n_{k,\ell},d)$, where given our choices for $m$, $t$, $k$ and $\ell$, we have:
	\[
	m = kN+\ell(q-3) \leq  \frac{4qN^2}{\eps} \quad\text{and}\quad n_{k,\ell} = m^2(q-1+t) \leq \frac{32q^3N^4}{\eps^2};
	\]
	see \eqref{eq:param:m} and \eqref{eq:param:n}.
	Given $\hat{Z}$, our input to the identity testing algorithm (henceforth called the \textsc{Tester}) is the graph $G$ and $L=L(n_{k,\ell})$ random $q$-colorings of $G^*$. By Lemma~\ref{lem:sampling-coloring}, we can generate one sample from $\mu_{G^*}$ in $O(n_{k,\ell})$ time. Thus, the total running time for one call of the \textsc{Tester} (including the generation of the samples) is $O(nL(n)+T(n))$ for $n = \ceil{32q^3\eps^{-2}N^4}$. 
	The following claim which is proved later follows from Lemma~\ref{lem:G-and-Gstar-coloring}.
	\begin{claim}
		\label{claim:tester-yes-no}
		Suppose $Z_3(B) \leq \hat{Z} \leq 3^N$ and $L \le 2^{N-4}$. 
		\begin{enumerate}
			\item If $Z_3(H) < 2^{-\eps}\hat{Z}$, then the \textsc{Tester} outputs \textsc{Yes} with probability at least $2/3$;
			\item If $Z_3(H) > 2^{\eps}\hat{Z}$, then the \textsc{Tester} outputs \textsc{No} with probability at least $2/3$.
		\end{enumerate}
	\end{claim}
	
	We test whether $\hat{Z}$ provides a bound for $Z_3(H)$ using the following 
	algorithm. 
	For $R\geq 1$ odd, we construct the corresponding graph $G=\hat{H}_{k,\ell}^\Gamma$, generate $L\cdot R$ random colorings of $G^*=\hat{B}_{k,\ell}^\Gamma$, 
	and	run the \textsc{Tester} $R$ times using $L$ samples each time (every sample is used only once).
	The output of this algorithm would be the majority answer in the $R$ rounds.
	We call this algorithm the \emph{$R$-round-\textsc{Tester}}  for $\hat{Z}$.
	The following claim, which follows directly from a Chernoff bound and is provided later, establishes the guarantee for the accuracy of the $R$-round-\textsc{Tester}.
	
	\begin{claim}
		\label{claim:r-round}
		Let $R = 48\ceil{\ln(2N/\delta)}+1$. 
		\begin{enumerate}
			\item 	If $Z_3(H) < 2^{-\eps}\hat{Z}$, then $R$-round-\textsc{Tester} for $\hat{Z}$ outputs \textsc{Yes} with probability at least $1 - \frac{\delta}{2N}$;
			\item  If $Z_3(H) > 2^{\eps}\hat{Z}$, then $R$-round-\textsc{Tester} for $\hat{Z}$ outputs \textsc{No} with probability at least $1 - \frac{\delta}{2N}$.
		\end{enumerate}
	\end{claim}
	
	\noindent
	The algorithm for counting $3$-colorings in $H$ is based on binary search over the interval $[Z_3(B),3^N]$;  in each iteration it uses the $R$-round-\textsc{Tester} to determine the interval for the next iteration. 
	We proceed as follows.
	\begin{enumerate}
		\item Run the $R$-round-\textsc{Tester} for $\hat{Z} = Z_3(B)$. If the $R$-round-\textsc{Tester} outputs \textsc{Yes}, then return $\hat{Z} = Z_3(B)$;
		\item Run the $R$-round-\textsc{Tester} for $\hat{Z} = 3^N$. If the $R$-round-\textsc{Tester} outputs \textsc{No}, then return $\hat{Z} = 3^N$;
		\item Let $(L_0, U_0) = (Z_3(B), 3^N)$. For $i\geq 1:$
		\begin{enumerate}
			\item Let $C_i = \floor{(L_{i-1}+U_{i-1})/2}$;
			\item Run the $R$-round-\textsc{Tester} for $\hat{Z} = C_i$;
			\item If the $R$-round-\textsc{Tester} outputs \textsc{Yes}, then set $(L_i, U_i) = (L_{i-1}, C_i)$;
			\item If the $R$-round-\textsc{Tester} outputs \textsc{No}, then set $(L_i, U_i) = (C_i, U_{i-1})$;
			\item If $U_i-L_i = 1$, return $\hat{Z} = U_i$; otherwise, set $i:=i+1$ and repeat.
		\end{enumerate}
	\end{enumerate}
	
	Observe that $U_i-L_i-1$ decreases by a factor $2$ in each iteration. Thus, the $R$-round-\textsc{Tester} is called at most $2+ \log_2(3^N)\leq 2N$ times for $N\geq 5$. 
	
	Now, let $\mathcal{F}$ be the event that in a single run of the binary search algorithm the following two conditions are maintained:
	\begin{enumerate}[(i)]
		\item If $Z_3(H) < 2^{-\eps}\hat{Z}$, then the $R$-round-\textsc{Tester} outputs \textsc{Yes} for $\hat{Z}$;
		
		
		\item If $Z_3(H) > 2^{\eps}\hat{Z}$, then the $R$-round-\textsc{Tester} outputs \textsc{No} for $\hat{Z}$.
		
	\end{enumerate}
	\noindent
	Claim~\ref{claim:r-round} and a union bound imply that
	\begin{equation}
	\label{eq:prob:f}
	\Pr[\neg \mathcal{F}] \leq \frac{\delta}{2N} \cdot 2N = \delta.
	\end{equation}
	We claim that when $\mathcal{F}$ occurs, the output of the binary search algorithm satisfies \eqref{eqn:bounds-on-Z}. For this we consider three cases.
	First, if the algorithm stops in step 1, then $\hat{Z} = Z_3(B)$; that is, the $R$-round-\textsc{Tester} outputs \textsc{Yes} for $\hat{Z} = Z_3(B)$. Therefore,
	\[
	2^{-\eps} (\hat{Z}-1) \leq Z_3(B) \leq Z_3(H) \leq 2^\eps \hat{Z},
	\]
	where the last inequality follows from condition (ii) in the definition of the event $\mathcal{F}$.
	Similarly, if the algorithm stops in step 2, then $\hat{Z} = 3^N$. Namely, the $R$-round-\textsc{Tester} outputs \textsc{No} for $\hat{Z} = 3^N$, and so
	\[
	2^{-\eps} (\hat{Z}-1) \leq 2^{-\eps} \hat{Z} \leq Z_3(H) \leq 3^N \leq 2^\eps \hat{Z},
	\]
	where the second inequality follows from condition (i) in the definition of $\mathcal{F}$.
	
	Finally, suppose that the binary search algorithm stops in step 3 and $\hat{Z}=U_i$ for some $i\geq 1$.
	Observe that $L_i< U_i$ for all $i\geq 1$. Moreover, for each $i\geq 1$ the $R$-round-\textsc{Tester} outputs \textsc{No} for $\hat{Z} = L_i$ and \textsc{Yes} for $\hat{Z} = U_i$.
	The algorithm stops when $U_i-L_i=1$ for some $i$.
	It follows from the definition of $\mathcal{F}$ that
	\[
	2^{-\eps} (\hat{Z}-1) = 2^{-\eps} L_i \leq Z_3(H) \leq 2^\eps U_i = 2^\eps \hat{Z}.
	\]
	
	\smallskip
	Therefore, the output of the binary  search algorithm satisfies \eqref{eqn:bounds-on-Z} whenever $\mathcal{F}$ occurs.
	From \eqref{eq:prob:f}, it follows that we obtain an $\varepsilon$-approximation for $Z_3(H)$ with probability at least $1-\delta$ as desired.
	
	It remains for us to consider the overall running time of the binary search procedure.
	As mentioned, the $R$-round-\textsc{Tester} algorithm is called at most $2N$ times,
	and the running time of each call is $O(nL(n)+T(n))$ where $n = \ceil{32q^3\eps^{-2}N^4}$.
	Hence, the overall running time of the algorithm is $O\left( (nL(n)+T(n))N\ln(N/\delta) \right)$.
	
	Finally, we mention that for the trivial case when $\eps< 2^{-N/4}$, we can simply enumerate every $3$-labeling $\sigma: V(H) \to \{1,2,3\}$ of $H$ and count the number of proper $3$-colorings. The running time of this process is $O(3^N) \leq O(\eps^{-8})$. 
\end{proof}

We finalize the proof of Theorem~\ref{thm:main-q-coloring} for $q\geq 4$ by providing the missing proofs of 
Fact~\ref{fact:fpras} and Claims \ref{claim:tester-yes-no}~and~\ref{claim:r-round} .

\begin{proof}[Proof of Fact~\ref{fact:fpras}]
	For $\eps\in(0,1)$, we have $2^\eps \hat{Z} \leq (1+\eps)\hat{Z}$. Moreover, for $\eps\in[2^{-N/4},1)$ we have
	\[
	\frac{1}{1-2^\eps (1-\eps)} \leq \frac{1}{1-(1+\eps)(1-\eps)} = \frac{1}{\eps^2} \leq 2^{N/2} \leq Z_3(B) \leq \hat{Z}.
	\]
	This implies that $2^{-\eps} (\hat{Z}-1) \geq (1-\eps) \hat{Z}$ and the theorem follows.
\end{proof}

\begin{proof}[Proof of Claim~\ref{claim:tester-yes-no}]
	Recall that we choose $\ell$ such that
	$
	2^{-\eps} \psi(k,\ell) \leq \hat{Z} \leq \psi(k,\ell)
	$, where $\psi(k,\ell) = (q-3)^{1/k} \; 2^{1+\ell/k}$. Hence,
	when $Z_3(H) < 2^{-\eps}\hat{Z}$ we have
	\[
	Z_3(H) < 2^{-\eps}\hat{Z} \leq 2^{-\eps} \psi(k,\ell) < \psi(k,\ell).
	\]
	Part 1 of Lemma~\ref{lem:G-and-Gstar-coloring} implies
	\[
	\TV{\mu_G}{\mu_{G^*}} \leq \frac{4}{3} \left( \frac{Z_3(H)}{\psi(k,\ell)} \right)^k \leq \frac{4}{3} \cdot 2^{-k\eps} \leq \frac{4}{3} \cdot 2^{-N},
	\]
	where the last inequality follows from the fact that $k  =\ceil{N/\varepsilon}$.
	Let $\mu_{G}^{\otimes L}$ (resp., $\mu_{G^*}^{\otimes L}$) be the product distribution corresponding to $L$ independent samples from $\mu_{G}$ (resp., $\mu_{G^*}$). Recall that $L\leq 2^{N-4}$ by assumption. Then we get
	\[
	\TV{\mu_{G}^{\otimes L}}{\mu_{G^*}^{\otimes L}} \leq L \TV{\mu_G}{\mu_{G^*}} \leq L \cdot \frac{4}{3} \cdot 2^{-N} \leq 2^{N-4} \cdot \frac{4}{3} \cdot 2^{-N} = \frac{1}{12}.
	\]
	Consider the optimal coupling $\mathbb P$ of the distributions $\mu_{G}^{\otimes L}$ and $\mu_{G^*}^{\otimes L}$. In a sample from $\mathbb P$,
	the colorings from $G$ and $G^*$ are equal with probability at least $11/12$.
	Hence, the input to the $\textsc{Tester},$ which is drawn from $\mu_{G^*}^{\otimes L}$, is distributed according to $\mu_{G}^{\otimes L}$ with probability at least $11/12$.
	Denote this event by $\mathcal{F}_\mathcal{S}$. Notice that if $\mathcal{F}_\mathcal{S}$ occurs and the \textsc{Tester} is correct, then the \textsc{Tester} would output \textsc{Yes}. Thus,
	\[
	\Pr[\text{\textsc{Tester} outputs No}] \leq \Pr[\neg \mathcal{F}_{\mathcal{S}}] + \Pr[\text{\textsc{Tester} makes a mistake}] \leq  \frac{1}{12} + \frac{1}{4} = \frac{1}{3},
	\]
	which establishes part 1 of the claim.
	
	If $Z_3(H) > 2^{\eps}\hat{Z}$, then
	$
	Z_3(H) > 2^{\eps}\hat{Z} \geq \psi(k,\ell)
	$
	and
	$
	Z_3(B) \leq \hat{Z} < 2^{-\eps}Z_3(H).
	$
	Part 2 of Lemma~\ref{lem:G-and-Gstar-coloring} implies that for $N\geq 5$
	\[
	\TV{\mu_G}{\mu_{G^*}} \geq \frac{2}{5} \left( 1 - \left( \frac{Z_3(B)}{Z_3(H)} \right)^k \right) \geq \frac{2}{5} \left( 1 - 2^{-k\eps} \right) \geq \frac{2}{5} \left( 1 - 2^{-N} \right) > \frac{1}{3}.
	\]
	It follows that
	\[
	\Pr[\text{\textsc{Tester} outputs \textsc{Yes}}] = \Pr[\text{\textsc{Tester} makes a mistake}] \leq \frac{1}{4} < \frac{1}{3}.\tag*{\qedhere}
	\]
\end{proof}

\begin{proof}[Proof of Claim~\ref{claim:r-round}]
	For $i=1,\dots,R$, let $X_i$ be the indicator of the event that in the $i$-th round the \textsc{Tester} outputs \textsc{Yes}. Let $X= \sum_{i=1}^R X_i$.
	If $Z_3(H) < 2^{-\eps}\hat{Z}$, then Claim~\ref{claim:tester-yes-no} implies that $\E [X] \geq \frac{2}{3}R$. The Chernoff bound then implies that the probability that the $R$-round-\textsc{Tester} outputs \textsc{No} is
	\begin{equation*}
	\Pr\left[ X \leq \frac{R}{2} \right] \leq \Pr\left[ X \leq \frac{3}{4} \E [X] \right] \leq \exp\left( -\frac{\E [X]}{32} \right) \leq \exp\left( -\frac{R}{48} \right) \leq \frac{\delta}{2N}.
	\end{equation*}
	The case when $Z_3(H) > 2^\eps\hat{Z}$ (part 2) can be derived analogously.
\end{proof}

\subsubsection{Colorings of $G$ and $G^*$: proof of Lemmas~\ref{lem:G-and-Gstar-coloring} and \ref{lem:sampling-coloring}}
\label{subsec:coloring-facts}

In this section we establish first several facts about the $q$-colorings of $G =\hat{H}_{k,\ell}^\Gamma$ and $G^* = \hat{B}_{k,\ell}^\Gamma$.
We then use these facts to
bound $\TV{\mu_G}{\mu_{G^*}}$ (Lemma~\ref{lem:G-and-Gstar-coloring})
and to design an
algorithm for sampling the $q$-colorings of $G^*$ (Lemma~\ref{lem:sampling-coloring}).

For $r\in \{2,3\}$,
let $Z_r(H)$ denote the number of $r$-colorings of $H$. Since $H$ is a connected bipartite graph, we have $Z_2(H)=2$.
The following lemma establishes a useful partition of the $q$-colorings of $\hat{H}_{k,\ell}$.
\begin{lemma}\label{lem:H-hat-coloring}
	Let $k,\ell\in\N^+$.
	Let $\Omega^a$ and $\Omega^b$ be the set of $q$-colorings of $\hat{H}_{k,\ell}$ in which
	$J$ is colored by exactly $q-3$ and $q-2$ colors respectively.
	Then $\{\Omega^a, \Omega^b\}$ is a partition for the set of $q$-colorings of~$\hat{H}_{k,\ell}$; moreover, 
	\[
	|\Omega^a| = \frac{1}{6}\,q!\,Z_3(H)^k \qquad\text{and}\qquad |\Omega^b| = \frac{1}{4} (q-3) \,q!\, (2^\ell-2) 2^k.
	\]
\end{lemma}

Observe that in the colorings from $\Omega^a$, the $H_i$'s are assigned the remaining $3$ colors,
and in those from $\Omega^b$ they are colored with $2$ colors.
We provide the proof of this lemma next.

\begin{proof}[Proof of Lemma~\ref{lem:H-hat-coloring}]
	Observe that $J$ is a complete $(q-3)$-partite graph, so it requires at least $q-3$ colors in every proper $q$-coloring of $\hat{H}_{k,\ell}$.
	Moreover, since each $H_i$ is a connected bipartite graph, it requires at least $2$ colors in every $q$-coloring. Also, every vertex in $H_i$ for $1\leq i\leq k$ is adjacent to every vertex in $J$. Thus, the $H_i$'s do not receive the colors that are used to color $J$. It then follows that
	$\{\Omega^a, \Omega^b\}$ is a partition for the set of $q$-colorings of $\hat{H}_{k,\ell}$.
	
	We count next the number of $q$-colorings of each type. For colorings in $\Omega^a$, there are $q!/3!$ ways to color $J$, and given the colors of $J$, there are $Z_3(H)$ colorings of each $H_i$ that use the remaining $3$ colors. This gives
	\[
	|\Omega^a| = \frac{q!}{3!} \cdot Z_3(H)^k = \frac{1}{6} \,q!\, Z_3(H)^k.
	\]
	For colorings in $\Omega^b$, the complete $(q-3)$-partite graph $J$ receives exactly $q-2$ colors. Hence, there is one cluster of $J$ that is assigned $2$ colors, and every other cluster of $J$ is colored by one color of its own.
	There are $q-3$ ways of selecting the bichromatic cluster, $q!/4!$ choices for the colors of the $q-4$ monochromatic clusters and $\binom{4}{2}$ choices for the colors of the bichromatic cluster. Also, there are $2^\ell-2$ colorings of the bichromatic cluster using exactly $2$ colors. Finally, given the colors of $J$, we have $Z_2(H)=2$ colorings for each $H_i$ using the remaining $2$ colors. Combining these, we get
	\[
	|\Omega^b| = (q-3) \cdot \frac{q!}{4!} \cdot \binom{4}{2} \cdot (2^\ell-2) \cdot Z_2(H)^k = \frac{1}{4} (q-3) \,q!\, (2^\ell-2) 2^k.\tag*{\qedhere}
	\]
\end{proof}

Recall that the phase of a $q$-coloring of a gadget $G(m,q,t)$ is the color of its ports.
Let $\sigma$ be a $q$-coloring of $\hat{H}_{k,\ell}^\Gamma$. The \emph{phase vector} of $\sigma$ is a mapping $\tau: V(\hat{H}_{k,\ell}) \to \{1,\dots,q\}$ defined as follows: for every vertex $v$ of $\hat{H}_{k,\ell}$, $\tau(v)$ is the phase of the coloring $\sigma$ in the gadget $G_v$ for the vertex $v$.
We show next that the phase vector of a $q$-coloring of $\hat{H}_{k,\ell}^\Gamma$ determines a $q$-coloring of $\hat{H}_{k,\ell}$.
\begin{lemma}\label{lem:induced-coloring}
	Let $\sigma$ be a $q$-coloring of $\hat{H}_{k,\ell}^\Gamma$ and $\tau$ be the phase vector of $\sigma$.
	Then, $\tau$ is a $q$-coloring of $\hat{H}_{k,\ell}$.
	Moreover, if $\tau$ is a $q$-coloring of $\hat{H}_{k,\ell}$, then there are $((q-1)!)^{m^2}$ $q$-colorings of $\hat{H}_{k,\ell}^\Gamma$ whose phase vector is $\tau$.
\end{lemma}
\begin{proof}
	In our construction, 
	for every edge $\{u,v\}$ of $\hat{H}_{k,\ell}$ we connect one port of the gadget $G_u$ with one port of $G_v$. Thus, the phase of $G_u$ and the phase of $G_v$ are distinct. This gives $\tau(u) \neq \tau(v)$ for every edge $\{u,v\}$ of $\hat{H}_{k,\ell}$. Hence, $\tau$ is a $q$-coloring of $\hat{H}_{k,\ell}$.
	
	Given the phase vector $\tau$ of a $q$-coloring of $\hat{H}_{k,\ell}^\Gamma$, the number of ways to color each gadget is $((q-1)!)^m$ by Lemma~\ref{lem:gadget-No-coloring}. Since gadgets are connected to each other only by edges between ports, we deduce that given the phase vector $\tau$ (namely, the colors of all the ports in all the gadgets) the number of $q$-colorings of $\hat{H}_{k,\ell}^\Gamma$ is
	\[
	\left[((q-1)!)^{m}\right]^{kN+\ell(q-3)} = ((q-1)!)^{m^2}
	\]
	where we recall that the number of vertices of $\hat{H}_{k,\ell}$ is $kN+\ell(q-3)$ and we set $m=kN+\ell(q-3)$.
\end{proof}

Combining Lemmas \ref{lem:H-hat-coloring} and \ref{lem:induced-coloring}, we can also partition the $q$-colorings of $\hat{H}_{k,\ell}^\Gamma$ into two types.
\begin{lemma}\label{lem:H-hat-Gamma-coloring}
	Let $k,\ell\in\N^+$.
	Let $\Omega^A$ and $\Omega^B$ be the set of $q$-colorings of $\hat{H}_{k,\ell}^\Gamma$
	whose phase vector is a $q$-coloring of $\hat{H}_{k,\ell}$ that belongs to $\Omega^a$ and $\Omega^b$ respectively.
	Then $\{\Omega^A, \Omega^B\}$ is a partition for the set of $q$-colorings of~$\hat{H}_{k,\ell}^\Gamma$; moreover,  
	\begin{align*}
	|\Omega^A| &= |\Omega^a| \cdot [(q-1)!]^{m^2} = \frac{1}{6}\,q!\,Z_3(H)^k ((q-1)!)^{m^2},~\text{and}\\
	|\Omega^B| &= |\Omega^b| \cdot [(q-1)!]^{m^2} = \frac{1}{4} (q-3) \,q!\, (2^\ell-2) 2^k ((q-1)!)^{m^2}.
	\end{align*}
\end{lemma}


\begin{proof}
	Follows immediately from Lemmas~\ref{lem:H-hat-coloring} and \ref{lem:induced-coloring}.
\end{proof}

We are now ready to prove Lemmas~\ref{lem:G-and-Gstar-coloring} and~\ref{lem:sampling-coloring}.

\begin{proof}[Proof of Lemma~\ref{lem:G-and-Gstar-coloring}]
	Let $\Omega_{G}$,
	$\Omega^A_{G}$ and $\Omega^B_{G}$ denote the set of all $q$-colorings, $q$-colorings from $\Omega^A$ and $q$-colorings from $\Omega^B$ of the graph $G=\hat{H}_{k,\ell}^\Gamma$ respectively. Define $\Omega_{G^*}$, $\Omega^A_{G^*}$ and $\Omega^B_{G^*}$ similarly for $G^* = \hat{B}_{k,\ell}^\Gamma$. By Lemma~\ref{lem:H-hat-Gamma-coloring}, we have $|\Omega_{G}| = |\Omega_{G}^A| + |\Omega_{G}^B|$, $|\Omega_{G^*}| = |\Omega_{G^*}^A| + |\Omega_{G^*}^B|$ and $|\Omega_{G}^B| = |\Omega_{G^*}^B|$. Since $H$ is a subgraph of $B$, we deduce that $\hat{H}_{k,\ell}$ is a subgraph of $\hat{B}_{k,\ell}$ and also $G$ is a subgraph of $G^*$ (by selecting the same ports when constructing $G$ and $G^*$). Therefore, $\Omega_{G} \supset \Omega_{G^*}$. It follows that
	\begin{align}
	\TV{\mu_G}{\mu_{G^*}} &= \sum_{\sigma:\; \mu_G(\sigma) > \mu_{G^*}(\sigma)} \mu_G(\sigma) - \mu_{G^*}(\sigma) = \sum_{\sigma \in \Omega_{G}\backslash \Omega_{G^*}} \frac{1}{|\Omega_{G}|} \notag \\
	&= 1 - \frac{|\Omega_{G^*}|}{|\Omega_{G}|} = 1 - \frac{|\Omega_{G^*}^A| + |\Omega_{G^*}^B|}{|\Omega_{G}^A| + |\Omega_{G}^B|} = \frac{|\Omega_{G}^A| - |\Omega_{G^*}^A|}{|\Omega_{G}^A| + |\Omega_{G}^B|} \label{eq:col-tv}.
	\end{align}
	
	If $Z_3(H) < \psi(k,\ell) = (q-3)^{1/k} \; 2^{1+\ell/k}$, then we deduce from Lemma~\ref{lem:H-hat-Gamma-coloring} that
	\begin{align*}
	\TV{\mu_G}{\mu_{G^*}} &\leq \frac{|\Omega_{G}^A|}{|\Omega_{G}^B|}
	= \frac{\frac{1}{6}\,q!\,Z_3(H)^k ((q-1)!)^{m^2}}{\frac{1}{4} (q-3) \,q!\, (2^\ell-2) 2^k ((q-1)!)^{m^2}}
	\le \frac{2 Z_3(H)^k}{3 (q-3) 2^{\ell-1 +k}}
	= \frac{4}{3} \left( \frac{Z_3(H)}{\psi(k,\ell)} \right)^k,
	\end{align*}
	where the second inequality uses the fact that $2^\ell-2 \geq 2^{\ell-1}$ for $\ell\geq 2$. This establishes part 1 of the lemma.
	
	For part 2,
	if $Z_3(H) \geq \psi(k,\ell)$, then by Lemma~\ref{lem:H-hat-Gamma-coloring}
	\[
	|\Omega_G^A| = \frac{1}{6}\,q!\,Z_3(H)^k ((q-1)!)^{m^2} \geq \frac{1}{6} (q-3) \,q!\,2^{\ell+k} ((q-1)!)^{m^2} \geq \frac{2}{3} |\Omega_{G}^B|.
	\]
	We deduce that
	\begin{align*}
	\TV{\mu_G}{\mu_{G^*}} &\geq \frac{|\Omega_{G}^A| - |\Omega_{G^*}^A|}{|\Omega_{G}^A|+\frac{3}{2}|\Omega_{G}^A|} = \frac{2}{5} \left( 1-\frac{|\Omega_{G^*}^A|}{|\Omega_{G}^A|} \right) = \frac{2}{5} \left( 1 - \left( \frac{Z_3(B)}{Z_3(H)} \right)^k \right). \qedhere
	\end{align*}
\end{proof}

\begin{proof}[Proof of Lemma~\ref{lem:sampling-coloring}]
	By Lemma~\ref{lem:induced-coloring}, the number of $q$-colorings of $G^* = \hat{B}_{k,\ell}^\Gamma$ given a phase vector $\tau$ is $((q-1)!)^{m^2}$, which is independent of $\tau$. Thus, the phase vector $\tau$ of a uniformly random $q$-coloring of $G^*$ is a uniformly random $q$-coloring of $\hat{B}_{k,\ell}$.
	Our algorithm for sampling from the distribution $\mu_{G^*}$ then works as follows:
	\begin{enumerate}
		\item Generate a random $q$-coloring $\tau$ of $\hat{B}_{k,\ell}$;
		\item For each $v\in \hat{V} = V(\hat{B}_{k,\ell})$, color all ports of the gadget $G_v$ in $\hat{B}_{k,\ell}^\Gamma$ with $\tau(v)$, and then color all non-ports of $G_v$, which are disjoint cliques of size $q-1$, with a random $(q-1)$-coloring using all colors but $\tau(v)$.
	\end{enumerate}
	
	To generate a $q$-coloring of $\hat{B}_{k,\ell}$ uniformly at random, we can proceed as follows:
	\begin{enumerate}
		\item Compute $|\Omega^a|$, $|\Omega^b|$ and $|\Omega| = |\Omega^a| + |\Omega^b|$;
		\item With probability $|\Omega^a|/|\Omega|$ generate a random $q$-coloring from $\Omega^a$;
		\item With probability $|\Omega^b|/|\Omega|$ generate a random $q$-coloring from $\Omega^b$.
	\end{enumerate}
	
	To compute $|\Omega^a|$ and $|\Omega^b|$, assume $(U,W)$ is the bipartition of the vertex set of the complete bipartite graph $B$. Suppose $|U|=N_1$ and $|W|=N_2$. Then we have
	\[
	Z_3(B) = 3\cdot 2 + 3\cdot (2^{N_1}-2) + 3\cdot (2^{N_2}-2) = 3 \left(2^{N_1} + 2^{N_2} -2\right).
	\]
	Lemma~\ref{lem:H-hat-coloring} implies that
	\[
	|\Omega^a| = \frac{1}{6}\,q!\,3^k \left(2^{N_1} + 2^{N_2} -2\right)^k \qquad\text{and}\qquad |\Omega^b| = \frac{1}{4} (q-3) \,q!\, (2^\ell-2) 2^k.
	\]
	
	To generate a coloring from $\Omega^a$, first choose $q-3$ random colors for the complete $(q-3)$-partite graph $J$
	and randomly assign one of these colors to each cluster of $J$.
	The $k$ copies of $B$ are colored with the remaining $3$ colors. Since $B$ is a complete bipartite graph, it is straightforward to generate a random $3$-coloring in linear time.
	
	In similar manner, to generate a coloring from $\Omega^b$, first choose $q-2$ colors and color $J$ with these $q-2$ colors.
	This can be done by first picking a random cluster of $J$ and coloring it with $2$ different random colors, and then coloring the other $q-4$ clusters with the remaining $q-4$ colors. Finally color the $k$ copies of $B$ with the $2$ colors not used in $J$.
	
	Since each step of the sampling procedure for $\hat{B}_{k,\ell}$ takes at most linear time, the running time of generating a random $q$-coloring of $\hat{B}_{k,\ell}$ is $O(kN+\ell(q-3))$. Therefore, the running time of sampling from $\mu_{G^*}$ is $O(n)$.
\end{proof}


\subsection{Proof of Theorem~\ref{thm:main-q-coloring}: the $q=3$ case}
\label{subsec:proof:col:main:q=3}

In this section we provide the proof of Theorem~\ref{thm:main-q-coloring} for $q=3$.
The proof of this case is very similar to that of $q \ge 4$, but we are required to modify the construction of the testing instance slightly
and rederive the results in Lemmas~\ref{lem:G-and-Gstar-coloring} and \ref{lem:sampling-coloring}.

Let $H=(V,E)$ be a connected bipartite graph on $N$ vertices for which we want to count the number of $3$-colorings.
Recall that for $k,\ell\in\N^+$, we define $\hat{H}_{k,\ell}$ to be the graph that contains $k$ copies of $H$, a complete $(q-3)$-partite graph $J$ with $(q-3)\ell$ vertices, and a complete bipartite graph connecting $J$ and all copies of $H$.
If $J$ is colored by $q-2$ colors, then every copy of $H$ is assigned the remaining $2$ colors;
on the other hand, if $J$ is colored by $q-3$ colors, then the copies of $H$ are colored with the remaining $3$ colors.
By checking which of the two types of $q$-colorings dominates using the \textsc{Tester}, we can obtain a bound on $Z_3(H)$.
This approach works only for $q\geq 4$ as the construction of $\hat{H}_{k,\ell}$ (in particular, the complete $(q-3)$-partite graph $J$) requires $q\geq 4$. 

For $q=3$, we need one additional idea. We construct first a graph $\tilde{H}$ which consists of the original graph $H$, two additional vertices $\{s,t\}$, and several intermediate vertices connecting $H$ and $\{s,t\}$ (see Figure~\ref{fig:H'}).
The graph $\tilde{H}$ is constructed  in a way such that in every $3$-coloring: if $s$ and $t$ receive the same color, then $H$ is colored by exactly two colors; and, if $s$ and $t$ receive two distinct colors, then $H$ can be colored by any proper $3$-coloring with equal probability. The problem then reduces to counting the $3$-colorings of $\tilde{H}$. We can define a graph $\tilde{H}_{k,\ell}$ using the similar construction as $\hat{H}_{k,\ell}$ for $q\geq 4$, but with two modifications: Firstly, we define $J$ to be an independent set instead of a complete $(q-3)$-partite graph; Secondly, we connect every vertex of $J$ with only the vertices $s$'s and $t$'s in all copies of $\tilde{H}$ instead of all vertices. After constructing the testing instance $\tilde{H}_{k,\ell}$ and $\tilde{H}^\Gamma_{k,\ell}$, the proof of Theorem~\ref{thm:main-q-coloring} for $q=3$ follows in the same manner as for $q\geq 4$.


\begin{figure}[t]
	\centering
	\begin{tikzpicture}[scale=0.4]
	\tikzmath{
		\x = 4;
		\y = 4;
		\sr3 = 1.732;
	}
	
	\foreach \i in {0,\y,2*\y} {
		\filldraw
		(0,\i) circle (2pt) node {}
		(\x,\i) circle (2pt) node {}
		(\x+\sr3,\i+1) circle (2pt) node {}
		(\x+\sr3,\i-1) circle (2pt) node {}
		;
		
		\path
		(0,\i) edge (\x,\i)
		(\x,\i) edge (\x+\sr3,\i+1)
		(\x,\i) edge (\x+\sr3,\i-1)
		(\x+\sr3,\i+1) edge (\x+\sr3,\i-1)
		;
	}
	
	\filldraw
	(0,3*\y) circle (2pt) node [label = below: $v$] {}
	(\x,3*\y) circle (2pt) node [label = below: $a_v$] {}
	(\x+\sr3,3*\y+1) circle (2pt) node [label = above: $b_v$] {}
	(\x+\sr3,3*\y-1) circle (2pt) node [label = below: $c_v$] {}
	;
	
	\path
	(0,3*\y) edge (\x,3*\y)
	(\x,3*\y) edge (\x+\sr3,3*\y+1)
	(\x,3*\y) edge (\x+\sr3,3*\y-1)
	(\x+\sr3,3*\y+1) edge (\x+\sr3,3*\y-1)
	;
	
	\filldraw
	(\x+\sr3+\x,\y-1) circle (2pt) node [label = right: $t$] {}
	(\x+\sr3+\x,2*\y+1) circle (2pt) node [label = right: $s$] {}
	;
	
	\path
	(\x+\sr3,-1) edge [bend right = 20] (\x+\sr3+\x,\y-1)
	(\x+\sr3,\y-1) edge [bend right = 10] (\x+\sr3+\x,\y-1)
	(\x+\sr3,2*\y-1) edge [bend left = 20] (\x+\sr3+\x,\y-1)
	(\x+\sr3,3*\y-1) edge [bend left = 30] (\x+\sr3+\x,\y-1)
	
	(\x+\sr3,1) edge [bend right = 30] (\x+\sr3+\x,2*\y+1)
	(\x+\sr3,\y+1) edge [bend right = 20] (\x+\sr3+\x,2*\y+1)
	(\x+\sr3,2*\y+1) edge [bend left = 10] (\x+\sr3+\x,2*\y+1)
	(\x+\sr3,3*\y+1) edge [bend left = 20] (\x+\sr3+\x,2*\y+1)
	;
	
	\draw[dashed]
	(0,1.5*\y) ellipse (50pt and 220pt)
	;
	
	\draw
	(0,1.5*\y) node {$H$}
	;
	\end{tikzpicture}
	\caption{The graph $\tilde{H}$.}
	\label{fig:H'}
\end{figure}
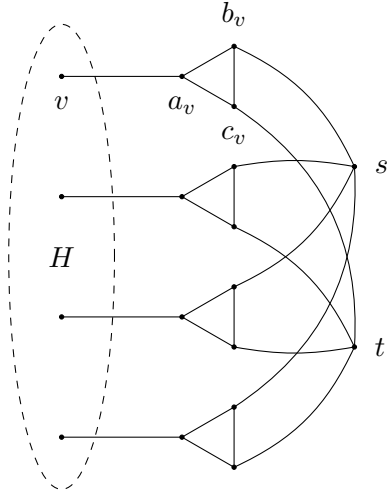

We  define next the graph $\tilde{H}=(\tilde{V},\tilde{E})$, which unlike $H$ is not a bipartite graph.
\begin{enumerate}
	\item Let $s,t$ be two vertices called \emph{interfaces};
	\item For each $v\in V$, let $T_v$ be a triangle on $\{a_v,b_v,c_v\}$ (clique on $3$ vertices);
	\item Set $\tilde{V} = V \cup \left( \bigcup_{v\in V} V(T_v) \right) \cup \{s,t\}$;
	\item Set $\tilde{E} = E \cup \left( \bigcup_{v\in V} E(T_v) \right) \cup \big\{ \{v,a_v\}, \{s,b_v\}, \{t,c_v\}: v\in V \big\};$
\end{enumerate}
see Figure~\ref{fig:H'} for an illustration of the graph $\tilde{H}$.
Observe that $\tilde{H}$ has $\tilde{N} = 4N+2$ vertices.

Let $I(\tilde{H})=\{s,t\}$.
For $k,\ell\in\N^+$, we also define the graph $\tilde{H}_{k,\ell}=(V(\tilde{H}_{k,\ell}),E(\tilde{H}_{k,\ell}))$ as follows:

\begin{enumerate}
	\item Let $\tilde{H}_1,\dots,\tilde{H}_k$ be $k$ copies of the graph $\tilde{H}$;
	\item Let $J$ be an independent set on $\ell$ vertices;
	\item Set $V(\tilde{H}_{k,\ell}) = \left( \bigcup_{i=1}^k V(\tilde{H}_i) \right) \cup V(J)$;
	\item In addition to the edges in $\tilde{H}_i$ for $1\leq i\leq k$, $E(\tilde{H}_{k,\ell})$ also contains edges between the interfaces of $\tilde{H}_i$ for $1\leq i\leq k$ and every vertex in $J$; i.e., for $I(\tilde{H}_i)=\{s_i,t_i\}$ and $v\in J$, we have $\{s_i,v\},\{t_i,v\}\in E(\tilde{H}_{k,\ell})$.
\end{enumerate}

Finally, we define the graph $\tilde{H}_{k,\ell}^{\Gamma}$ where $\Gamma=\{m,3,t\}$ in the same way as for $q\geq 4$; namely, we replace every vertex of $\tilde{H}_{k,\ell}$ by a copy of the graph $G(m,3,t)$ and every edge by an edge between two (unused) ports of the corresponding two gadgets. Furthermore, to make the graph $\tilde{H}_{k,\ell}^{\Gamma}$ well-defined, we set
\[
m = k\tilde{N}+\ell = k(4N+2) + \ell \qquad\text{and}\qquad t = \ceil{ \sqrt{q-\frac{3}{4}} + \frac{1}{2} } =2.
\]

Let $B$ be a complete bipartite graph with the same vertex bipartition as $H$. By setting $H=B$, we also define the graphs $\tilde{B}_{k,\ell}$ and $\tilde{B}_{k,\ell}^{\Gamma}$. Given $k,\ell\in\N^+$, let $G = \tilde{H}_{k,\ell}^{\Gamma}$ and $G^*=\tilde{B}_{k,\ell}^{\Gamma}$. Suppose $d\geq d_c(3)=4$. Then, Lemma~\ref{lem:gadget-degree-coloring} implies that $G,G^* \in \mathcal{M}(n,d)$ for
\[
n=(k\tilde{N}+\ell) \cdot 4m = 4m^2.
\]

The next two lemmas
will
play the role of Lemmas~\ref{lem:G-and-Gstar-coloring} and \ref{lem:sampling-coloring} in the proof of Theorem~\ref{thm:main-q-coloring} for the case $q=3$.

\begin{lemma}
	\label{lem:G-and-Gstar-coloring-q=3}
	Let $k,\ell\in\N^+$ with $\ell\geq 2$. Then the following holds:
	\begin{enumerate}
		\item If $Z_3(H) < 2^{\ell/k}-2$, then
		\[
		\TV{\mu_G}{\mu_{G^*}} \leq 2 \left( \frac{Z_3(H)+2}{2^{\ell/k}} \right)^k.
		\]
		\item If $Z_3(H) \geq 2^{\ell/k}-2$, then
		\[
		\TV{\mu_G}{\mu_{G^*}} \geq \frac{1}{2} \left( 1 - \left( \frac{Z_3(B)+2}{Z_3(H)+2} \right)^k \right).
		\]
	\end{enumerate}
\end{lemma}

\begin{lemma}\label{lem:sampling-coloring-q=3}
	There exists an algorithm with running time $O(n)$ that generates a sample from the distribution $\mu_{G^*}$.
\end{lemma}

With these two lemmas in hand,
the proof of Theorem~\ref{thm:main-q-coloring} for $q=3$ is then identical to that for the $q\geq 4$ case and is thus omitted.

\subsubsection{The colorings of $G$ and $G^*$: proof of Lemmas~\ref{lem:G-and-Gstar-coloring-q=3} and \ref{lem:sampling-coloring-q=3}}

It remains for us to prove Lemmas~\ref{lem:G-and-Gstar-coloring-q=3} and \ref{lem:sampling-coloring-q=3}.
First, we establish several facts about the $3$-colorings of $G =\tilde{H}_{k,\ell}^\Gamma$ and $G^* = \tilde{B}_{k,\ell}^\Gamma$.
We then use these facts as basis to bound $\TV{\mu_G}{\mu_{G^*}}$ (Lemma~\ref{lem:G-and-Gstar-coloring-q=3})
and give a sampling algorithm for $\mu_{G^*}$ (Lemma~\ref{lem:sampling-coloring-q=3}).
Some of these facts are counterparts of those established in Section~\ref{subsec:coloring-facts} for $q \ge 4$.

Let $Z_3(H)$ denote the number of $3$-colorings of $H$.
For $i,j\in\{1,2,3\}$, let $Z_3^{i,j}(\tilde{H})$ be the number of $3$-colorings of $\tilde{H}$ such that $s$ receives color $i$ and $t$ receives color $j$.
\begin{lemma}\label{lem:coloring-th}
	For $i,j\in\{1,2,3\}$,
	$
	Z_3^{i,i}(\tilde{H}) = 2^{N+1}
	$
	and
	$
	Z_3^{i,j}(\tilde{H}) = 2^N Z_3(H)
	$
	when $i \neq j$.
\end{lemma}

\begin{proof}
	We first compute $Z_3^{1,1}(\tilde{H})$. Suppose that both $s$ and $t$ are colored with color $1$. Then, for each $v\in V$,
	colors $2$ and $3$ are both required to color
	$b_v$ and $c_v$. As a result, $a_v$ receives color $1$ and $v$ can not be colored by $1$ for all $v \in V$.
	Hence, the vertices of $H$ in $\tilde{H}$  can only be assigned colors $2$ or $3$.
	There are only two ways to color the connected bipartite graph $H$ with two colors. This gives
	\[
	Z_3^{1,1}(\tilde{H}) = 2\cdot 2^N = 2^{N+1},
	\]
	and by symmetry $Z_3^{2,2}(\tilde{H}) = Z_3^{3,3}(\tilde{H}) = 2^{N+1}$.
	
	We compute next $Z_3^{1,2}(\tilde{H})$. Let $\sigma$ be any $3$-coloring of $H$. We claim that there are $2^N$ colorings of $\tilde{H}$ in which $s$ is assigned color $1$, $t$ is assigned color $2$, and $\sigma$ is the coloring in $H$. From this, it follows immediately that
	\[
	Z_3^{1,2}(\tilde{H}) = 2^N Z_3(H).
	\]
	Let $v\in V$ and consider $3$-colorings of the triangle $T_v$.
	If $\sigma(v)=1$, then the only $3$-colorings of the triangle $(a_v,b_v,c_v)$ are $(2,3,1)$ and $(3,2,1)$ since $\{v,a_v\}$, $\{s,b_v\}$ and $\{t,c_v\}$ are all edges of $\tilde{H}$. 
	Similarly, when $\sigma(v)=2$ or $\sigma(v)=3$, there are also two possible colorings for $(a_v,b_v,c_v)$ in each case.
	Therefore, if $s$ and $t$ are assigned colors $1$ and $2$ respectively, and $\sigma $ is the coloring in $H$, there are exactly two proper $3$-colorings of $T_v$ for each $v\in V$. As the triangles $T_v$ are disjoint for $v\in V$, once the colors of $s$, $t$ and $H$ are assigned, there are $2^N$ proper $3$-colorings of $\tilde{H}$. This proves our claim,
	and by symmetry, for any $i,j\in \{1,2,3\}$ with $i\neq j$, we obtain $Z_3^{i,j}(\tilde{H}) = 2^N Z_3(H)$.
\end{proof}
\noindent
As in Lemma~\ref{lem:H-hat-coloring}, we can partition the $3$-colorings of $\tilde{H}_{k,\ell}$ into two categories.
\begin{lemma}\label{lem:H-prime-hat-coloring}
	Let $k,\ell\in\N^+$.
	Let $\Omega^a$ and $\Omega^b$ be the set of $3$-colorings of $\tilde{H}_{k,\ell}$ in which
	$J$ is colored by exactly $1$ and $2$ colors respectively.
	Then $\{\Omega^a, \Omega^b\}$ is a partition for the set of $3$-colorings of~$\tilde{H}_{k,\ell}$; moreover, 
	\[
	|\Omega^a| = 3 \cdot 2^{k(N+1)} (Z_3(H)+2)^k \qquad\text{and}\qquad |\Omega^b| = 3 \cdot (2^\ell-2) 2^{k(N+1)}.
	\]
\end{lemma}

\begin{proof}
	Observe that in every $3$-coloring of $\tilde{H}_{k,\ell}$ the number of colors we can assign to the independent set $J$ is at least one and at most two, since all the vertices of $J$ have at least one common neighbor. 
	It follows immediately that $\{\Omega^a, \Omega^b\}$ is a partition for the set of $3$-colorings of~$\tilde{H}_{k,\ell}$. 
	For the $3$-colorings in $\Omega^a$, we first assign a color to $J$, say color $1$. Then, we count the number of $3$-colorings of each $\tilde{H}_i$ whose interfaces $\{s_i,t_i\}$ cannot be assigned color $1$. Lemma~\ref{lem:coloring-th} and symmetry imply
	\[
	|\Omega^a| = 3 \left( Z_3^{2,2}(\tilde{H}) + Z_3^{2,3}(\tilde{H}) + Z_3^{3,2}(\tilde{H}) + Z_3^{3,3}(\tilde{H}) \right)^k = 3 \cdot 2^{k(N+1)} (Z_3(H)+2)^k.
	\]
	For $3$-colorings in $\Omega^b$, we pick the two colors that color $J$, say color $1$ and $2$. Then, the interfaces of $\tilde{H}_i$ have to be assigned color $3$ for each $i$. The number of ways to color $J$ with both colors $1$ and $2$ is $2^\ell-2$. Then, by Lemma~\ref{lem:coloring-th} and symmetry, we get
	\[
	|\Omega^b| = 3 \cdot (2^\ell-2) Z_3^{3,3}(\tilde{H})^k = 3 \cdot (2^\ell-2) 2^{k(N+1)}.\tag*{\qedhere}
	\]
\end{proof}

\begin{lemma}\label{lem:induced-coloring-H-prime}
	Let $\sigma$ be a $3$-coloring of $\tilde{H}_{k,\ell}^{\Gamma}$ and $\tau$ be the phase vector of $\sigma$.
	Then, $\tau$ is a $3$-coloring of $\tilde{H}_{k,\ell}$.
	Moreover, if $\tau$ is a $3$-coloring of $\tilde{H}_{k,\ell}$, then there are $2^{m^2}$ $3$-colorings of $\tilde{H}_{k,\ell}^{\Gamma}$ whose phase vector is $\tau$.
\end{lemma}
\begin{proof}
	The proof is analogous to that of Lemma~\ref{lem:induced-coloring}.
\end{proof}

\begin{lemma}
	\label{lem:H-prime-hat-Gamma-coloring}
	Let $k,\ell\in\N^+$.
	Let $\Omega^A$ and $\Omega^B$ be the set of $3$-colorings of $\tilde{H}_{k,\ell}^\Gamma$
	whose phase vector is a $3$-coloring of $\tilde{H}_{k,\ell}$ that belongs to $\Omega^a$ and $\Omega^b$ respectively.
	Then $\{\Omega^A, \Omega^B\}$ is a partition for the set of $3$-colorings of $\tilde{H}_{k,\ell}^\Gamma$; moreover,
	\begin{align*}
	|\Omega^A| &= |\Omega^a| \cdot 2^{m^2} = 3 \cdot 2^{k(N+1)} (Z_3(H)+2)^k \cdot 2^{m^2},~\text{and}\\
	|\Omega^B| &= |\Omega^b| \cdot 2^{m^2} = 3 \cdot (2^\ell-2) 2^{k(N+1)} \cdot 2^{m^2}.\end{align*}
\end{lemma}
\begin{proof}
	Follows immediately from Lemmas~\ref{lem:H-prime-hat-coloring} and \ref{lem:induced-coloring-H-prime}.
\end{proof}

\begin{proof}[Proof of Lemma~\ref{lem:G-and-Gstar-coloring-q=3}]
	We use the notation from the proof of Lemma~\ref{lem:G-and-Gstar-coloring}.
	Following the derivation of \eqref{eq:col-tv}, we get
	\[
	\TV{\mu_G}{\mu_{G^*}} = \frac{|\Omega_G^A| - |\Omega_{G^*}^A|}{|\Omega_{G}^A| + |\Omega_{G}^B|}.
	\]
	If $Z_3(H) < 2^{\ell/k}-2$, then we deduce from Lemma~\ref{lem:H-prime-hat-Gamma-coloring} that
	\begin{align*}
	\TV{\mu_G}{\mu_{G^*}} &\leq \frac{|\Omega_{G}^A|}{|\Omega_{G}^B|} = \frac{3 \cdot 2^{k(N+1)} (Z_3(H)+2)^k \cdot 2^{m^2}}{3 \cdot (2^\ell-2) 2^{k(N+1)} \cdot 2^{m^2}}\\
	&\leq \frac{(Z_3(H)+2)^k}{2^{\ell-1}} = 2 \left( \frac{Z_3(H)+2}{2^{\ell/k}} \right)^k.
	\end{align*}
	If $Z_3(H)\geq 2^{\ell/k}-2$, then by Lemma~\ref{lem:H-prime-hat-Gamma-coloring}
	\[
	|\Omega_{G}^A| = 3 \cdot 2^{k(N+1)} (Z_3(H)+2)^k \cdot 2^{m^2} \geq 3 \cdot 2^{k(N+1)} 2^\ell \cdot 2^{m^2} \geq |\Omega_{G}^B|.
	\]
	Thus, we get
	\[
	\TV{\mu_G}{\mu_{G^*}} \geq \frac{|\Omega_G^A| - |\Omega_{G^*}^A|}{2|\Omega_{G}^A|} = \frac{1}{2} \left( 1-\frac{|\Omega_{G^*}^A|}{|\Omega_{G}^A|} \right) = \frac{1}{2} \left( 1 - \left( \frac{Z_3(B)+2}{Z_3(H)+2} \right)^k \right).\tag*{\qedhere}
	\]
\end{proof}

\begin{proof}[Proof of Lemma~\ref{lem:sampling-coloring-q=3}]
	This can be done in the same way as the proof of Lemma~\ref{lem:sampling-coloring}. It suffices to first generate a random $3$-coloring $\tau$ of $\tilde{B}_{k,\ell}$ and then sample from $\mu_{G^*}$ given $\tau$ as the phase vector where $G^*=\tilde{B}^\Gamma_{k,\ell}$. To sample a random $3$-coloring of $\tilde{B}_{k,\ell}$, we do the following:
	\begin{enumerate}
		\item Compute $|\Omega^a|$, $|\Omega^b|$ and $|\Omega| = |\Omega^a| + |\Omega^b|$;
		\item With probability $|\Omega^a|/|\Omega|$ generate a random $3$-coloring from $\Omega^a$;
		\item With probability $|\Omega^b|/|\Omega|$ generate a random $3$-coloring from $\Omega^b$.
	\end{enumerate}
	We can compute $|\Omega^a|$ and $|\Omega^b|$ by Lemma~\ref{lem:H-prime-hat-coloring}. To sample from $\Omega^a$, we first pick one color for $J$ and then color each copy of $\tilde{B}$. Notice that since $B$ is a complete bipartite graph, we can sample a random $3$-coloring of $B$, and consequently $\tilde{B}$, in linear time. To sample from $\Omega^b$, we pick two colors to color $J$; then in every copy of $\tilde{B}$, the complete bipartite graph $B$ will receive only two colors. The total running time for sampling a random $3$-coloring of $\tilde{B}_{k,\ell}$ is $O(k\tilde{N}+\ell)$ and the running time for sampling from $\mu_{G^*}$ is $O(n)$.
\end{proof}

\section{Discussion}
\label{sec:disc}

Our hardness results for identity testing for the Ising model require $|\beta| d \ge c \log n$ for a suitable constant $c > 0$.
We further assume that $\bu = \beta$; namely, our lower bounds hold even under this additional promise.
 Our proof extends without any significant modification to the case where $\max\{|\beta|,|\bu|\} \cdot d \ge c \log n$. As mentioned, there are polynomial running time algorithms for identity testing when either
 $|\bu| d = O(\log n)$, in which case we can use structure learning methods,
 or when $|\beta| = O(d^{-1})$ is in the tree uniqueness region,
 and known sampling methods can be combined with the testing algorithm in \cite{DDK}.
Therefore, when $\beta$ is the non-uniqueness region ($|\beta| d < c \log n$) and  $|\bu| d = \omega(\log n)$, the computational complexity of identity testing is open, as there is no known polynomial running time algorithm, and our lower bound does not apply to this regime of parameters.

\bibliographystyle{plain}

\end{document}